\documentclass[a4paper, 11pt]{article}
\usepackage{amsthm,amsmath,amssymb}
\usepackage{subcaption}
\usepackage[usenames,dvipsnames]{xcolor}
\usepackage[pdftex,breaklinks,colorlinks,
    citecolor={BlueViolet}, linkcolor={Blue},urlcolor=Maroon]{hyperref}
\usepackage{tikz,pgfkeys}
\usepackage{tkz-graph}
\usetikzlibrary{arrows.meta}
\usetikzlibrary{quotes}
\usetikzlibrary{decorations.pathreplacing}
\usepackage{graphicx}
\usepackage{charter,eulervm}%
\usepackage{multirow,booktabs,array}
\usepackage{tabularx,enumerate}
\usetikzlibrary{calc}
\usepackage[final,expansion=alltext,protrusion=true]{microtype}
\usepackage{tikz}

\newenvironment{tightcenter}
 {\parskip=0pt\par\nopagebreak\centering}
 {\par\noindent\ignorespacesafterend}

\usepackage{xargs}
\usepackage{xifthen}
\usepackage{framed}

\usepackage{ctable}

\newlength{\RoundedBoxWidth}
\newsavebox{\GrayRoundedBox}
\newenvironment{GrayBox}[1]%
   {\setlength{\RoundedBoxWidth}{.75\textwidth}
    \def\boxheading{#1}
    \begin{lrbox}{\GrayRoundedBox}
       \begin{minipage}{\RoundedBoxWidth}%
   }{%
       \end{minipage}
    \end{lrbox}%
    \begin{tightcenter}%
    \begin{tikzpicture}%
       \node(Text)[draw=black!20,fill=white,rounded corners,%
             inner sep=2ex,text width=\RoundedBoxWidth]%
             {\usebox{\GrayRoundedBox}};
        \coordinate(x) at (current bounding box.north west);
        \node [draw=white,rectangle,inner sep=3pt,anchor=north west,fill=white] 
        at ($(x)+(6pt,.75em)$) {\boxheading};
    \end{tikzpicture}
    \end{tightcenter}\vspace{0pt}%
    \ignorespacesafterend
}    

\newenvironment{problem}[2][]{\noindent\ignorespaces%
                                \FrameSep=6pt%
                                \parindent=0pt%
                \vspace*{-.5em}
                \ifthenelse{\isempty{#1}}{%
                  \begin{GrayBox}{#2}%
                }{%
                  \begin{GrayBox}{#2 parameterized by~{#1}}%
                }
                \newcommand\Prob{Output:}%
                \newcommand\Input{Input:}%
                \begin{tabular*}{\textwidth}{@{\hspace{.1em}} >{\itshape} p{1.2cm} p{0.85\textwidth} @{}}%
            }{
                \end{tabular*}%
                \end{GrayBox}%
                \vspace*{-.5em}
                \ignorespacesafterend
            }   

\theoremstyle{plain} 
\newtheorem{theorem}{Theorem}[section]
\newtheorem{lemma}[theorem]{Lemma}
\newtheorem{corollary}[theorem]{Corollary}
\newtheorem{proposition}[theorem]{Proposition}

\tikzstyle{filled vertex}  = [{circle,blue,draw,fill=black!50,inner sep=1pt}]  
\tikzstyle{uvertex} = [{violet, draw, fill=violet!50,inner sep=2pt}]  

\newcommand{\stpath}[2]{$#1$-$#2$ path}
\newcommand{\stsep}[2]{$#1$-$#2$ separator}
\newcommand{\comment}[1]{\textbackslash\!\!\textbackslash {\em #1}}

\title{Enumerating Maximal Induced Subgraphs}
\author{Yixin Cao\thanks{Department of Computing, Hong Kong Polytechnic University, Hong Kong, China. \href{mailto:yixin.cao@polyu.edu.hk} {\tt yixin.cao@polyu.edu.hk}.} }
\date{}

\begin{document}
\maketitle

\begin{abstract}
  Given a graph $G$, the maximal induced subgraphs problem asks to enumerate all maximal induced subgraphs of $G$ that belong to a certain hereditary graph class.  While its optimization version, known as the minimum vertex deletion problem in literature, has been intensively studied, enumeration algorithms are known for a few simple graph classes, e.g., independent sets, cliques, and forests, until very recently [Conte and Uno, STOC 2019].  There is also a connected variation of this problem, where one is concerned with only those induced subgraphs that are connected.  We introduce two new approaches, which enable us to develop algorithms that solve both variations for a number of important graph classes.  A general technique that has been proved very powerful in enumeration algorithms is to build a solution map, i.e., a multiple digraph on all the solutions of the problem, and the key of this approach is to make the solution map strongly connected, so that a simple traversal of the solution map solves the problem.  We introduce retaliation-free paths to certificate strong connectedness of the solution map we build.  Generalizing the idea of Cohen, Kimelfeld, and Sagiv [JCSS 2008], we introduce the $t$-restricted version, $t$ being a positive integer, of the maximal (connected) induced subgraphs problem, and show that it is equivalent to the original problem in terms of solvability in incremental polynomial time.  Moreover, we give reductions between the two variations, so that it suffices to solve one of the variations for each class we study.  Our work also leads to direct and simpler proofs of several important known results.
\end{abstract}

\section{Introduction}\label{sec:intro}

A vertex deletion problem asks to transform an input graph to a graph in a certain graph class $\mathcal{P}$ by deleting vertices.
Many classic optimization problems belong to the family of vertex deletion problems, and their algorithms and complexity have been intensively studied.
More often than not, the graph class is \emph{hereditary}, i.e., closed under taking  induced subgraphs.  Examples include complete graphs, edgeless graphs, acyclic graphs, bipartite graphs, planar graphs, and perfect graphs.
  For a  {hereditary} graph class, this problem is either NP-hard or trivial~\cite{lewis-80-node-deletion-np}.
An equivalent formulation of a vertex deletion problem toward $\mathcal{P}$ is to ask for a maximum induced subgraph that belongs to $\mathcal{P}$.
A plethora of algorithms, including approximation algorithms~\cite{lund-93-approximation-maximum-subgraph, cao-15-interval-deletion, cao-17-unit-interval-editing, cao-17-cluster-vertex-deletion, jansen-16-approximation-and-kernelization-chordal-deletion, agrawal-19-chordal-deletion}, exact algorithms~\cite{fomin-15-large-induced-subgraphs, fomin-19-exact-via-monotone-local-search}, and parameterized ones~\cite{cai-96-hereditary-graph-modification, cao-16-chordal-editing, cao-15-interval-deletion, cao-17-unit-interval-editing, jansen-16-approximation-and-kernelization-chordal-deletion, agrawal-19-chordal-deletion, agrawal-19-kernel-interval-vertex-deletion}, have been proposed for both formulations.

Yet another approach toward this problem is to enumerate, or generate or list, inclusion-wise maxim\emph{al} induced subgraphs that belong to the graph class $\mathcal{P}$, hence called the \emph{maximal induced $\mathcal{P}$ subgraphs} problem.
Trivially, one can always use an enumeration algorithm to solve the optimization version of the same problem.  A classical example of nontrivial use is to solve the coloring problem by enumerating maximal independent sets of the input graph \cite{lawler-76-chromatic,eppstein-03-coloring,byskov-04-coloring,bjorklund-09-set-partitioning-via-inclusion-exclusion}.  Indeed, the enumeration of maximal independent sets and the enumeration of maximal cliques, practically the same, are both well-studied classic problems \cite{akkoyunlu-73-enumeration-cliques, chiba-85-arboricity-and-subgraph-listing, makino-04-enumerating-maximal-cliques}.  They have also been used in solving the connected vertex cover problem and the edge dominating set problem \cite{fernau-06-edge-dominating-set}.
  In general, the {maximal induced $\mathcal{P}$ subgraphs} problem is motivated by the intrinsic difficulty in formulating a combinatorial optimization problem: There are always factors that are ill characterized or even omitted in the formulation \cite{eppstein-98-k-shortest-paths, marino-15-analysis-enumeration}.\footnote{After all, how many times you click ``I'm Feeling Lucky'' on Google.com?}
  This is particularly the case for the vertex deletion problems, of which the primary applications are the processing of noisy data, where modifications to a graph are meant to exclude outliers or to fix noise.
  
  Motivations of {maximal induced $\mathcal{P}$ subgraphs} problems also come from database theory \cite{cohen-08-all-maximal-induced-subgraphs}, where one is usually concerned with only induced $\mathcal{P}$ subgraphs that are connected.  Of any hereditary graph class $\mathcal{P}$, we can define a subclass by allowing only connected graphs in $\mathcal{P}$.  With very few exceptions, this naturally defined class is not hereditary, and hence this variation, called the \emph{maximal connected induced $\mathcal{P}$ subgraphs} problem, poses different challenges.  Indeed, as we will see, it is somewhat more difficult than the original one.

  Let $n$ denote the number of vertices in the input graph.  Since there might be an exponential number (on $n$) of \emph{solutions}, maximal vertex sets inducing subgraphs in $\mathcal{P}$, care needs to be taken when we talk about the running time of an enumeration algorithm.
For example, in a graph consisting of disjoint triangles, there are $3^{n/3}$ maximal independent sets.
 Johnson et al.~\cite{johnson-88-generating-maximal-independent-sets} defined three complexity classes for enumeration algorithms, namely, polynomial total time (polynomial on $n$ and the total number of solutions), incremental polynomial time (for all $s$, the time to output the first $s$ solutions is polynomial on $n$ and $s$), and polynomial delay (for all $s$, the time to output the first $s$ solutions is polynomial on $n$ and linear on $s$).
 Both maximal independent sets and maximal cliques can be enumerated with polynomial delay, and so are maximal bicliques (complete bipartite graphs) \cite{dias-05-generating-bicliques, gely-09-maximal-cliques-bicliques} and maximal forests \cite{schwikowski-02-enumerate-fvs}.
 See the survey~\cite{wasa-16-enumeration} and the recent results~\cite{conte-19-polynomial-delay}.

 Our algorithms are summarized below.  Since a connected cluster graph is a clique, the result for the maximal connected induced cluster subgraphs problem is already known.  Also,
 for a graph class that can be characterized by a finite set of forbidden induced subgraphs, algorithms for  the maximal induced $\mathcal{P}$ subgraphs problem, but not its connected variation, can be derived from Eiter and Gottlob \cite{eiter-95-hypergraph-transversals}.

\begin{theorem}\label{thm:all-classes}
  The maximal induced $\mathcal{P}$ subgraphs problem and its connected variation can be solved with polynomial delay for the following graph classes: interval graphs, trivially perfect graphs, split graphs, complete split graphs, pseudo-split graphs, threshold graphs, cluster graphs, complete bipartite graphs, complete $p$-partite graphs, and $d$-degree-bounded graphs.
  
  The maximal induced $\mathcal{P}$ subgraphs problem and its connected variation can be solved in incremental polynomial time for the following graph classes: wheel-free graphs, unit interval graphs, block graphs, 3-leaf powers, basic 4-leaf powers, and any graph class that can be characterized by a finite set of forbidden induced subgraphs. 
\end{theorem}

\paragraph{Solution maps and retaliation-free paths.}

In a seminal work, Schwikowski and Speckenmeyer~\cite{schwikowski-02-enumerate-fvs} proposed an algorithm for enumerating maximal induced forests, as well as its directed variation, enumerating maximal induced directed acyclic subgraphs.
They introduced a successor function, which, given a solution, i.e., a maximal induced forest of the input graph $G$, returns a set of other solutions of $G$.  In so doing they \emph{implicitly} built a multiple digraph $M(G)$ whose nodes are the solutions, and there is an arc from node $S_1$ to node $S_2$ if $S_2$ is one of the successors of $S_1$.
The key properties of $M(G)$ are
(1) the successor function can be calculated in polynomial time; and (2) $M(G)$ is strongly connected.
As a result, traversing $M(G)$ from any node, we can visit all the solutions in time polynomial on $n$ and linear on the number of solutions.  With standard bookkeeping mechanism, we can easily implement it with polynomial delay.  We call the multiple digraph on the solutions of an enumeration problem its \emph{solution map}.
G{\'{e}}ly et al.~\cite{gely-09-maximal-cliques-bicliques} later rediscovered this idea, and used it to re-analyze enumeration algorithms for maximal cliques and maximal bicliques.  Solution maps are actually very general and powerful.  Conte and Uno~\cite{conte-19-polynomial-delay} built solution maps to solve the maximal (connected) induced $\mathcal{P}$ subgraphs problem for several important graph classes.  In particular, they solved the connected variation for forests, and directed acyclic graphs.  (They also designed algorithms for enumerating maximal subgraphs, a direction that we will not pursue in the present paper, and other non-graphic problems.)

To solve an enumeration problem with a solution map consists in defining the successor function and proving that the implied solution map $M$ is strongly connected.  All the mentioned algorithms follow the same general scheme, although the details are quite problem specific.  For convenience, we may assume that the solutions are subsets of some ground set $U$; note that this is the case for the maximal (connected) induced $\mathcal{P}$ subgraphs problem, of which the ground set is the vertex set of the input graph.
\begin{description}
\item [Successors of a solution $S$:] For each $v\in U\setminus S$, define a sub-instance restricted to $S\cup \{v\}$, and find a set of solutions of this sub-instance.  The union of these $|U\setminus S|$ sets makes the successors of $S$.
\item [Strong connectedness:] For each solution $S^*$, define a specific metric and show that every other solution $S$ has a successor ``closer'' to $S^*$ than $S$ with respect to this metric.  
\end{description}
Two simple metrics have been devised and are very handy to use.  The first is the number of elements of $S^*$ absent in a solution, and the second is the ordered sequence of elements in a solution, with elements in $U$ numbered in a way that those in $S^*$ are the smallest.
In both metrics, $S^*$ is the smallest among all solutions, $0$ for the first metric and $\langle 1, 2, \ldots, |S^*|\rangle$ for the second, and hence it suffices to show that each solution $S$ has a successor with a smaller metric than $S$.
The existence of a path from $S$ to $S^*$ in $M$, hence the strong connectedness of $M$, will then follow from the finiteness of the ground set $U$.  As a matter of fact, it follows that the distance from any solution to another solution in $M$ designed as such is less than $|U|$.

 This recipe is so nice as long as we can find a good successor function.  However, the requirement that every other solution have a \emph{direct} successor that is closer to $S^*$, with respect to the metric decided by $S^*$, is too strong.  For some problems, it seems difficult to construct a successor function with this property.  To show strong connectedness of a solution map, after all, what is crucial is the reachability: We are concerned with whether a solution can reach $S^*$ instead of how long it takes to do so.  (Indeed, in the most ideal case, the solution map can be a simple directed cycle, in which the distance of a pair of solutions can be the number of solutions minus one.)

We introduce retaliation-free paths to certificate strong connectedness of a solution map.  We explain it with the second metric.  We allow all the successors of a solution $S$ to be lexicographic larger than $S$.  Let $s$ be the smallest element in $S^*\setminus S$.  If any successor $S'$ of $S$ contains $[s]$, meaning $\{1, \ldots, s\}$, then $S'$ is lexicographically smaller than $S$ and we are done.  Otherwise, we look for a solution that contains $[s]$ and is reachable from $S$ by a nontrivial path.  It is better to view this path as a transforming procedure.  In the first step, we choose a successor $S'$ of $S$ containing $s$; we call $s$ the \emph{gainer}, and elements in $[s-1]\setminus S'$ the \emph{victims} of $s$ (at this step).  We then try to add the victims of $s$ back, with the guarantee that $s$ is never removed during the whole procedure.
A victim $r$ of $s$ may become a gainer in a later step, hence introducing further victims.  In this case, we also ensure that $r$ is kept before all its victims are restored.  Such a path is thus called \emph{retaliation-free}.
(After we reach a solution $S_1$ containing $[s]$, we set for a solution containing $[s+1]$, in the new pursuit the element $s$ may become a victim some point.  The path from $S$ to $S^*$ we produced as such can be insanely long.) 

\begin{theorem}[Informal]
  \label{thm:retaliation-free}
  Let $M$ be the solution map of an enumeration problem.  If for any solution $S^*$, there is a metric on the solutions such that any solution $S$ can reach, via a retaliation-free path, a solution that is closer to $S^*$ than $S$ in this metric, then $M$ is strongly connected.
\end{theorem}

We use Theorem~\ref{thm:retaliation-free} to develop enumeration algorithms for the maximal (connected) induced trivially perfect subgraphs problem and the maximal (connected) induced interval subgraphs problem, both running with polynomial delay.
The two problems are paradigmatic for the use of this technique.  The  main structures we use for interval graphs are the clique paths, which are linear, while for trivially perfect graphs, we work on their generating forests, which are hierarchical. 

A trivially perfect graph $G$ can be represented as a forest $F$, called its {generating forest}, such that two vertices $u$ and $v$ are adjacent in $G$ if and only if one of them is an ancestor of the other in $F$ \cite{wolk-62, yan-96-trivially-perfect}.  It is simpler if we consider the connected variation, where generating forests are trees.  For the purpose of transforming a solution $S$ to another solution $S^*$, we transform a generating tree $T$ of $G[S]$ to a generating tree $T^*$ of $G[S^*]$.
We may assume without loss of generality that the root of $T^*$ is already in $S$, though it may or may not be the root of $T$.  The $T^*$-ancestors of a vertex always induce a clique, and those in $T$ are lie on a path from the root, while it is hard to make any claim on the $T^*$-descendants of any vertex.  Therefore, it may seem quite obvious that we should proceed top-down.  This idea however does not work, because two vertices in $S\cap S^*$ may have different ancestor-descendant relationship in $T$ and $T^*$.  Consider, for example, $G$ being a path on four vertices, and the two common vertices in the only two solutions of $G$.  We add the vertices in $S^*\setminus S$ bottom-up, i.e., children before parents.
Let $v$ be the smallest vertex in $S^*\setminus S$ in the postorder in $T^*$.  It is safe to add $v$ if $v$ is a leaf of $T^*$, because all its ancestors have larger numbers than $v$, or if $v$ is a parent with a single child $v'$ in $T^*$, because $v'$ is in $S$ by the ordering and $v$ and $v'$ are true twins in $G[S^*]$.  The difficulty poses itself when $v$ has multiple children in $T^*$.  By the ordering, all $T^*$-descendants of $v$ are in $T$.  However, with a polynomial number of successors of $S$, it is difficult, if possible at all, to ensure that in one of the successors, precisely those  $T^*$-children of $v$ are kept as descendants of $v$.
The key observation that enables us to overcome this obstacle is that as long as we keep two children of $v$ in $T^*$, say $c_1$ and $c_2$, then the relationship between $T^*$-descendants of $v$ and $T^*$-ancestors of $v$ will be correctly maintained.  We afford to lose other children of $v$ and their descendants in this step, because they can be added back easily.  Let $S'$ be the new solution after adding $v$, then the $T^*$-descendants of $v$ that are not in $S'$ are the victims of the gainer $v$.  To add a victim $x$ of $v$ to a generating tree $T'$ for $G[S']$, we find a common ancestor of $c_1$ and $c_2$ in $T'$, which is not necessarily $v$ itself because a $T^*$-ancestor of $v$ may be a $T'$-descendant of $v$, and add $x$ under it.  We can make sure at least one successor gives this desired solution.  This ends with a retaliation-free path from $S$ to $S^*$.

Interval graphs are a very important and well-studied graph class \cite{booth-76-pq-tree, fulkerson-65-interval-graphs, hsu-95-recognition-cag, hsu-99-recognizing-interval-graphs}, and the interval vertex deletion problem receives a lot of attention \cite{bliznets-16-max-chordal-interval-subgraphs, cao-15-interval-deletion, cao-16-almost-interval-recognition, agrawal-19-kernel-interval-vertex-deletion}.  It is known that the maximal cliques of an interval graph can be arranged in a linear manner, called a clique path \cite{fulkerson-65-interval-graphs}.  Let $G$ be the input graph, and $S$ a solution of $G$.  We build a clique path $\mathcal{K}$ for $G[S]$, and for each vertex $v\in V(G)\setminus S$, we introduce $O(n^2)$ successors, each corresponding to an range $[a, b]$ in $\mathcal{K}$.  We try to put $v$ in between the $a$th and the $b$th maximal cliques on this clique path to make a new solution, 
 so we need to delete non-neighbors of $v$ in between and delete neighbors of $v$ not in this range.  After that, this subset is extended to a solution, which is made a successor of $S$.
To show the existence of a path from $S$ to $S^*$, we number the vertices in $S^*$ from left to right with respect to a fixed clique path $\mathcal{K}^*$ for $G[S^*]$.  The convenience of this order is that neighbors of $s$, the first vertex in $S^*\setminus S$, in $[s-1]$ form a clique.
If vertices in $[s-1]$, a subset of both $S$ and $S^*$, are arranged in the same way in  $\mathcal{K}^*$ and the clique path for $G[S]$, then $s$ can be easily added after some maximal clique $G[S]$ containing $N(s)\cap [s-1]$.  However, an interval graph may have many different clique paths, and this is exactly the main obstacle the recognition algorithms need to overcome \cite{booth-76-pq-tree, hsu-99-recognizing-interval-graphs}.  Thus, adding $s$ may induce some victims. Our observation is that all the victims are in a small range in terms of the clique path $\mathcal{K}^*$, and all the changes we need to make to restore these victims are also confined in this range.  What we show is therefore a retaliation-free path from $S$ to a solution containing $[s]$, implying that the solution map is strongly connected.

In all the mentioned algorithms, the successor function has to be hand crafted for each graph class.
For example, our successor function for interval graph does not work for unit interval graphs.  Nor can the successor function of \cite{conte-19-polynomial-delay} for chordal graphs be applied to interval graphs.  In the following we will aim for approaches that are more general and systematic.

\paragraph{Algorithms in Incremental polynomial time.}
There is nothing inherent in solution maps about polynomial delay, though all of the mentioned algorithms based on solution maps are of this type.  The running time of such an algorithm is determined by the successor function, which is almost always determined by the number of successors a solution can have.  If we allow the number of successors to be a polynomial on both $n$ and the number of solutions, then the resulted enumeration algorithm runs in polynomial total time.
One ``weakness'' of the approach based on solution maps is the lack of a simple way to develop algorithms in incremental polynomial time.  We obviate this concern with the following observation.
\begin{theorem} \label{thm:connected-incp=totalp}
  For any hereditary graph class $\mathcal{P}$, the maximal connected induced $\mathcal{P}$ subgraphs problem can be solved in polynomial total time if and only if it can be solved in incremental polynomial time.
\end{theorem}
For the maximal induced $\mathcal{P}$ subgraphs problem, a similar statement as Theorem~\ref{thm:connected-incp=totalp} follows from the classical result of Bioch and Ibaraki~\cite{bioch-95-identification-dualization-positive-boolean-functions}.  But the class of connected $\mathcal{P}$ graphs is not hereditary in general, and hence Theorem~\ref{thm:connected-incp=totalp} cannot be derived from \cite{bioch-95-identification-dualization-positive-boolean-functions}.  Theorem~\ref{thm:connected-incp=totalp} and the result in \cite{bioch-95-identification-dualization-positive-boolean-functions} enable us to use solution maps transparently in developing algorithms for enumerating maximal (connected) induced $\mathcal{P}$ subgraphs in incremental polynomial time.
We remark that a similar claim on general enumeration problems is very unlikely \cite{strozecki-10-thesis, capelli-19-incremental-delay-enumeration}.
\paragraph{Restricted versions.}  In the design of successor functions, a special version of the maximal (connected) induced $\mathcal{P}$ subgraphs problem presents itself, where the input graph $G$ has a special vertex $v$ such that $G - v$ is in $\mathcal{P}$.  It also arises naturally in other approaches for enumerating maximal (connected) induced $\mathcal{P}$ subgraphs.  Cohen et al.~\cite{cohen-08-all-maximal-induced-subgraphs} called it the \emph{input-restricted version} and conducted a systematic study.  They managed to show that this ostensibly simpler version is the core of the original problem in terms of solvability in incremental polynomial time: The original problem can be solved in incremental polynomial time if and only if its input-restricted version can.  They also proved a similar statement for polynomial total time, which is rendered redundant by Theorem~\ref{thm:connected-incp=totalp}.  Moreover, the original problem can be solved with polynomial delay when its input-restricted version can be solved in polynomial time.

Since the appearance of \cite{cohen-08-all-maximal-induced-subgraphs}, there have been attempts at extending its core idea.
A ``natural'' way
seems to be defining a restricted version that is equipped with a special set of $t, t > 1$, vertices instead of a single vertex $v$; i.e., $G - Z$ is in $\mathcal{P}$ for some set $Z$ of $t$ vertices.  However, the extra $t - 1$ vertices in the set $Z$ turn out to be not helpful.  Consider, for instance, the forests.  For a graph $G$ and any $t$, we can make a new graph $H$ by introducing $t - 1$ copies of disjoint triangles to $G$, and then $H$ satisfies the new definition if and only if $G$ is input-restricted.  Therefore, solving this special version is not easier than solving the input-restricted version.

This bad example does not suggest a dead end, and there is a natural generation that does work.  We may view the input-restricted version as an enumeration problem by itself, and try to build a solution map to solve it.  Then in designing the successor function, in the sub-instance restricted to $S\cup \{v'\}$, aside from the vertex $v$ we had, we are bestowed with another vertex $v'$ such that the removal of either of $v$ and $v'$ leaves a subgraph in $\mathcal{P}$.  We are thus inspired to define the \emph{$t$-restricted version} of the maximal induced $\mathcal{P}$ subgraphs problem:
\begin{quote}
  There exists a set $Z$ of $t$ vertices in $G$ such that $G - z$ is in $\mathcal{P}$ for \emph{every} $z\in Z$.
\end{quote}
It is equivalent to that \emph{every} minimal induced subgraph of $G$ that is not in $\mathcal{P}$ contains \emph{all} vertices in $Z$.
This requirement is far stronger than $G - Z$ being in $\mathcal{P}$, though equivalent when $t = 1$.  
We are able to show that even this far more restricted version is still equivalent to the original problem in terms of enumerability in incremental polynomial time.
We remark that our proofs, based on solution maps, are significantly simpler than those on the input-restricted version \cite{cohen-08-all-maximal-induced-subgraphs}, which can now be viewed as $1$-restricted version.  With the benefit of hindsight, we can see that all the results of \cite{cohen-08-all-maximal-induced-subgraphs} can be easily interpreted using solution maps, with simpler proofs.

\begin{theorem}\label{thm:poly-total}
  Let $\mathcal{P}$ be a hereditary graph class.
  The maximal (connected) induced $\mathcal{P}$ subgraphs problem can be solved in incremental polynomial time if and only if there exists a positive integer $t$ such that its $t$-restricted version can be solved in polynomial total time.
\end{theorem}

The strong requirement stipulated in defining the $t$-restricted version makes it significantly easier than the original problem.  All the algorithms in incremental polynomial time stated in Theorem~\ref{thm:all-classes} are obtained by reducing these problems to their $t$-restricted versions.  
Of these results, we would like to draw special attention to those on wheel-free graphs, though this graph class in its own sense may seem to be less interesting compared to others we study.
Lokshtanov~\cite{lokshtanov-08-wheel-free-deletion} proved that the vertex deletion problem toward wheel-free graphs is W[2]-hard with respect to standard parameterization, hence very unlikely fixed-parameter tractable.  This suggests that the complexity of the maximal induced $\mathcal{P}$ subgraphs problem can be quite different from its optimization counterpart.

Theorem~\ref{thm:poly-total} also implies, among others, the following result on graph classes that can be characterized by a finite set $\mathcal{F}$ of forbidden induced subgraphs.  Indeed, with $t$ being the maximum order of graphs in $\mathcal{F}$, the $t$-restricted version of the maximal (connected) induced $\mathcal{F}$-free subgraphs problem is trivial.
Eiter and Gottlob \cite{eiter-95-hypergraph-transversals} have proved this result for the maximal induced $\mathcal{F}$-free subgraphs problem, but their proof does not applies to the connected variation; see the discussion in the end of this section.
\begin{corollary}
  \label{thm:fintie}
  Let $\mathcal{F}$ be a finite set of graphs.  The maximal (connected) induced $\mathcal{F}$-free subgraphs problem can be solved in incremental polynomial time.  
\end{corollary}
 We note that the vertex deletion problem to $\mathcal{F}$-free graphs for finite $\mathcal{F}$ has been well studied.  In particular, they are fixed-parameter tractable \cite{cai-96-hereditary-graph-modification} and admit constant-approximation \cite{lund-93-approximation-maximum-subgraph}.

 \paragraph{Reductions between  the two variations.}
Since we are dealing with both the maximal induced $\mathcal{P}$ subgraphs problem and its connected variation, it is really irksome if we have to develop two algorithms for each class.  Although Cohen et al.~\cite{cohen-08-all-maximal-induced-subgraphs} and Conte and Uno~\cite{conte-19-polynomial-delay} dealt with both variations, they stopped at noting that with one of them solved, a slight modification would be able to solve the other.  These modifications have to be done, however, case by case.

The main issue of the connected variation is that the set of connected graphs in a hereditary graph class $\mathcal{P}$, viewed as a graph class in its own regard, is mostly not hereditary.  
If there are two nonadjacent vertices in any connected graph $G$ in $\mathcal{P}$, then the edgeless graph on two vertices is an induced subgraph of $G$, hence in $\mathcal{P}$.  Therefore, the class of connected graphs in $\mathcal{P}$ remains hereditary only when $\mathcal{P}$ is the class of cluster graphs, in which the connected ones are the complete graphs, and the class of edgeless graphs.
A connected edgeless graphs has precisely one vertex, and this seems to be the only class on which the connected variation, which is trivial, is the easier between the two variations.
For a hereditary graph class $\mathcal{P}$, and a graph $G$, let us use $N_1$ and $N_2$ to denote, respectively, the number of maximal induced $\mathcal{P}$ subgraphs and the number of maximal connected induced $\mathcal{P}$ subgraphs of $G$.  It is not difficult to see that $N_1/N_2$ can be an exponential number on $n$, while $N_2/N_1$ can never be more than $n$.  In the trivial case when $\mathcal{P}$ is edgeless, $N_2$ is precisely $n$, while $N_1$ can be $3^{n/3}$.
For a nontrivial example, the graph consisting of $n/4$ disjoint 4-cycles has $2^{n/2}$ maximal induced forests, while only $n$ maximal induced trees.
This example applies to interval graphs, chordal graphs, etc.
Therefore, it is very unlikely we can use an enumeration algorithm for the maximal induced $\mathcal{P}$ subgraphs problem to solve its connected variation.

The other direction is promising.  Two simple observations help here.  First, if we can add vertices to a graph without making new forbidden induced subgraphs, then it does not change the number of maximal induced $\mathcal{P}$ subgraphs, though each of them gets more vertices.  Second, if the extra vertices make all the maximal induced $\mathcal{P}$ subgraphs connected, then by enumerating all maximal connected induced $\mathcal{P}$ subgraphs of the new graph, we obtain effortlessly all maximal induced $\mathcal{P}$ subgraphs of the original graph.
This idea works when $\mathcal{P}$ is closed under adding universal vertices, examples of which include interval graphs, chordal graphs, etc.
A slighter nontrivial reduction can be devised for several other graph classes, e.g., wheel-free graphs and triangle-free graphs.  These reductions focus us on the connected variation of the problems only, instead of working on both.

The more interesting connection we are able to establish, via Theorem~\ref{thm:poly-total}, is the following theorem.
Although the condition in the statement looks rather technical, it is satisfied by many natural graph classes, e.g., planar graphs, apart from those we study in Section~\ref{sec:incremental-poly}.
As for the $t$-restricted version of both variations, there are only a polynomial number solutions that are not a proper superset of $Z$, and they are easy to handle.  Therefore, the focus is on those solutions containing all vertices in $Z$.  We observe that under the connectivity condition, in every maximal induced $\mathcal{P}$ subgraph of $G$ that contains all vertices in $Z$, all the vertices in $Z$ are always in the same component.  As a result, there is a one-to-one mapping between such solutions for these two variations.
  
\begin{theorem}\label{thm:equivalence}
  Let $\mathcal{F}$ be a set of graphs such that every graph in $\mathcal{F}$ of order $c$ or above is biconnected.  The maximal induced $\mathcal{F}$-free subgraphs problem can be solved in incremental polynomial time  if and only if its  connected variation can be solved in incremental polynomial time.
\end{theorem}

\paragraph{Characterizing the easy classes.}  We try to better understand maximal (connected) induced $\mathcal{P}$ subgraphs problems that can be solved with polynomial delay by considering its input-restricted version.
We say that a graph class $\mathcal{P}$ has the \emph{\textsc{cks} property}, after the authors of \cite{cohen-08-all-maximal-induced-subgraphs}, if the input-restricted version of the maximal (connected) induced $\mathcal{P}$ subgraphs problem 
can be solved in time polynomial on only $n$.
Since it is straightforward to see that the class of edgeless graphs (independent sets), the class of complete graphs (cliques), and the class of complete bipartite graphs (bicliques) have the \textsc{cks} property, the polynomial-delay algorithms for them can be immediately explained by this general result.
On the other hand, the polynomial-delay algorithms of~\cite{schwikowski-02-enumerate-fvs} and \cite{conte-19-polynomial-delay}, and our polynomial-delay algorithms for the maximal (connected) induced interval subgraphs problem
and for the maximal (connected) induced trivially perfect subgraphs problem
cannot be derived from this observation.  As we will see, none of these graph classes has the \textsc{cks} property.

It turns out that star forests, forests in which every tree is a star, play a crucial role in characterizing graph classes with the \textsc{cks} property.  Therefore, to find those graph classes with the \textsc{cks} property, it suffices  to consider those forbidding both a star forest and the complement of a star forest.
\begin{theorem}\label{lem:no-claw-forest}
  Let $\cal F$ be a nonempty set of graphs.  If the class of $\cal F$-free graphs has the \textsc{cks} property, then $\cal F$ contains at least one star forest and the complement of at least one star forest.
\end{theorem}

  Unfortunately, this condition is not sufficient, and many graph classes satisfy this condition but do not have the \textsc{cks} property.  Moreover, it is possible that a class $\mathcal{P}$ of graphs has the \textsc{cks} property but a proper subclass of $\mathcal{P}$ does not.  This fact makes a full characterization of graph classes with the \textsc{cks} property more difficult.

\paragraph{Related work.}  Let us put our work into context.  The aforementioned characterization of a hereditary graph class by a set $\cal F$ of forbidden induced subgraphs provides another perspective to view the maximal induced $\mathcal{P}$ subgraphs problem, but not its connected variation in general.
A subset $S$ of vertices is a solution if and only if $V(G)\setminus S$ intersects all \emph{forbidden sets} of $G$, i.e., a minimal set $X$ of vertices such that $G[X]\in \cal F$.  If we list all the forbidden sets of $G$ as a set system, then the problem becomes
enumerating \emph{minimal} vertex sets that intersect each of the forbidden sets.
This brings us to the well-studied problem of enumerating minimal hitting sets of a set system.  A set system is also known as a \emph{hypergraph}, and the problem is called \emph{hypergraph transversal} in literature.  Some important general results on enumeration of maximal induced $\mathcal{P}$ subgraphs, e.g., the algorithm of for graph classes characterized by a finite number of forbidden induced subgraphs, were first proved in this setting.
  Reducing to the {hypergraph transversal} problem was also a common approach used by many earlier heuristic enumeration algorithms, for maximal independent sets and for maximal forests.
  A graph may have an exponential number of simple cycles, and thus the maximal induced forests problem and its associated hypergraph transversal problem have significantly different input sizes: The latter, including the number of elements and the number of forbidden sets, may be exponential on the former.  This can happen for all graph classes $\mathcal{P}$ that have an infinite number of forbidden induced subgraphs.  Deciding whether a graph belongs to such a class $\mathcal{P}$ may be NP-hard, and hence the maximal induced $\mathcal{P}$ subgraphs problem cannot be solved in polynomial total time, though it is still possible for the corresponding hypergraph transversal problem, with all the forbidden sets given.
Thus, we have to exclude from our study those graph classes that cannot be recognized in polynomial time.
  For a graph class that can be recognized in polynomial time, the maximal (connected) induced $\mathcal{P}$ subgraphs problem is in \textsc{EnumP}, the counterpart of NP for enumeration \cite{capelli-19-incremental-delay-enumeration}.  (One may consider the oracle model, where whether a graph is in the class is answered by an oracle in $O(1)$ time, but this paper will not take this direction.)

The (minimal) hitting sets of a hypergraph $H$ define another hypergraph $H'$, and it is an easy exercise to verify that each set in $H$ is a (minimal) hitting set of $H'$.
The {hypergraph transversal} problem has a decision version: Given a pair of hypergraphs on the same set of elements, decide whether one consists of exactly the minimal hitting sets of the other.
Bioch and Ibaraki~\cite{bioch-95-identification-dualization-positive-boolean-functions} showed that the hypergraph traversal problem can be solved in polynomial total time if and only if it can be solved in incremental polynomial time, by showing that both are equivalent to that the decision version can be solved in polynomial time.  
Note that under the exponential time hypothesis, these two complexity classes for enumeration problems are not equal in general \cite{strozecki-10-thesis, capelli-19-incremental-delay-enumeration}.

The results of \cite{bioch-95-identification-dualization-positive-boolean-functions} were actually developed in the setting of monotone Boolean functions.
A Boolean function $f:\{0, 1\}^n \to \{0, 1\}$ is \emph{monotone} if $x \le y$ always implies $f(x) \le f(y)$.  The \emph{identification} problem is to find all minimal true vectors \emph{and} all maximal false vectors of $f$.  It is easy to see that the maximal induced $\mathcal{P}$ subgraphs problem is a special case of it.  Let $G$ be a graph on $n$ vertices, for any subset $X\subseteq V(G)$, we use $x$ to denote the {characteristic vector} of $X$, i.e., an $n$-dimension Boolean vector with $1$ at the $i$th position if and only if $v_i\in X$, and set $f(x) = 0$ if and only if $G[X]\in \cal P$.  Then minimal forbidden sets and the vertex sets of maximal induced $\mathcal{P}$ subgraphs of $G$ correspond to minimal true vectors {and} maximal false vectors of $f$, respectively.
In this terminology,  the maximal induced $\mathcal{P}$ subgraphs problem asks for maximal false vectors only, hence different in the output size.  (Moreover, the identification problem is usually asked in the oracle model.)
The decision version of the identification problem, known as {dualization of monotone Boolean functions},
is equivalent to the hypergraph traversal problem, and its incarnations can be found in database systems, artificial intelligence, and game theory.
Fredman and Khachiyan \cite{fredman-96-dualization-monotone-dnf} proved that {dualization of monotone Boolean functions}, hence hypergraph traversal, can be solved in quasi-polynomial time.
It has been open for nearly four decades whether this problem can be solved in polynomial time; see the survey of Eiter et al.~\cite{eiter-08-survey-monotone-dualization}.

The main motivation of studying the connected variation is its practical applications in database theory, where several important problems can be modeled as enumerating maximal connected induced $\mathcal{P}$ subgraphs.
We refer to Cohen et al.~\cite{cohen-08-all-maximal-induced-subgraphs} and references wherein for the background.  Since this variation cannot be directly cast into hypergraphs or monotone Boolean functions, fewer results on it have been known in literature \cite{cohen-08-all-maximal-induced-subgraphs, conte-19-polynomial-delay}, and they had to be dealt with separately.

We are exclusively concerned with maximal (connected) induced $\mathcal{P}$ subgraphs.  There is also research on enumerating all induced $\mathcal{P}$ subgraphs and enumerating all maximum induced $\mathcal{P}$ subgraphs.  The first is usually very easy, if not completely trivial, while the latter has to be very hard: After all, finding a single maximum induced subgraph in a nontrivial and hereditary graph class is already NP-hard \cite{lewis-80-node-deletion-np}.
Yet another, probably more practical, approach is to list top $k$ solutions, or solutions of costs at most (or at least) $k$.  This is frequently studied in the framework of parameterized computation.
There is work using the input size as the sole measure for the running time of enumeration algorithms, e.g., \cite{tomita-06-maximal-cliques}.
In summary, efforts toward a systematic understanding of enumeration are gaining momentum \cite{creignou-19-complexity-enumeration}. 

\paragraph{Outline.} The rest of the paper is organized as follows.  Section~\ref{sec:map} formally introduces the framework of solutions maps, and uses it to prove Theorem~\ref{thm:connected-incp=totalp}.  Section~\ref{sec:poly-delay} presents this idea of retaliation-free paths and uses it to solve the maximal (connected) induced interval subgraphs problem and the maximal (connected) induced trivially perfect subgraphs problem with polynomial delay.
Section~\ref{sec:incremental-poly} defines the $t$-restricted version and uses it to prove Theorem~\ref{thm:poly-total}, which is then used to design algorithms for several subclasses of chordal graphs as well as wheel-free graphs.
Section~\ref{sec:cks} gives a necessary condition for the \textsc{cks} property, and show that several graph classes have this property.

\section{Preliminaries}\label{sec:map}
All input graphs discussed in this paper are undirected, simple, and finite.  The vertex set and edge set of a graph $G$ are denoted by, respectively, $V(G)$ and $E(G)$.
For a subset $U\subseteq V(G)$, denote by $G[U]$ the subgraph of $G$ induced by $U$, and by $G - U$ the subgraph $G[V(G)\setminus U]$, which is further simplified to $G - v$ when $U = \{v\}$.
The \emph{neighborhood} of a vertex $v \in V(G)$, denoted by $N(v)$, comprises vertices adjacent to $v$, i.e., $N(v) = \{ u \mid uv \in E(G) \}$, and the \emph{neighborhood} of a vertex set $U$ is $N(U) = \bigcup_{v\in U}N(v)\setminus U$.  The \emph{complement} $\overline G$ of graph $G$ is defined on the same vertex set $V(G)$, where a pair of distinct vertices $u$ and $v$ is adjacent in $\overline G$ if and only if $u v \not\in E(G)$.
A \emph{$k$-path} (resp., \emph{$k$-cycle}) is a path (resp., cycle) on $k$ edges; note that they have $k+1$ vertices and $k$ vertices respectively.
A \emph{clique} is a set of pairwise adjacent vertices, and 
an \emph{independent set} is a set of pairwise nonadjacent vertices.

A graph is \emph{connected} if there is a path between any pair of vertices, and \emph{biconnected} if it remains connected after any vertex deleted.  A vertex $v$ is a \emph{cutvertex} if $G - v$ has more components than $G$.  A connected graph on three or more vertices either is biconnected or contains a cutvertex.

By a \emph{graph class} we mean a nonempty collection of graphs.
A graph class is \emph{trivial} if it contains a finite number of graphs or it contains all but a finite number of graphs.
A graph class $\mathcal{P}$ is \emph{hereditary} if it is closed under taking induced subgraphs: If a graph $G\in \cal P$, then so does $G[U]$ for any nonempty $U\subseteq V(G)$.
A graph is {\em $F$-free}, for some graph $F$, if it does not contain $F$ as an induced subgraph; for a set $\cal F$ of graphs, a graph is {\em $\cal F$-free} if it is {\em $F$-free} for every $F\in \cal F$.  A graph class is hereditary if and only if there is a (possibly infinite) set $\cal F$ of graphs, called the \emph{forbidden induced subgraphs} of this graph class, such that graphs in this class are precisely all $\cal F$-free graphs.  
Every graph $F$ in $\cal F$ is tacitly assumed to be minimal, in the sense that no proper induced subgraph of $F$ is in $\cal F$.
For the convenience of the reader, forbidden induced subgraphs of graph classes studied in this paper are collected in the appendix.

Since we are exclusively concerned with induced subgraphs, it suffices to output their vertex sets.
A set $U\subseteq V(G)$ is a \emph{$\mathcal{P}$ set} of $G$ if $G[U] \in \cal P$, and a \emph{maximal $\mathcal{P}$ set} if none of the proper supersets of $U$ is a $\mathcal{P}$ set of $G$.
Connected $\mathcal{P}$ sets and maximal connected $\mathcal{P}$ sets are defined analogously.
Focusing on (connected) $\mathcal{P}$ sets will simplify our presentation.   We formally define our main problems as follows.

\begin{problem}{Maximal (connected) induced $\mathcal{P}$ subgraphs}
  \Input & a graph $G$.\\
  \Prob & all maximal (connected) $\mathcal{P}$ sets of $G$.
\end{problem}

Throughout this paper, we use $n$ to denote the number of vertices in the input graph $G$, and $N$ the number of solutions, i.e., maximal (connected) $\mathcal{P}$ sets of the input graph $G$.  Clearly, $N < 2^n$.
An enumeration algorithm runs in \emph{polynomial total time} if the number of steps it uses to output all solutions is bounded by a polynomial on $n$ and $N$; in \emph{incremental polynomial time} if for $1\le s\le N$, the number of steps it uses to output the $s$th solution is bounded by a polynomial on $n$ and $s$; and with \emph{polynomial delay} if the number of steps it uses to output each solution is bounded by a polynomial on only $n$.  We start from a widely known fact.

\begin{proposition}[Folklore]\label{lem:enumeration-recognition}
  Let $\mathcal{P}$ be a hereditary graph class.  If the maximal (connected) induced $\mathcal{P}$ subgraphs problem can be solved in polynomial total time, then one can decide in polynomial time whether a (connected) graph is in $\mathcal{P}$.
\end{proposition}
\begin{proof}
  Suppose that there is an algorithm $A$ solving the maximal (connected) induced $\mathcal{P}$ subgraphs problem with at most $a n^c N^{d}$ steps, where $a, c$, and $d$ are fixed constants.  On a (connected) graph in $\mathcal{P}$, the algorithm $A$ is guaranteed to finish in $a n^c$ steps because $N = 1$.
  We use the algorithm $A$ to decide whether a (connected) graph $G$ is in this graph class as follows.  We apply $A$ to $G$, aborted after $a n^{c}$ steps, and return ``yes'' if $A$ returns $V(G)$ itself, or ``no'' otherwise ($A$ may return nothing, or any number of proper subsets of $V(G)$).
\end{proof}

As alluded to earlier, our algorithms proceed by solving the same problem on smaller instances, and then extend each of the obtained solutions of the smaller instances to a solution of the original instance.  By definition, any (connected) $\mathcal{P}$ set of an induced subgraph of $G$ is a (connected) $\mathcal{P}$ set of $G$ itself.  
By \emph{extending} a (connected) $\mathcal{P}$ set $S$ of a graph $G$ we mean the operations of obtaining a maximal (connected) $\mathcal{P}$ set $S'$ of $G$ such that $S\subseteq S'$.
For a graph class that can be recognized in polynomial time, it is very straightforward to extend an arbitrary $\mathcal{P}$ set $S$ of a graph $G$ to a maximal one.  It suffices to try each vertex $v$ in $V(G)\setminus S$ to see whether $S\cup \{v\}$ is still a $\mathcal{P}$ set, thereby calling the recognition algorithm $O(n)$ times.
The point is that we do not need to check the same vertex twice: If $S\cup \{v\}$ is not a $\mathcal{P}$ set, then nor is $S'\cup \{v\}$ for any superset $S'$ of $S$.  The situation is slightly more complicated for the connected variation, where a na\"{i}ve algorithm that tries the vertices in $V(G)\setminus S$ in an arbitrary order may call the recognition algorithm $O(n^2)$ times.
Cohen et al.~\cite{cohen-08-all-maximal-induced-subgraphs} observed that $n$ calls will suffice if we always try the neighbors of a connected $\mathcal{P}$ set.
\begin{proposition}[{\cite{cohen-08-all-maximal-induced-subgraphs}}]
  \label{lem:extension}
  Let $\mathcal{P}$ be a hereditary graph class and $G$ a graph.  For any (connected) $\mathcal{P}$ set $S$ of $G$, we can find a maximal (connected) $\mathcal{P}$ set $S'$ of $G$ with $S\subseteq S'$ in time $n f(n)$, where $f(n)$ is the time to decide whether a (connected) graph is in $\mathcal{P}$.
\end{proposition}

In extending a (connected) $\mathcal{P}$ set $S$, trying the vertices in $V(G)\setminus S$ in different orders may lead to different outputs.  As quite obvious to see, there can be an exponential number of maximal (connected) $\mathcal{P}$ sets that are supersets of $S$.  The other direction is simpler, though nontrivial.  The intersection of a maximal $\mathcal{P}$ set $S$ of $G$ with a subset $U\subseteq V(G)$ is always a $\mathcal{P}$ set of $G[U]$, but may or may not be maximal.  In the case that $S\cap U$ is indeed maximal in $G[U]$, it is the unique maximal $\mathcal{P}$ set of $G[U]$ contained in $S$.  On the other hand, the subgraph $G[S\cap U]$ may not be connected in the first place, and then there can be multiple maximal connected $\mathcal{P}$ set of $G[U]$ that are all subsets of $S$.  Bounds were given by Cohen et al.~\cite[Proposition 3.2]{cohen-08-all-maximal-induced-subgraphs}; note that (i) and (iii) follow from (ii) and (iv) respectively.
\begin{proposition}[{\cite{cohen-08-all-maximal-induced-subgraphs}}]
  \label{lem:subgraph-cardinality}
  Let $G$ be a graph and $\mathcal{P}$ a hereditary graph class.  For any $S\subseteq V(G)$,
  \begin{enumerate}[(i)]
  \item the number of maximal $\mathcal{P}$ sets of $G[S]$ is no more than that of $G$; 
  \item each maximal $\mathcal{P}$ set of $G$ contains at most one maximal $\mathcal{P}$ set of $G[S]$;
  \item the number of maximal connected $\mathcal{P}$ sets of $G[S]$ is no more than $n$ times of that of $G$; and
  \item each maximal connected $\mathcal{P}$ set of $G$ contains at most $n$ maximal connected $\mathcal{P}$ sets of $G[S]$.
  \end{enumerate}
\end{proposition}

\subsection{Solution maps}
We now formally introduce the general framework of using solution maps to solve enumeration problems, with general language instead of only the maximal induced $\mathcal{P}$ subgraphs problem. 
Let $\Pi$ be an enumeration problem, where each solution is a subset of a ground set $U$.  A \emph{successor function} $\operatorname{succ}: 2^U \times U\to 2^{2^U}$ maps a solution and an element to a set of solutions.  (One may alternatively define a {successor function} as $\operatorname{succ}: 2^U\to 2^{2^U}$, and they are practically equivalent.  The explicit specification of the vertex $v$ will be crucial for the next section.)  A \emph{solution map} of this problem is a multiple digraph $M$, or $M(U)$ to emphasize the instance, on all its solutions, where an arc from a solution $S$ to another solution $S'$, labeled with $u\in U$, indicates $S'\in \operatorname{succ}(S, u)$.  The algorithm can be obtained by simply conducting a breadth-first search on the solution map.
Simpler or more specific variations of the following statement have been used in, among others, \cite{schwikowski-02-enumerate-fvs, conte-19-polynomial-delay}.
  
\begin{theorem}[\cite{schwikowski-02-enumerate-fvs, conte-19-polynomial-delay}]
  \label{thm:solution-map}
  Let $\Pi$ be an enumeration problem and let $M$ be the solution map of $\Pi$ implied by the {successor function} $\operatorname{succ}$.   If there are polynomial functions $p$, $q$, and $r$ such that
  \begin{enumerate}[(i)]
  \item one solution of $\Pi$ can be found in $p(n, N)$ time,
  \item the solution map $M$ is strongly connected, and 
  \item the successor function $\operatorname{succ}$ can be evaluated in time $q(n, N)$ for any pair of input,
  \end{enumerate}
  then $\Pi$ can be solved in polynomial total time.  Moreover, if $p$, $q$, and $r$ are all polynomial functions on $n$ only, then $\Pi$ can be solved with polynomial delay.
\end{theorem}

For the maximal (connected) induced $\mathcal{P}$ subgraphs problem, the first condition of Theorem~\ref{thm:solution-map} is equivalent to the existence of a polynomial-time algorithm for deciding whether a (connected) graph is in $\mathcal{P}$.  If such an algorithm exists, we can call Proposition~\ref{lem:extension} to extend any single vertex set to an initial solution.
As a consequence of Proposition~\ref{lem:enumeration-recognition}, we do not concern ourselves with those graph classes for which such an algorithm does not exist, e.g., unit disk graphs \cite{breu-98-unit-disk-graph} and $3$-colorable graphs.
Since all the graph classes we study are well known to be recognizable in polynomial time, we will be focused on constructing the successor function and proving the strong connectedness of the solution map implied by it.

Two quick remarks on Theorem~\ref{thm:solution-map} are in order.   First, the bound on traversal time implies a bound on the maximum out-degree of $M$, while there is no requirement on in-degrees of the solution map $M$.  Second, this approach can be generalized (1) to the case where $M$ is not strongly connected, as long as one solution from each strongly connected component can be found in certain time bound; and (2) to the case where some nodes of $M$ have larger out-degrees, as long as one can start the search from a different node, and has already sufficient number of solutions before reaching the first solution in it: We can amortize the time by delaying the output of each solution; see \cite{johnson-88-generating-maximal-independent-sets} on more details.

\subsection{Reductions between the two variations}

We have seen that the number of maximal $\mathcal{P}$ sets of a graph may be significantly greater than that of maximal connected $\mathcal{P}$ sets of the same graph.  
Every maximal connected $\mathcal{P}$ set of a graph $G$ is a subset of some maximal $\mathcal{P}$ set of $G$, and two maximal connected $\mathcal{P}$ sets that are both subsets of a same maximal $\mathcal{P}$ set have to be disjoint.  Therefore, it is quite straightforward to verify that the number of maximal connected $\mathcal{P}$ sets of a graph cannot be more than $n$ times of the number of maximal $\mathcal{P}$ sets of the same graph.
The following fact is immediate from the definition of hereditary.

\begin{proposition}\label{lem:redundant-vertices}
  Let $\cal P$ be a hereditary graph class.  Let $G$ be a graph, and $U$ a set of vertices of $G$ that are not contained in any forbidden set of $G$.  A set $S\subseteq V(G)$ is a maximal $\cal P$ set of $G$ if and only if $U\subseteq S$ and $S\setminus U$ is a maximal $\cal P$ set of $G - U$.
\end{proposition}

Proposition~\ref{lem:redundant-vertices} enables us to reduce the maximal induced $\mathcal{P}$ subgraphs problem to its connected variation, under certain conditions.
A vertex in a graph $G$ is \emph{universal} if it is adjacent to all the other vertices in this graph.  We say that a graph class $\mathcal{P}$ is \emph{closed under adding universal vertices} if for any graph $G$ in $\mathcal{P}$, the graph obtained by adding a universal vertex to $G$, i.e., a new vertex with edges connecting it to all the vertices in $G$, is also in $\mathcal{P}$.

\begin{corollary}\label{lem:connected}
  Let $\mathcal{P}$ be a graph class that is closed under adding universal vertices.  If the maximal connected induced $\mathcal{P}$ subgraphs problem can be solved in time $f(n, N)$ for some function $f$, then the maximal induced $\mathcal{P}$ subgraphs problem can be solved in time $f(n, N)$ as well.
\end{corollary}
\begin{proof}
  Let $G$ be any graph, and let $G'$ be obtained by adding a universal vertex $u$ to $G$.
Since $\mathcal{P}$ is closed under adding universal vertices, any maximal induced $\mathcal{P}$ subgraph of $G'$ contains $u$, hence being connected.  Therefore, a set $S\subseteq V(G)$ is a maximal $\mathcal{P}$ set of $G$ if and only if $S\cup \{u\}$ is a maximal connected $\mathcal{P}$ set of $G'$.
\end{proof}

A graph class $\mathcal{P}$ is {closed under adding universal vertices} if and only if none of its forbidden induced subgraph contains a universal vertex.  It is easy to verify that the following graph classes are {closed under adding universal vertices}: interval graphs, trivially perfect graphs, complete multi-partite graphs, split graphs, complete split graphs, pseudo-split graphs, and threshold graphs.  
On the other hand, for any $p \ge 2$, a complete $p$-partite graph is either edgeless or connected.  Therefore, an algorithm for the connected variation, together with an algorithm for enumerating maximal independent sets, can be used to enumerate maximal induced complete $p$-partite subgraphs.  Although the class of wheel-free graphs is not closed under adding universal vertices, we are still able to find an easy reduction for this purpose.

\begin{proposition}\label{lem:connected-wheel-free}
  If the maximal connected induced wheel-free subgraphs problem can be solved in time $f(n, N)$ for some function $f$, then the maximal induced wheel-free subgraphs problem can be solved in time $f(n, N)$.
\end{proposition}
\begin{proof}
  Let $G$ be the input graph to the maximal induced wheel-free subgraphs problem.  For each vertex $v$ in $G$, we add a new vertex $v'$ and make it adjacent to only $v$, and then we add a new vertex $u$, and make it adjacent to all the new vertices.  Let the resulting graph be denoted by $G'$.
  Since the vertex $v'$ has degree two in $G'$, it is not in any wheel of $G'$.  It also follows that $u$ is not in any wheel either.  According to Proposition~\ref{lem:redundant-vertices}, every maximal induced wheel-free subgraph of $G'$ contains all the new vertices, hence being connected.  The statement follows.
\end{proof}

The same construction as used in the proof of Proposition~\ref{lem:connected-wheel-free} applies to triangle-free graphs.  With a similar idea as Proposition~\ref{lem:connected-wheel-free}, for $d\ge 2$, we can reduce the maximal induced $d$-degree-bounded subgraphs problem to the maximal connected induced $(d+1)$-degree-bounded subgraphs problem: Make a binary tree with $n$ leaves, and connect each leaf to a different vertex in $G$.  However, we are not able to reduce the maximal induced $\mathcal{P}$ subgraphs problem to the maximal connected induced $\mathcal{P}$ subgraphs problem when $\mathcal{P}$ is one of the following classes: cluster graphs (whose connected variation is the maximal cliques problem), unit interval graphs, block graphs, 3-leaf powers, and basic 4-leaf powers.  Another important graph class for which a reduction is difficult to make is the class of forests.

For most classes, it is very unlikely that we can reduce the connected variation to its counterpart.  The situation changes completely when it comes to $t$-restricted versions, where we only need to worry about solutions containing the special vertices.  The details are left to Section~\ref{sec:incremental-poly}.

\subsection{Proof of Theorem~\ref{thm:connected-incp=totalp}}
  Similar as the proof of Proposition~\ref{lem:enumeration-recognition}, we need to make calls to a given polynomial-total-time algorithm, which we may abort prematurely after certain number of steps.  One should be warned that we cannot make any assumptions on the behavior of the given algorithm.  In particular, the algorithm may output the solutions at the very end of its run, and hence we cannot expect any output unless we afford to wait for it to finish.  The information we can gather from an aborted run is that the number of solutions is larger than the specified value; in other words, the given algorithm is essentially used as an oracle in such cases.
We use this information to produce a ``core subgraph'' of the input graph that has a suitable number of solutions, not too many so that it can be solved in the desired time, and not too few so that one of its solutions can be extended to a new solution of the input graph.

\begin{figure}[h!]
  \tikz\path (0,0) node[draw=gray!50, text width=.9\textwidth, rectangle, rounded corners, inner xsep=20pt, inner ysep=10pt]{
    \begin{minipage}[t!]{\textwidth} \small
      Procedure $\textsc{next}(G, {\cal S}, A)$

      {\sc Input}: A graph $G$ with vertices $v_1$, $\ldots$, $v_n$, a collection $\cal S$ of solutions, and a $p(n, N)$-time 
      \\ \phantom{\sc Input:}
      algorithm $A$ for the maximal connected induced $\mathcal{P}$ subgraphs problem.
      \\
      {\sc Output}: A solution not in $\cal S$, or ``completed'' if there is no further solution.

      \begin{tabbing}
        AAa\=AAa\=AAa\=AAa\=MMMMMMMMMMMMMAAAAAAAAAAAAAAAAAAAAAAAAA\=A \kill
        1.\> apply algorithm $A$ to $G$, aborted after $p(n, |{\cal S}| + 1)$ steps;
        \\
        2. \> {\bf if} it finishes {\bf then}
        \\
        2.1. \>\> {\bf if} all solutions found are in $\cal S$ {\bf then return} ``completed'';
        \\
        2.2. \>\> {\bf else return} a solution not in $\cal S$;
        \\
        3. \> $G_0 \leftarrow G$; 
        \\
        4. \> {\bf for each} $i\leftarrow 1, \ldots, n$ {\bf do}
        \\
        4.1. \>\> apply algorithm $A$ to $G_{i-1} - v_{i}$, aborted after $p(n, n |{\cal S}| + 1)$ steps;
        \\
        4.2. \>\> {\bf if} it finishes {\bf then}
        \\
        4.2.1. \>\>\> {\bf if} a solution $S'$ of $G_{i-1} - v_{i}$ is not a subset of any set in $\cal S$ {\bf then}
        \\
        \>\>\>\> extend $S'$ to a solution $S$ of $G$ and {\bf return} $S$;
        \\
        4.2.2. \>\>\> $G_{i} \leftarrow G_{i-1}$;
        \\
        4.3. \>\> {\bf else} $G_{i} \leftarrow G_{i-1} - v_{i}$;
       \\
       5. \> apply algorithm $A$ to $G_n$; \qquad\qquad \comment{This time wait for it to finish.}
       \\
       6. \> find a solution $S'$ of $G_n$ that is not a subset of any set in $\cal S$;
       \\
       7. \> extend $S'$ to a solution $S$ of $G$ and {\bf return} $S$.
      \end{tabbing}  

    \end{minipage}
  };
  \caption{The procedure for finding the next solution of the maximal connected induced $\mathcal{P}$ subgraphs problem.}
\label{fig:alg-connected-incremental}
\end{figure}

\paragraph{Theorem~\ref{thm:connected-incp=totalp} (restated).}
  For any hereditary graph class $\mathcal{P}$, the maximal connected induced $\mathcal{P}$ subgraphs problem can be solved in polynomial total time if and only if it can be solved in incremental polynomial time.
\begin{proof}
  The if direction is trivially true, and we now show the only if direction.
  Let $G$ be the input graph, and by a solution we mean  a maximal connected $\mathcal{P}$ set of $G$.
  Suppose that algorithm $A$ solves the maximal connected induced $\mathcal{P}$ subgraphs problem with at most $p(n, N)$ steps for some polynomial function $p$.  By Proposition~\ref{lem:enumeration-recognition}, there is a polynomial function $q$ such that we can decide in $q(n)$ time whether a graph on $n$ vertices is a connected graph in $\mathcal{P}$.
We can thus use Proposition~\ref{lem:extension} to extend any connected $\mathcal{P}$ set of $G$ to a maximal connected $\mathcal{P}$ set in time $n q(n)$.
We repetitively call the procedure \textsc{next} described in Figure~\ref{fig:alg-connected-incremental}, which makes calls, abortive or not, to $A$ to find the next solution of $G$, until the procedure returns ``completed.'' 

The main work of procedure \textsc{next} is done in steps~4 and 5, so let us start from understanding them, and in particular, properties of $G_n$ when the procedure reaches step~5.  The procedure proceeds to step~4 only when the call to $A$ in step~1 did not finish, which means that there are more than $|\mathcal{S}|+1$ solutions of $G$.
This implies, in particular, that $G$ itself is not a connected graph in $\mathcal{P}$, and any solution of $G$ is a proper subset of $V(G)$.
Note that at least one call made in step~4.1 finishes: If it does not for the first $n-2$ iterations, then $V(G_{n-2}) = \{v_{n-1}, v_n\}$, and the call on $G_{n-2} - v_{i-1}$ surely returns because a graph on a single vertex is in $\mathcal{P}$ (which is nonempty by assumption).  Also note that when the condition in step~4.2 is satisfied, the condition in step~4.2.1 is trivially satisfied if $\mathcal{S} = \emptyset$.  Therefore, $\mathcal{S}$ is nonempty when the procedure reaches step~5.
We then argue that at least one call of $A$ made in step~4.1 does not finish if the procedure reaches step~5.
Let  $S$ be a solution of $G$ that is not in $\mathcal{S}$, and let $p$ be the smallest number such that $v_p\in V(G)\setminus S$.  Suppose that the calls of $A$ made in step~4.1 finish for all the first $p$ iterations, then $G_{p-1} = G$, and the condition of step~4.2.1 must be true in the $p$th iteration, because $S$ would be a solution of $G_{p-1} - v_p$.  Now suppose that $q$ is the largest number such that the call of $A$ made in step~4.1 of the $q$th iteration, on $G_{q-1} - v_q$, does not finish.  Then $G_n = G_q = G_{q-1} - v_q$, and there are at least $n |\mathcal{S}| + 1$ solutions of $G_n$.  This also implies that $V(G_n)$ is not empty.

For the correctness of this procedure, we show that whatever the results of the calls to $A$, procedure \textsc{next} always returns a correct answer: a new solution of $G$ if $|{\cal S}| < N$ or ``completed'' otherwise.  If $A$ finishes in step~1, which means that it has found all solutions of $G$ , then it is clear that step 2 returns the correct answer.  This must happen when $|{\cal S}| = N$.
Once a solution is returned by step~4.2.1, its correctness is ensured by Proposition~\ref{lem:extension}.
When the procedure reaches step~5, there are more than $n |\cal S|$ solutions of $G_n$.  Hence, the call made in step~5 is guaranteed to find all the solutions of $G_n$.  By Proposition~\ref{lem:subgraph-cardinality}(iv), at least one solution of $G_n$ is not contained in any solution in $\cal S$.  This justifies step~6 and then step~7 always returns a correct solution, again, by Proposition~\ref{lem:extension}.

We now calculate the running time of the procedure, for which the focus is on step~5, because this is the only call of $A$ that is never aborted.  
Let $v_i$ be a vertex in $G_n$; note that the call made on $G_{i - 1} - v_i$ (step~4.1) has finished.
If $G_{i - 1} - v_i$, an induced subgraph of $G$, has more than $n |\cal S|$ maximal connected $\mathcal{P}$ sets, then by Proposition~\ref{lem:subgraph-cardinality}(iv), at least one of them is not a subset of any solution in $\mathcal{S}$, and the condition of step~4.2.1 is true, whereupon the procedure should have terminated before reaching step~5.
  Therefore, the subgraph $G_{i - 1} - v_i$ has at most $n |\cal S|$ maximal connected $\mathcal{P}$ sets.
 Since $G_n - v_i$ is an induced subgraph of $G_{i - 1} - v_i$, by Proposition~\ref{lem:subgraph-cardinality}(iii), $G_n - v_i$ has at most $n^2 |\cal S|$ maximal connected $\mathcal{P}$ sets.
We have shown that each solution of $G_n$ is a proper subset of $V(G_n)$, which is hence also a solution of $G_n - v$ for some $v\in V(G_n)$.  Thus, the total number of solutions of $G_n$ is at most
$|V(G_n)| \cdot n^2|{\cal S}|\le n^3|{\cal S}|$, which means that step~5 takes time $p(n, n^3 |{\cal S}|)$.
Steps~1--3 take $p(n, |{\cal S}| + 1)$, $O(n^2)$, and $O(1)$ time respectively.  Step~4 takes $n \cdot ( p(n, n |{\cal S}| + 1) + n^2\cdot |{\cal S}| + n\cdot q(n))$ time.  Steps~6 and 7 take $O(n^2\cdot |{\cal S}| + n\cdot q(n))$ time.  Putting them together, we can conclude that the running time of procedure \textsc{next} is polynomial on $n$ and $|\cal S|$, and  this completes the proof.
\end{proof}

A similar result as Theorem~\ref{thm:connected-incp=totalp} for the maximal induced $\mathcal{P}$ subgraphs problem was proved by Bioch and Ibaraki~\cite{bioch-95-identification-dualization-positive-boolean-functions} in the setting of dualization of monotone Boolean functions.  A slight modification of our proof of Theorem~\ref{thm:connected-incp=totalp} also works for this classic result, and we include this direct and arguably simpler proof in the appendix.
\begin{theorem}[\cite{bioch-95-identification-dualization-positive-boolean-functions}]
  \label{thm:incp=totalp}
  For any hereditary graph class $\mathcal{P}$, the maximal induced $\mathcal{P}$ subgraphs problem can be solved in polynomial total time if and only if it can be solved in incremental polynomial time.
\end{theorem}

\section{Enumeration with polynomial delay}\label{sec:poly-delay}

We use the general terminology of enumeration problems to introduce retaliation-free paths, where every solution is a subset of some ground set $U$.  Let $\operatorname{succ}: 2^U \times U\to 2^{2^U}$ be the successor function for the problem such that for each solution $S$ and $x\in U\setminus S$, every solution in  $\operatorname{succ}(S, x)$ contains $x$.
For a subset $X\subseteq U$, we use $\Sigma(X)$ to denote all sequences of elements in $X$, and for a sequence $\sigma\in \Sigma(X)$ and $x\in X$ that is not in $\sigma$, let $\sigma+x$ denote the new sequence obtained by appending $x$ to the end of $\sigma$.
For a fixed solution $S^*$ of $U$, the \emph{victim function} $D:\Sigma(S^*) \to \Sigma(S^*)$ satisfies (1) the elements in $D(\langle\rangle)$ are precisely those in $S^*$, and (2) for an element $x$ in $D(\sigma)$, every element in $D(\sigma+x)$ is before $x$ in $D(\sigma)$.
Note that we do \emph{not} require that two elements in $D(\sigma+x)$ to appear in the same order in $D(\sigma+x)$ as in $D(\sigma)$. 

Let us motivate these definitions.  In the lexicographically metric used by Schwikowski and Speckenmeyer~\cite{schwikowski-02-enumerate-fvs} and Conte and Uno~\cite{conte-19-polynomial-delay}, elements in $U$ are numbered in a way that the smallest ones are in $S^*$.  This ordering of $S^*$ can be viewed as $D(\langle \rangle)$ in our setting, and for another solution $S$, we try to add to it elements in $S^*\setminus S$ in this order.  However, we allow our successor function to be defined in a way that every successor of a solution $S$ is lexicographically larger than $S$.  If the addition of an element $x$ leads to the removal of elements smaller than $x$ in $D(\langle \rangle)$, called \emph{victims}, then $x$ is a \emph{gainer}, and we immediately start dealing with its victims.   To keep track the gainers we associate a sequence $\sigma$ to each solution.
Elements in the sequence $D(\langle x \rangle)$ are all the possibly victims of $x$, and the order dictates who is to be added first.  If further victims are induced during processing a victim $y$ of a sequence $\sigma$ of gainers, then $y$ becomes a new gainer and is appended to the end of $\sigma$; if the victims of the last gainer in $\sigma$ have all been restored, it is removed from $\sigma$; otherwise, $\sigma$ remains unchanged, and we proceed to the next victim of the last element in $\sigma$.  In summary, the input of $D$ is a sequence of gainers, while the output of $D$ is all the possible victims of the last gainer, in a certain order.

 Let $M(U)$ be the implied solution map by the successor function $\operatorname{succ}$, and let $S$ be a solution of $U$.  A path $S_0 S_1 S_2\ldots$ of solutions in $M(U)$ is called a \emph{retaliation-free path} from $S$ if $S_0 = S$ and $\sigma_0 = \langle\rangle$, and for any $i$ with $S_{i} \ne S^*$,
\begin{enumerate}[(i)]
\item\label{retaliation-free-2} $S_{i+1}\in \operatorname{succ}(S_i, x_i)$, where $x_i$ is the first element of $D(\sigma_i)$ that is not in $S_i$; 
  \item\label{retaliation-free-3} every element in $(S_{i}\setminus S_{i+1}) \cap S^*$ appears in $D(\sigma_i)$; and
    \item\label{retaliation-free-4} $\sigma_{i+1}$ is the longest prefix $\sigma'$ of $\sigma_{i} + x_i$ such that at least one element in $D(\sigma')$ is not in $S_{i+1}$.
\end{enumerate}

Therefore, a retaliation-free path terminates only when it reaches $S^*$.  Note that for any $i\ge 0$, the sequence $\sigma_{i+1}$ either is a prefix of $\sigma_{i}$, or has precisely one more element than $\sigma_{i}$.  This implies, in particular, that every prefix of $\sigma_i$ is $\sigma_j$ for some $j \le i$.
We are now ready for the formal version of Theorem~\ref{thm:retaliation-free}.
\begin{theorem}\label{thm:formal-retaliation-free}
  Let $M$ be the solution map of an enumeration problem.  If 
  for any solution $S^*$, there exists a retaliation-free path from any other solution $S$, then $M$ is strongly connected.
\end{theorem}
\begin{proof}
  We first argue that for any sequence $\sigma'\in \Sigma(S^*)$, those sequences in $\sigma_0, \sigma_1, \sigma_2, \ldots$ that have $\sigma'$ as a prefix are consecutive.
  Suppose for contradiction that there are $p$ and $q$, $0\le p < q$, such that both $\sigma_p$ and $\sigma_q$ start from the sequence $\langle i_1 i_2 \cdots i_s\rangle$ but $\langle i_1 i_2 \cdots i_s\rangle$ is not a prefix for some $\sigma_j$ with $p < j < q$.
We may assume without loss of generality that for all $j, p <j< q$, the sequence $\sigma_j$ starts with $\langle i_1 i_2 \cdots i_{s-1}\rangle$ but is not followed by $i_s$.  Otherwise, we can find another pair of $p'$ and $q'$ such that $p \le p' < q'\le q'$ and $q'-p' < q - p$, possibly with another sequence.
Under this assumption, $\sigma_{q-1} = \langle i_1 i_2 \cdots i_{s-1}\rangle$ and $\sigma_q = \langle i_1 i_2 \cdots i_{s}\rangle$.  By the definition of retaliation-free paths, $i_s\in S_{q}\setminus S_{q-1}$.

Let $i^*$ be the smallest element in $D(\langle i_1 i_2 \cdots i_{s-1} \rangle)$ such that (1) $\langle i_1 i_2 \cdots i_{s-1} i^* \rangle$ is a prefix of $\sigma_{q'}$ for some ${q'}, p < {q'} \le q$; and (2) for any $j, p<j <{q'}$, if $\sigma_{j} \ne \langle i_1 i_2 \cdots i_{s-1}\rangle$, then the $s$th element of $\sigma_{j}$ is after $i^*$ in $D(\langle i_1 i_2 \cdots i_{s-1} \rangle)$.  Note that the element $i^*$ exists and it is either $i_s$ or before $i_s$ in $D(\langle i_1 i_2 \cdots i_{s-1} \rangle)$ because $\sigma_q= \langle i_1 i_2 \cdots i_{s-1}i_s\rangle$.  
By our assumption,
$\sigma_{p+1}$ starts from $\langle i_1 i_2 \cdots i_{s-1}\rangle$ and it is not followed by $i_s$.  From (iii) of the definition of retaliation-free paths we can conclude that $\sigma_{p+1} = \langle i_1 i_2 \cdots i_{s-1}\rangle$, and all elements in $D(\langle i_1 i_2 \cdots i_{s-1} i_s \rangle)$ as well as $i_s$ itself are in $S_{p+1}$.

We argue that $i^*\in S_{p+1}$.  We have seen this if $i^* = i_s$, and in the rest $i^* \ne i_s$.
If $i^*$ is in $D(\langle i_1 i_2 \cdots i_{s-1} i_s \rangle)$, then it follows from (iii) of the definition of retaliation-free paths.  Otherwise, let $p'$ be the largest number with $p' < p$ such that $\sigma_{p'} = \langle i_1 i_2 \cdots i_{s-1} \rangle$ and $\sigma_{p'+1} = \langle i_1 i_2 \cdots i_{s-1} i_s\rangle$.  It exists because $\langle i_1 i_2 \cdots i_{s-1} i_s\rangle$ is a prefix of $\sigma_p$.  Then by (i) of the definition of retaliation-free paths, $i^*$ is in $S_{p'}$.  For every $j, p' < j < p$, the sequence $\langle i_1 i_2 \cdots i_{s-1} i_s\rangle$ is a prefix of $\sigma_{j}$, and hence a vertex in $(S_j\setminus S_{j+1})\cap S^*$ is in $D(\langle i_1 i_2 \cdots i_{s-1} i_s\rangle)$.  Therefore, we always have $i^*\in S_{p+1}$.

By the selection of $q'$, the solution  $S_{q'-1}$ does not contain $i^*$.  Suppose that $p'$ is the largest number such that $p' < q'$ and $i^*\in S_{p'}\setminus S_{p'+1}$.  By the selection of $i^*$, the $s$th element of $\sigma_{p'+1}$, (which is either the $s$th element of $\sigma_{p'}$, or $x_{p'}$, the first element of $D(\sigma_{q'})$ that is not in $S_{p'}$,) is after $i^*$ in $D(\langle i_1 i_2 \cdots i_{s-1} \rangle)$.  But since $i^*$ is not in $S_{p'+1}$, $\ldots$ $S_{q'}$, the first $s$ elements in $\sigma_{p'+1}$ should remain a prefix of $\sigma_{q'}$, a contradiction to that $\langle i_1 i_2 \cdots i_{s-1} i^* \rangle$ is a prefix of $\sigma_{q'}$.

We then argue that a retaliation-free path must reach $S^*$; i.e., it must terminate.  There are only a finite number of sequences, and each one is used in consecutive steps.  Suppose that the path does not terminate, then there is some number $i$ such that $\sigma_j = \sigma_i$ for all $j > i$.  Let $d$ denote the number of elements in $D(\sigma_i)$.  But then all elements in $D(\sigma_i)$ must be in $S_{i + d'}$ for some $d' \le d$, which contradicts that $\sigma_{i + d'} = \sigma_i$.
\end{proof}

\subsection{Trivially perfect graphs}

From a rooted tree $T$, we can define a graph $G$, called the comparability graph of $T$, as follows.  Its vertex set is the same as $V(T)$, and two distinct vertices $u$ and $v$ are adjacent in $G$ if and only if one of them is the ancestor of the other in $T$.  This can be generalized to a rooted forest $F$ by taking the comparability graph of each component of $F$ separately.  Such a graph is called a \emph{trivially perfect graph}, and the forest $F$ is a \emph{generating forest} of this graph \cite{wolk-62, yan-96-trivially-perfect}.

Let $T$ be a tree, and $v$ a vertex in $T$.  All the vertices in the subtree of $T$ rooted at $v$ are its \emph{descendants}.  Note that a vertex is a descendant of itself, and a \emph{proper descendant} of $v$ is a descendant of $v$ that is different from $v$ itself.  A vertex $x$ is a \emph{(proper) ancestor} of $v$ if $v$ is a (proper) descendant of $x$.   When there are more than one trees in the context, we may use $T$-descendant and $T$-ancestor to emphasize the tree to which we are referring.
For two sets $X$ and $Y$, we use $X\triangle Y$ to denote the symmetric difference of $X$ and $Y$, i.e., vertices in precisely one of $X$ and $Y$.

  Since the class of trivially perfect graphs is closed under adding universal vertices, we may focus on the connected variation, and hence all the generating forests are trees.  We now introduce the successor function for the problem.  Let $G$ be the input graph.  For each solution $S$ of $G$ and each vertex $v\in V(G)\setminus S$, we define the set $\operatorname{succ}(S, v)$ of successors as follows.  We fix a generating tree $T$ for $G[S]$.
For each vertex $u\in N(v)$, we let $S_u = S\cap N(u)$, create an empty multi-set $\mathcal{S}_u$, 
\begin{enumerate}[(tp1)]
\item[(tp0)] add $S_u\setminus N(v)$ to $\mathcal{S}_u$;
\item for each $v'\in S_u\cap N(v)$, add $S_u\setminus (A\cup D)$ to $\mathcal{S}_u$, where $A$ is the set of $T$-ancestors of $v'$ that are not adjacent to $v$, and $D$ is the set of proper $T$-descendants of $v'$ that are adjacent to $v$;
\item for each $v'\in S_u\cap N(v)$, add $S_u\setminus (N(v')\triangle N(v))$  to $\mathcal{S}_u$; and
\item for each pair of nonadjacent vertices $v_1, v_2\in S_u\cap N(v)$, find the least common ancestor $v'$ of $v_1$ and $v_2$, add $S_u\setminus (A\cup B\cup C)$  to $\mathcal{S}_u$, where $A$ comprises of the proper descendants of $v'$ that are proper ancestors of $v_1$ or $v_2$, $B$ the set of vertices adjacent to at least one of $v_1$ and $v_2$ but not $v$, and $C$ the set of vertices adjacent to $v$ but neither of $v_1$ and $v_2$.
 \end{enumerate}
For each subset $S'$ in $\mathcal{S}_u$, we use Proposition~\ref{lem:extension} to extend $S'\cup \{u, v\}$, which induces a connected subgraph because $u$ is universal in this subgraph, to a solution of $G$, and add this solution to $\operatorname{succ}(S, v)$.  It is worth noting that $u$ may or may not be universal in the solution extended from $S'\cup \{u, v\}$.
After the whole process is finished, there are $O(n^3)$ solutions in $\operatorname{succ}(S, v)$, all of which contain $v$.

\begin{lemma}\label{lem:tpg-successor-function}
Let $G$ be the input graph to
  the maximal connected induced trivially perfect subgraphs problem.
  For any solution  $S$ of $G$, and any vertex $v$ in $V(G)\setminus S$, every set in $\operatorname{succ}(S, v)$ is a solution of $G$ containing $v$, and $\operatorname{succ}(S, v)$ can be calculated in $O(n^5)$ time.
\end{lemma}
\begin{proof}
  For the correctness, by Proposition~\ref{lem:extension}, it suffices to show that for every $S'\in \mathcal{S}_u$, the set $S'\cup \{u, v\}$ is a connected trivially perfect set of $G$.    Let $G' = G[S'\cup \{u, v\}]$, and we construct a generating tree for $G'$ as follows.  Let $T$ be the generating tree for $G[S]$ used in the calculation of $\mathcal{S}_u$.  We take the sub-forest of $T$ induced by $S'$, and make a new tree $T'$ by adding $u$ as the root to the sub-forest.

  If $S'$ is produced in (tp0), $v$ has no neighbor in $S'$, and adding $v$ as a child to $u$ makes a generating tree for $G'$.
  If $S'$ is produced in (tp1), then we can make a generating tree for $G'$ by adding $v$ as a child to $v'$.
  If $S'$ is produced in (tp2), then we can make a generating tree for $G'$ by adding $v$ in between $v'$ and its parent in $T$; i.e., adding $v$ as a child to the parent of $v'$, and then adding $v'$ as a child to $v$.

  If $S'$ is produced in (tp3), then we proceed as follows.
  We add $v$ as a child to $v'$, and add $v_1$ and $v_2$ as children to $v$.  For any other vertex in $S'$ whose parent is in $S\setminus S'$, we add it as a child of its lowest surviving ancestor, or $u$ if there is none.
  To see that this new tree is a generating tree for $G'$, it suffices to consider the adjacencies between all vertices and $v$.  Any proper descendant $x$ of $v$ is a descendant of $v_1$ or $v_2$; since $x$ is adjacent to $v_1$ or $v_2$ and remains in $S'$, it is adjacent to $v$ as well.  Similarly for proper ancestors of $v$: They are ancestors of $v'$, hence ancestors of $v_1$ or $v_2$.  By definition of generating trees, these are all the vertices in $S'$ that are adjacent to at least one of $v_1$ and $v_2$.  By the steps of (tp3), they are all the vertices in $S'$ that are adjacent to $v$.  

  For each $u$, we make $O(n^2)$ subsets.  We can produce each subset $S'$ in time $O(n^2)$, and extend $S'\cup\{u, v\}$ to a solution in $O(n^3)$ time.  The total time of producing $\operatorname{succ}(S, v)$ is thus  $O(n^3\cdot (n^2 + n^3)) = O(n^6)$.
\end{proof}

\begin{figure}[h!]
  \centering
  \begin{subfigure}[b]{.9\linewidth}
    \centering
    \begin{tikzpicture}[scale=2.5]
      \node[filled vertex] (v) at (2, 1.25) {};
      \node[filled vertex] (v9) at (2, .25) {};
      \node[uvertex, "$21$" above right] (z1) at (3.25, .7) {};
      \node[filled vertex] (z2) at (4., .5) {};
      \node[filled vertex] (z3) at (4., 0) {};
      \node[filled vertex] (u) at (1, 0.5) {};
      \node[uvertex, "$20$" left] (u2) at (.75, 0.5) {};
      \node[filled vertex] (w) at (2.9, 0.5) {};

      \foreach[count=\i, count=\j from 22] \p in {1, 2, 3, 5, 6, 7} {
        \node[uvertex, "$\j$" left] (a\i) at (\p*.5, 0.1) {};
        \node[filled vertex] (b\i) at (\p*.5-.1, -.2) {};
        \node[filled vertex] (c\i) at (\p*.5+.1, -.2) {};
        \draw (b\i) -- (a\i) -- (c\i);
      }
      \foreach \i in {1, 2, 3} {
        \draw (u) -- (c\i) -- (v);
      }
      \foreach \i in {4, 5, 6} {
        \draw (w) -- (c\i) -- (v);
        \draw (a\i) -- (z1);
        \draw (b\i) -- (z1);
        \draw (c\i) -- (z1);
      }

      \draw (w) -- (v) -- (u);
      \draw (v) -- (u2) (c1) -- (u2) -- (c2);
      \draw (z3) -- (z2) -- (z1) -- (z3);
      \draw (v9) -- (v) -- (z1) -- (w);      
    \end{tikzpicture}
    \caption{The input graph $G$, with a universal vertex omitted.  Round vertices are in $S^*$.}
  \end{subfigure}  

  \begin{subfigure}[b]{.43\linewidth}
    \centering
    \begin{tikzpicture}[every node/.style={filled vertex}, scale=1.6]
      \node["$9$"below ] (v9) at (2, .25) {};
      \node["{\footnotesize $10$}"] (v) at (2, 1.25) {};
      \node["{\footnotesize $18$}"] (z2) at (4., .5) {};
      \node["{\footnotesize $17$}" below] (z3) at (4., 0) {};
      \node["$4$"] (u) at (1, 0.5) {};
      \node["$8$"] (w) at (3, 0.5) {};

      \foreach[count=\i from 11] \p in {1, 2, 3, 5, 6, 7} 
      \node["{\footnotesize $\i$}" below] (b\i) at (\p*.5-.15, -.2) {};
      \foreach[count=\i] \p in {1, 2, 3, 5, 6, 7} {
        \node["$\p$" below] (c\i) at (\p*.5+.1, -.2) {};
      }
      \foreach \i in {1, 2, 3} {
        \draw (u) -- (c\i);
      }
      \foreach \i in {4, 5, 6} {
        \draw (w) -- (c\i);
      }

      \draw (w) -- (v) -- (u);
      \draw (z3) -- (z2);
      \draw (v9) -- (v);
    \end{tikzpicture}
    \caption{$S^*$ and the ordering.}
  \end{subfigure}  
  \quad
  \begin{subfigure}[b]{.43\linewidth}
    \centering
    \begin{tikzpicture}[scale=1.6]
      \node[uvertex, "{\footnotesize $21$}"] (z1) at (3.25, 1) {};
      \node[filled vertex] (z2) at (4., .5) {};
      \node[filled vertex] (z3) at (4., 0) {};
      \node[filled vertex] (v9) at (2, .25) {};
      
      \foreach[count=\i, count=\j from 22] \p in {1, 2, 3, 5, 6, 7} {
        \node[uvertex, "{\footnotesize $\j$}"] (a\i) at (\p*.5, 0) {};
        \node[filled vertex](b\i) at (\p*.5-.1, -.2) {};
        \node[filled vertex](c\i) at (\p*.5+.1, -.2) {};
        \draw (b\i) -- (a\i) -- (c\i);
      }
      \foreach \i in {4, 5, 6} {
        \draw (a\i) -- (z1);
      }

      \draw (z3) -- (z2) -- (z1);
    \end{tikzpicture}
    \caption{$S_0$; $\sigma_0 = \langle\rangle$; $\alpha(0) = 4$.}
  \end{subfigure}

  \begin{subfigure}[b]{.43\linewidth}
    \centering
    \begin{tikzpicture}[scale=1.6]
      \node[uvertex, "{\footnotesize $21$}" right] (z1) at (3.25, 1) {};
      \node[filled vertex] (z2) at (4., .5) {};
      \node[filled vertex] (z3) at (4., 0) {};
      \node[filled vertex] (v9) at (2, .25) {};
      \node[filled vertex, "$4$"] (u) at (1, 0.5) {};
      
      \foreach[count=\i from 4, count=\j from 25] \p in {5, 6, 7} {
        \node[uvertex, "{\footnotesize $\j$}"](a\i) at (\p*.5, 0) {};
        \node[filled vertex](b\i) at (\p*.5-.1, -.2) {};
        \node[filled vertex](c\i) at (\p*.5+.1, -.2) {};
        \draw (b\i) -- (a\i) -- (c\i);
        \draw (a\i) -- (z1);
      }
      \foreach[count=\i] \p in {1, 3} {
        \node[filled vertex](b\i) at (\p*.5-.1, -.2) {};
        \node[filled vertex](c\i) at (\p*.5+.1, -.2) {};
        \draw (c\i) -- (u);
      }
      \foreach \p in {2} {
        \node[uvertex, "{\footnotesize $23$}"](a\p) at (\p*.5, 0) {};
        \node[filled vertex](b\p) at (\p*.5-.1, -.2) {};
        \draw (a\p) -- (b\p);
      }

      \draw (z3) -- (z2) -- (z1);
    \end{tikzpicture}
    \caption{$S_1$; $\sigma_1 = \langle 4\rangle$; $\alpha(1) = 2$.}
  \end{subfigure}  
  \quad
  \begin{subfigure}[b]{.43\linewidth}
    \centering
    \begin{tikzpicture}[scale=1.6]
      \node[uvertex, "{\footnotesize $21$}" right] (z1) at (3.25, 1) {};
      \node[filled vertex] (z2) at (4., .5) {};
      \node[filled vertex] (z3) at (4., 0) {};
      \node[filled vertex] (v9) at (2, .25) {};
      \node[filled vertex] (u) at (1, 0.5) {};

      \foreach[count=\i from 4, count=\j from 25] \p in {5, 6, 7} {
        \node[uvertex, "{\footnotesize $\j$}"](a\i) at (\p*.5, 0) {};
        \node[filled vertex](b\i) at (\p*.5-.1, -.2) {};
        \node[filled vertex](c\i) at (\p*.5+.1, -.2) {};
        \draw (b\i) -- (a\i) -- (c\i);
        \draw (a\i) -- (z1);
      }
      \foreach[count=\i] \p in {1, 2, 3 } {
        \node[filled vertex](b\i) at (\p*.5-.1, -.2) {};
        \node[filled vertex](c\i) at (\p*.5+.1, -.2) {};
        \draw (c\i) -- (u);
      }
      \node[filled vertex, "$2$" above right] at (1.1, -.2) {};
            
      \draw (z3) -- (z2) -- (z1);
    \end{tikzpicture}
    \caption{$S_2$; $\sigma_2 = \langle\rangle$; $\alpha(2) = 8$.}
  \end{subfigure}  

  \begin{subfigure}[b]{.43\linewidth}
    \centering
    \begin{tikzpicture}[scale=1.6]
      \node[uvertex, "{\footnotesize $21$}" right] (z1) at (3.25, 1) {};
      \node[filled vertex] (z2) at (4., .5) {};
      \node[filled vertex] (z3) at (4., 0) {};
      \node[filled vertex] (v9) at (2, .25) {};
      \node[filled vertex] (u) at (1, 0.5) {};
      \node[filled vertex, "$8$" left] (w) at (3, 0.5) {};
      
      \foreach[count=\i from 6] \p in {7} {
        \node[uvertex, "{\footnotesize $27$}" right](a\i) at (\p*.5, 0) {};
        \node[filled vertex](b\i) at (\p*.5-.1, -.2) {};
        \draw (b\i) -- (a\i);
        \draw (a\i) -- (z1);
      }
      \foreach[count=\i] \p in {1, 2, 3, 5, 6} {
        \node[filled vertex](b\i) at (\p*.5-.1, -.2) {};
        \node[filled vertex](c\i) at (\p*.5+.1, -.2) {};
      }
      \foreach \i in {1, 2, 3} \draw (c\i) -- (u);
      \foreach \i in {4, 5} \draw (c\i) -- (w) (b\i) -- (z1);

      \draw (z1) -- (w);
      \draw (z3) -- (z2) -- (z1);
    \end{tikzpicture}
    \caption{$S_3$; $\sigma_3 = \langle 8\rangle$; $\alpha(3) = 7$.}
  \end{subfigure}  
  \quad
  \begin{subfigure}[b]{.43\linewidth}
    \centering
    \begin{tikzpicture}[scale=1.6]
      \node[uvertex, "{\footnotesize $21$}" right] (z1) at (3.25, 1) {};
      \node[filled vertex] (z2) at (4., .5) {};
      \node[filled vertex] (z3) at (4., 0) {};
      \node[filled vertex] (v9) at (2, .25) {};
      \node[filled vertex] (u) at (1, 0.5) {};
      \node[filled vertex] (w) at (3, 0.5) {};
      
      \foreach[count=\i] \p in {1, 2, 3, 5, 6, 7} {
        \node[filled vertex](b\i) at (\p*.5-.1, -.2) {};
        \node[filled vertex](c\i) at (\p*.5+.1, -.2) {};
      }
      \foreach \i in {1, 2, 3} \draw (c\i) -- (u);
      \foreach \i in {4, 5, 6} \draw (c\i) -- (w) (b\i) -- (z1);

      \node[filled vertex, "$7$"] at (3.6, -.2) {};
      \draw (z1) -- (w);
      \draw (z3) -- (z2) -- (z1);
    \end{tikzpicture}
    \caption{$S_4$; $\sigma_4 = \langle \rangle$; $\alpha(4) = 10$.}
  \end{subfigure}  

  \begin{subfigure}[b]{.43\linewidth}
    \centering
    \begin{tikzpicture}[scale=1.6]
      \node[filled vertex, "{\footnotesize $10$}"] (v) at (2, 1.25) {};
      \node[filled vertex] (v9) at (2, .25) {};
      \node[filled vertex] (z2) at (4., .5) {};
      \node[filled vertex] (z3) at (4., 0) {};
      \node[filled vertex] (u) at (1, 0.5) {};
      \node[filled vertex] (w) at (3, 0.5) {};
      
      \foreach[count=\i, count=\j from 22] \p/\q in {1/22, 3/24} {
        \node[uvertex, "{\footnotesize $\q$}"] (a\i) at (\p*.5, 0) {};
        \node[filled vertex](b\i) at (\p*.5-.1, -.2) {};
        \draw (b\i) -- (a\i);
      }
      \foreach[count=\i from 3] \p in {2, 5, 6, 7} {
        \node[filled vertex](b\i) at (\p*.5-.1, -.2) {};
        \node[filled vertex](c\i) at (\p*.5+.1, -.2) {};
      }
      \foreach \i in {3} \draw (c\i) -- (u);
      \foreach \i in {4, 5, 6} \draw (c\i) -- (w);


      \draw (w) -- (v) -- (u);      
      \draw (z3) -- (z2);
      \draw (v9) -- (v);      
    \end{tikzpicture}
    \caption{$S_5$; $\sigma_5 = \langle 10\rangle$; $\alpha(5) = 1$.}
  \end{subfigure}  
  \quad
  \begin{subfigure}[b]{.43\linewidth}
    \centering
    \begin{tikzpicture}[scale=1.6]
      \node[filled vertex] (v) at (2, 1.25) {};
      \node[filled vertex] (v9) at (2, .25) {};
      \node[filled vertex] (z2) at (4., .5) {};
      \node[filled vertex] (z3) at (4., 0) {};
      \node[uvertex, "{\footnotesize $20$}"] (u2) at (.75, 0.5) {};
      \node[filled vertex] (w) at (3, 0.5) {};
      
      \foreach[count=\i] \p in {3} {
        \node[uvertex, "{\footnotesize $24$}"] (a\i) at (\p*.5, 0) {};
        \node[filled vertex](b\i) at (\p*.5-.1, -.2) {};
        \draw (b\i) -- (a\i);
      }
      \foreach[count=\i from 2] \p in {1, 2, 5, 6, 7} {
        \node[filled vertex](b\i) at (\p*.5-.1, -.2) {};
        \node[filled vertex](c\i) at (\p*.5+.1, -.2) {};
      }
      \foreach \i in {2, 3} \draw (c\i) -- (u2);
      \foreach \i in {4, 5, 6} \draw (c\i) -- (w);

      \node[filled vertex, "$1$" above] at (.6, -.2) {};

      \draw (w) -- (v) -- (u2);      
      \draw (z3) -- (z2);
      \draw (v9) -- (v);      
    \end{tikzpicture}
    \caption{$S_6$; $\sigma_6 = \langle 10\rangle$; $\alpha(6) = 3$.}
  \end{subfigure}  

  \begin{subfigure}[b]{.43\linewidth}
    \centering
    \begin{tikzpicture}[scale=1.6]
      \node[filled vertex] (v) at (2, 1.25) {};
      \node[filled vertex] (v9) at (2, .25) {};
      \node[filled vertex] (z2) at (4., .5) {};
      \node[filled vertex] (z3) at (4., 0) {};
      \node[uvertex, "{\footnotesize $20$}"] (u2) at (.75, 0.5) {};
      \node[filled vertex] (w) at (3, 0.5) {};
      
      \foreach[count=\i from 1] \p in {1, 2, 3, 5, 6, 7} {
        \node[filled vertex](b\i) at (\p*.5-.1, -.2) {};
        \node[filled vertex](c\i) at (\p*.5+.1, -.2) {};
      }
      \foreach \i in {1, 2} \draw (c\i) -- (u2);
       \draw (c3) -- (v);
      \foreach \i in {4, 5, 6} \draw (c\i) -- (w);

      \node[filled vertex]["$3$" right] at (1.6, -.2) {};

      \draw (w) -- (v) -- (u2);      
      \draw (z3) -- (z2);
      \draw (v9) -- (v);      
    \end{tikzpicture}
    \caption{$S_7$; $\sigma_7 = \langle 10\rangle$; $\alpha(7) = 4$.}
  \end{subfigure}  
  \quad
  \begin{subfigure}[b]{.43\linewidth}
    \centering
    \begin{tikzpicture}[scale=1.6]
      \node[filled vertex] (v) at (2, 1.25) {};
      \node[filled vertex] (v9) at (2, .25) {};
      \node[filled vertex] (z2) at (4., .5) {};
      \node[filled vertex] (z3) at (4., 0) {};
      \node[filled vertex]["$4$"] (u) at (1, 0.5) {};
      \node[filled vertex] (w) at (3, 0.5) {};
      
      \foreach[count=\i] \p in {2} {
        \node[uvertex, "{\footnotesize $23$}"] (a\i) at (\p*.5, 0) {};
        \node[filled vertex] (b\i) at (\p*.5-.1, -.2) {};
        \draw (b\i) -- (a\i);        
      }
      \foreach[count=\i from 2] \p in {1, 3, 5, 6, 7} {
        \node[filled vertex](b\i) at (\p*.5-.1, -.2) {};
        \node[filled vertex](c\i) at (\p*.5+.1, -.2) {};
      }
      \foreach \i in {1,3} \draw (c\i) -- (u);
      \foreach \i in {4, 5, 6} \draw (c\i) -- (w);

      \node[filled vertex] at (3.6, -.2) {};

      \draw (w) -- (v) -- (u);
      \draw (z3) -- (z2);
      \draw (v9) -- (v);
    \end{tikzpicture}
    \caption{$S_8$; $\sigma_8 = \langle 10, 4\rangle$; $\alpha(8) = 2$.}
  \end{subfigure}  
  
  \caption{Demonstration of a retaliation-free path for the maximal connected induced trivially perfect subgraphs problem.  The solutions, $S^*$, $S_0$, $\ldots$, $S_8$, of $G$ are presented as generating trees in (b)--(k), and in all of them, the root, which is the universal vertex of $G$ and numbered 19, is omitted.
In every transition step, we use the universal vertex as $u$.  Finally, $S_9 = S^*$.}
  \label{fig:trivially-perfect-demonstration}
\end{figure}
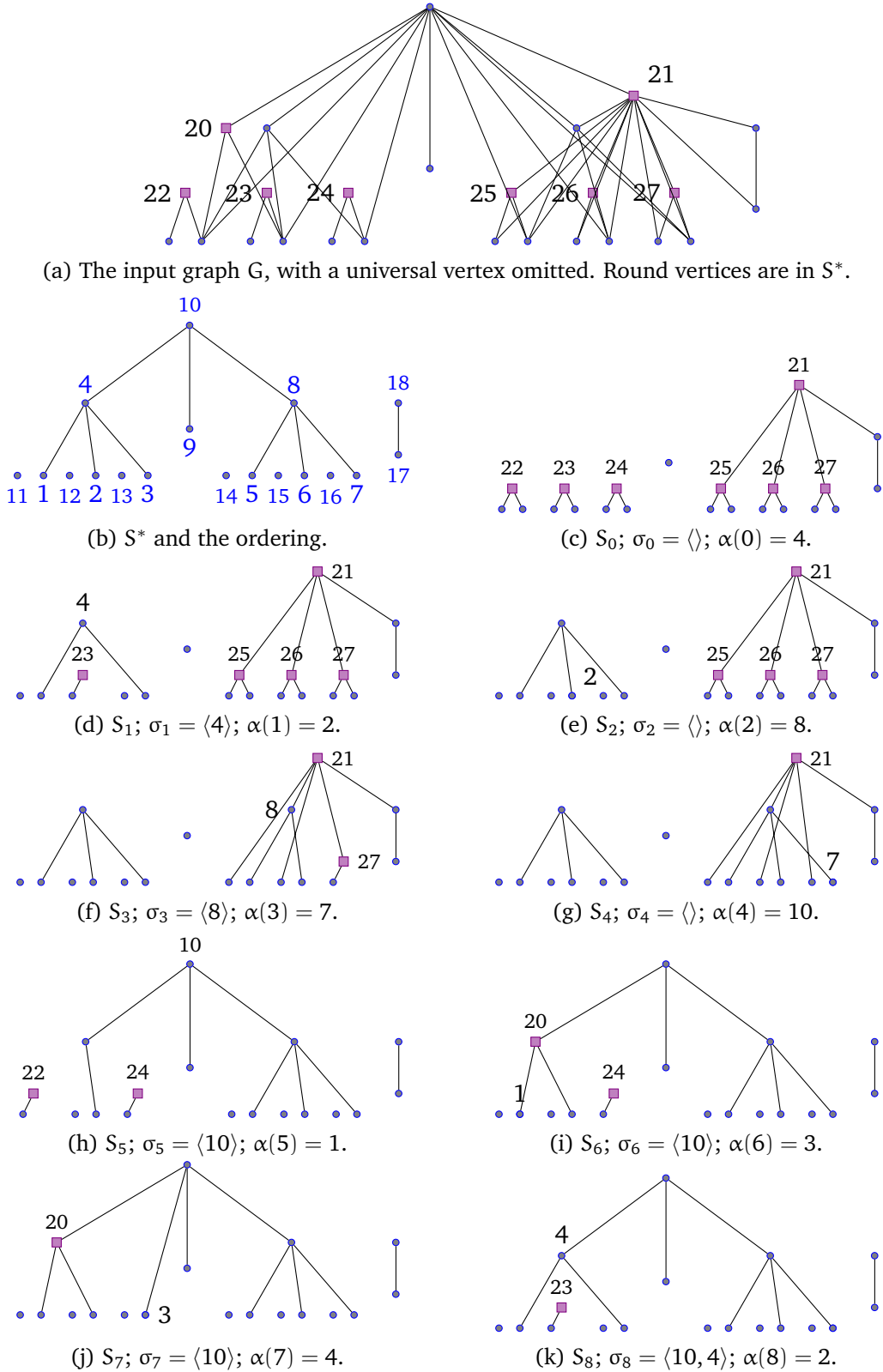

We add an arc with label $v$ from $S$ to each solution in $\operatorname{succ}(S, v)$.  We use the resulting multiple digraph, denoted by $M(G)$, as the solution map for the graph $G$.  By Lemma~\ref{lem:tpg-successor-function}, the out-degree of a node in $M(G)$ is at most $O(n^4)$.  In the rest we show that $M(G)$ is strongly connected by demonstrating a retaliation-free path from any solution $S$ to another solution $S^*$.   Following this path amounts to building a generating tree for $G[S^*]$ in the bottom-up manner, and we make sure that during the whole process, all victims of a gainer $v$ are the proper $T^*$-descendants of $v$.  See Figure~\ref{fig:trivially-perfect-demonstration} for an illustration of this process.

\begin{lemma}\label{lem:trivially-perfect-solution-map}
  The solution map $M(G)$ constructed above is strongly connected.
\end{lemma}
\begin{proof}
  Let $S^*$ be a fixed destination solution of $G$.  We show that for any solution $S$ of $G$, there is a retaliation-free path from $S$ to $S^*$, and then the statement follows from Theorem~\ref{thm:formal-retaliation-free}.

  We fix a generating tree $T^*$ of $G[S^*]$, and denote by $r^*$ the root of $T^*$.
  We take the post-order ordering of $T^*$,
  and extend it to an ordering of $V(G)$ by numbering vertices in $V(G)\setminus S^*$ arbitrarily from $|S^*| + 1$ to $n$.
  Then $S^* = \{v_1, \ldots, v_{|S^*|}\}$ is the lexicographically smallest among all solutions of $G$.
  We define the victim function $D:\Sigma(S^*) \to \Sigma(S^*)$ as follows.  For each output of $D$, its elements are ordered from the smallest to the largest.   Hence, we only specify the set of elements.
The elements in the sequence $D(\langle\rangle)$ is $S^*$.  
For a sequence $\sigma$ and vertex $v$, elements in $D(\sigma + v)$ are the proper $T^*$-descendants of $v$ if $v\in D(\sigma)$, and empty otherwise.  Note that $D(\sigma)$ is nonempty only if every vertex in $\sigma$ is a $T^*$-ancestors of its successor.

For $k\ge 0$ with $S_k\ne S^*$, denote by $T_k$ the generating tree for $G[S_k]$ used in the successor function.  Let ${\alpha(k)}$ denote the first vertex in $D(\sigma_k)$ that is not in $S_k$. 
We argue that the transition from $S_k$ to $S_{k+1}$ satisfies the definition of retaliation-free paths, and

\begin{enumerate}[($\star$)]
\item for every $v\in \sigma_k$, there are two $T^*$-descendants of different $T^*$-children of $v$ in $S_{k}$.
\end{enumerate}
This invariant holds vacuously for $k = 0$, because $\sigma_0 = \langle\rangle$.  We always use $u = r^*$.

Case 1, ${\alpha(k)}$ is a $T^*$-leaf.
If $\sigma_k= \langle\rangle$, then we take the solution $S_{k+1}$ by (tp0).
Since $S_{k}\setminus S_{k+1}\subseteq N(\alpha(k))$, every vertex in $(S_{k}\setminus S_{k+1})\cap S^*$ is a $T^*$-ancestor of $\alpha(k)$, and is after $\alpha(k)$ in $D(\sigma_k)$.  Moreover, $\sigma_{k+1}$ remains empty, and hence invariant ($\star$) remains vacuously true.

In the rest  $\sigma_k\ne \langle\rangle$.  Denote by $z$ the last vertex of the sequence $\sigma_k$.  By the definition of $D(\sigma_k)$, the vertex $\alpha(k)$ is a proper $T^*$-descendant of $z$.
By the invariant ($\star$), there are two $T^*$-descendants of different $T^*$-children of $z$ in $S_{k}$; let them be $c_1$ and $c_2$.
We take $p$ to be the lowest common $T_k$-ancestor of $c_1$ and $c_2$ that is in $S^*$; it exists because $z\in S^*$ and is a common $T_k$-ancestor of $c_1$ and $c_2$.
By the selection of $p$, a $T_k$-ancestor $p'$ of $p$ is adjacent to both $c_1$ and $c_2$; if $p'$ is in $S^*$, then it has to be a $T^*$-ancestor of $z$, hence adjacent to $\alpha(k)$.
We use the solution $S_{k+1}$ by (tp1) with $v'= p$.  A vertex $x$ in $(S_{k}\setminus S_{k+1})\cap S^*$ has to be a proper $T_k$-descendant of $p$, and adjacent to $\alpha(k)$.  Since $\alpha(k)$ is a leaf of $T^*$, the vertex $x$ is a $T^*$-ancestor of $\alpha(k)$, and has a larger number than $\alpha(k)$.  Thus, $\sigma_{k+1}$ is a prefix of $\sigma_{k}$.  If both $c_1$ and $c_2$ remain in $S_{k+1}$, then they are two $T^*$-descendants of different $T^*$-children of $z$.  If $c_1\not\in S_{k+1}$, then it is a $T^*$-ancestor of $\alpha(k)$, and $\alpha(k)$ and $c_2$ are two $T^*$-descendants of different $T^*$-children of $z$.  It is similar if $c_2\not\in S_{k+1}$.  Otherwise, both $c_1$ and $c_2$ remain in $S_{k+1}$, and the condition in ($\star$) always holds for $z$.  A vertex $y$ in $\sigma_k$ different from $z$ is a $T^*$-ancestor of $z$.  Since $(S_{k}\setminus S_{k+1})\cap S^*$ are all $T^*$-descendants of one $T^*$-child of $y$, the condition in ($\star$) remains true for $y$.

Case 2, $\alpha(k)$ has a single  $T^*$-child $c$.
We take the solution $S_{k+1}$ by (tp2) with $v'= c$.  By the selection of $v'$, a vertex in $(S_{i}\setminus S_{i+1})$ is adjacent to precisely one of $\alpha(k)$ and $v'$, and hence cannot be in $S^*$.  Therefore, $(S_{i}\setminus S_{i+1})\cap S^*=\emptyset$, which implies that $\sigma_{k+1}$ is a prefix of $\sigma_{k}$, and ($\star$) holds for $k + 1$.  

Case 3, $\alpha(k)$ has more than one $T^*$-child.  We take the solution $S_{k+1}$ by (tp3) with $v_1$ and $v_2$ being any two different $T^*$-children of $\alpha(k)$.
By the post-order of the vertices, both $v_1$ and $v_2$ are in $S_k$.
Every common $T_k$-ancestor of $v_1$ and $v_2$ in $S^*$ is a $T^*$-ancestor of $\alpha(k)$.  For $i = 1, 2$, if a vertex $x$ in $S^*$ is a $T_k$-descendant of $v_i$ or a $T_k$-ancestor of $v_i$ that is not a $T_k$-ancestor of $v_{3-i}$, then $x$ is a $T^*$-descendant of $\alpha(k)$.  Therefore, both sets $A$ and $B$ in (tp3) are disjoint from $S^*$, and $(S_{i}\setminus S_{i+1})\cap S^*\subseteq N(\alpha(k))\setminus N(v_1, v_2)$.  Since every $T^*$-ancestor of $\alpha(k)$ is a common $T_k$-ancestor of $v_1$ and $v_2$, every vertex in $(S_{i}\setminus S_{i+1})\cap S^*$ is a $T^*$-descendant of $\alpha(k)$, which is in $D(\sigma_k)$.  We now verify ($\star$) remains true.  Since $(S_{k}\setminus S_{k+1})\cap S^*$ are all $T^*$-descendants of $\alpha(k)$, the condition in ($\star$) remains true for every vertex in $\sigma_k$, which is a $T^*$-ancestor of $\alpha(k)$.  The condition holds for $\alpha(k)$ itself because $v_1$ and $v_2$ are the vertices required by the condition.
\end{proof}

\begin{lemma}
  The maximal induced trivially perfect subgraphs problem
  and the maximal connected induced trivially perfect subgraphs problem
  can be solved with polynomial delay.
\end{lemma}
\begin{proof}
The first result follows from Theorem~\ref{thm:solution-map}, and Lemmas~\ref{lem:tpg-successor-function} and \ref{lem:trivially-perfect-solution-map}.  The second then follows from Proposition~\ref{lem:connected} because the class of trivially perfect graphs is closed under adding universal vertices.
\end{proof}

\subsection{Interval graphs}

A graph is an \emph{interval graph} if its vertices can be assigned to intervals on the real line such that there is an edge between two vertices if and only if their corresponding intervals intersect.
Fulkerson and Gross~\cite{fulkerson-65-interval-graphs} showed that a graph $G$ is an interval graph if and only if its maximal cliques can be arranged in a way that for any $v\in V(G)$, the maximal cliques containing $v$ appear consecutively, and such an arrangement is called a \emph{clique path} of the graph.
Since the class of interval graphs is closed under adding universal vertices, we may focus on the maximal connected induced interval subgraphs problem, which we build a solution map to solve.
To show that the solution map is strongly connected, we fix a clique path for each solution, and show that for any pair of solutions $S$ and $S^*$, there exists a path from $S$ to $S^*$ that can be viewed as ``adding vertices in $S^*\setminus S$ to $S$ from left to right as they appear in the clique path of $G[S^*]$.''

The key difficulty of carrying out this idea, just the same as in the recognition of interval graphs, is that an interval graph may have more than one clique path.  Booth and Lueker~\cite{booth-76-pq-tree} invented the complex PQ-tree data structure to manage the different arrangements of maximal cliques.
Using the notion of modules, Hsu~\cite{hsu-95-recognition-cag} was able to give a similar and simpler characterization.  A subset $U$ of vertices forms a \emph{module} of $G$ if all vertices in $U$ have the same neighborhood outside $U$.  In other words, for any pair of vertices $u,v \in U$, a vertex $x \not\in U$ is adjacent to $u$ if and only if it is adjacent to $v$ as well.  Two adjacent vertices that form a module are called \emph{true twins}; they have the same closed neighborhood.  The set $V(G)$ and all singleton vertex sets are modules, called \emph{trivial}.   A graph on four or more vertices is \emph{prime} if it contains only trivial modules.
  The following observation of Hsu~\cite{hsu-95-recognition-cag} is behind Hsu and Ma's recognition algorithms for interval graphs~\cite{hsu-99-recognizing-interval-graphs}.
  
\begin{theorem}[\cite{hsu-95-recognition-cag}]
  \label{thm:prime-interval-graph}
  A prime interval graph has a unique clique path, up to full reversal.
\end{theorem}
We need a constructive version of Theorem~\ref{thm:prime-interval-graph}.
\begin{lemma}\label{lem:modules}
  Let $G$ be an interval graph with $\ell$ maximal cliques, and let $K_1$, $\ldots$, $K_\ell$ be a clique path of $G$.  If there is another clique path of $G$ of which (1) $K_\ell$ is not one of the ends, and (2) $K_p$ is the end that becomes disconnected from $K_1$ by the removal of $K_\ell$, then

  $$\bigcup_{j\in \{p, \ldots, \ell\}\cup J} K_j \setminus (K_p\cap K_\ell),$$
  is a nontrivial module of $G$, where $J$ is the set of indices $j$ with $K_j$ in between $K_\ell$ and $K_p$ in the second clique path.
\end{lemma}
\begin{proof}
  Let $U = \bigcup_{j\in \{p, \ldots, \ell\}\cup J} K_j \setminus (K_p\cap K_\ell)$, we show that $N(v)\setminus U =  K_p\cap K_\ell$ for every $v\in U$.  Since $v$ is in a clique that is between $K_p$ and $K_\ell$, in at least one of the two clique paths.  It follows from the definition of clique paths that $K_p\cap K_\ell\subseteq N(v)$.  It remains to show that $N(v)\setminus U\subseteq  K_p\cap K_\ell$.  We may assume that $K_\ell$ and $K_p$ are the right ends of the two clique paths.  If every maximal clique $K_i$ containing $v$ has $i\in J\cup \{p, \ldots, \ell\}$, then $N(v)\setminus U\subseteq  K_p\cap K_\ell$.  Suppose that $v$ is contained in a maximal clique $K_i$ with $i\not\in J\cup \{p, \ldots, \ell\}$, then $K_i$ lies to the left of $K_p$ in the first clique path and lies to the left of $K_\ell$ in the second.  Then since $v$ can be found in both sides of $K_p$ in the first clique path, it has to be in $K_p$ as well.  For the same reason, $v\in K_\ell$.  But then $v$ is in $K_p\cap K_\ell$, and should not be in $U$, a contradiction.  This concludes the proof.
\end{proof}

Let $K_1$, $K_2$, $\ldots$, $K_\ell$ be a clique path of an interval graph $G$.  For any nonempty vertex set $U$ of $G$, we may produce a clique path of $G[U]$ as follows.  For $i = 1, \ldots, \ell$, we replace $K_i$ by $K_i\cap U$, (or equivalently, remove vertices in $K_i\setminus U$ from $K_i$,) and then remove those sets that are not maximal cliques of $G[U]$ as well as duplicate ones.
Let $x$ and $y$ be two nonadjacent vertices in an interval graph $G$.  In any clique path of $G$, the group of maximal cliques containing $x$ and those containing $y$ are disjoint, and the first group is either to the left, or to the right of the second.  We observe that in the clique path for the subgraph $G[U]$ obtained as above, a pair of nonadjacent vertices in $U$ have the same order as in the clique path of $G$.

We are now ready to introduce the successor function for the maximal connected induced interval subgraphs problem.  Let $G$ be the input graph.  For each solution $S$ of $G$ and each vertex $v\in V(G)\setminus S$, we define the set $\operatorname{succ}(S, v)$ of successors as follows. 
Let $K_1$, $\ldots$, $K_{\ell}$, denoted by $\mathcal{K}$, be any fixed clique path for $G[S]$.
For the convenience of presentation, we append two empty sets to both ends of the clique path, and refer to them as $K_0$ and $K_{\ell+1}$ respectively.
For each $a$ with $0\le a \le \ell$, we produce a subset of $S$ by removing 
\begin{equation}
  \label{eq:transtion1}
  \tag{\texttt{succ1}}
 K_a\cap K_{a+1}\setminus N(v); \text{ and } N(v)\setminus (K_a\cup K_{a+1}).
\end{equation}
For each pair of $a, b$ with $0\le a < b \le \ell+1$ and $a < b - 1$, we produce a subset of $S$ by removing
\begin{equation}
  \label{eq:transtion2}
  \tag{\texttt{succ2}}
 K_a\cap K_b\setminus N(v); N(v)\setminus K_a; \bigcup^{b-1}_{j=a+1}K_j\setminus (K_a\cup K_b); \text{ and } (K_b\setminus K_a)\cap N(K_a\setminus K_b)
\end{equation}
 from $S$, and another subset by removing 
\begin{equation}
  \label{eq:transtion3}
  \tag{\texttt{succ3}}
 K_a\cap K_b\setminus N(v); N(v)\setminus K_b; \bigcup^{b-1}_{j=a+1}K_j\setminus (K_a\cup K_b); \text{ and }
 (K_a\setminus K_b)\cap N(K_b\setminus K_a).
\end{equation}
from $S$. 
What we try to achieve by these operations is to make space to accommodate the maximal cliques containing $v$ between $K_a$ and $K_b$, where we let $b = a +1$ for \eqref{eq:transtion1}.
For each $S'$ of the $\ell + 1 + \ell (\ell - 1) = O(n^2)$ subsets,
we use Proposition~\ref{lem:extension} to extend this component of $G[S']$ that contains $v$ to a solution of $G$.
These solutions form the set $\operatorname{succ}(S, v)$; note that each solution in $\operatorname{succ}(S, v)$ contains $v$.

\begin{lemma}\label{lem:interval-successor-function}
Let $G$ be the input graph to 
  the maximal connected induced interval subgraphs problem.
  For any solution  $S$ of $G$, and any vertex $v$ in $V(G)\setminus S$, every set in $\operatorname{succ}(S, v)$ is a solution of $G$ containing $v$, and $\operatorname{succ}(S, v)$ can be calculated in $O(n^5)$ time.
\end{lemma}
\begin{proof}
  Let $S'$ denote the subset obtained from $S$ by removing the subsets defined in one of \eqref{eq:transtion1}--\eqref{eq:transtion3}.  For the correctness, by Proposition~\ref{lem:extension}, it suffices to show that $S'\cup \{v\}$ is an interval set of $G$.
  We construct a clique path for $G[S'\cup \{v\}]$ as follows.  We may use $b = a + 1$ for \eqref{eq:transtion1}.
  For $i = 1, \ldots, \ell$, replace $K_i$ by $K'_i = K_i\cap S'$; for $j=a+1, \ldots, b-1$, remove $K'_j$; add two cliques $(N(v)\cap K'_a)\cup \{v\}$ and $(N(v)\cap K'_b)\cup \{v\}$ in between $K'_a$ and $K'_b$; and then remove all those sets that are not maximal cliques of $G[S'\cup \{v\}]$ as well as duplicate ones.  We now prove that the result is a clique path for $G[S'\cup \{v\}]$.

  We show first that every maximal clique $K$ of $G[S'\cup \{v\}]$ is present.  The only maximal cliques containing $v$ are $(N(v)\cap K'_a)\cup \{v\}$ and $(N(v)\cap K'_b)\cup \{v\}$, possibly only one of them.  If $v\not\in K$, then $K$ is a subset of some maximal clique of $G[S]$, hence $K'_j$ for some $j, 1\le j \le \ell$.  We only need to worry when $a < j < b$, then $K'_j$ has been removed.  This cannot happen for \eqref{eq:transtion1}.  For \eqref{eq:transtion2}, if $K'_j$ is a subset of neither $K'_a$ nor $K'_b$, then there must be some vertex $x\in K_a\setminus K_b$ in $K'_j$.  But then all vertices in $K'_j\cap K_b\setminus K_a$ have been removed, and $K'_j\subseteq K'_a$, a contradiction.  The argument is similar for \eqref{eq:transtion3}.

  We then argue that for every vertex $x\in S'\cup \{v\}$, maximal cliques of $G[S'\cup \{v\}]$ containing $x$ appear consecutively, for which it suffices to consider those vertices in $K'_a\cup K_b'$ and $v$ itself.  The only maximal cliques containing $v$ are $(N(v)\cap K'_a)\cup \{v\}$ and $(N(v)\cap K'_b)\cup \{v\}$, (one of which might be a subset of the other,) hence together.  For each vertex $x$ in $K'_a\cap K'_b$, there are $a' \le a$ and $b' \ge b$ such that $K_{a'}, \ldots, K_{b'}$ are all the maximal cliques of $G[S]$ containing $x$.  On the other hand, $x$ has to be in $N(v)$, and hence both $(N(v)\cap K'_a)\cup \{v\}$ and $(N(v)\cap K'_b)\cup \{v\}$ contain it.  Thus, the maximal cliques of $G[S'\cup \{v\}]$ containing $x$ are consecutive.  For a vertex $x\in K'_a\setminus K'_b$,
  the maximal cliques of $G[S'\cup \{v\}]$ containing $x$ are $K'_{a''}, \ldots, K'_a$, for some $a'' \le a$, together with possibly $(N(v)\cap K'_a)\cup \{v\}$, hence consecutive.  It is symmetric for vertices in $K'_b\setminus K'_a$.
  
  For each pair of indices $a$ and $b$, we can obtain $S'$ in time $O(n^2)$, and extend $S'$ to a solution in $O(n^3)$ time.  The total time of producing $\operatorname{succ}(S, v)$ is thus  $O(n^2\cdot (n^2 + n^3)) = O(n^5)$.
\end{proof}

For each solution $S'$ in $\operatorname{succ}(S, v)$, we add an arc from $S$ to $S'$, and label it with $v$.  We use the resulting multiple digraph, denoted by $M(G)$, as the solution map for the graph $G$.  By Lemma~\ref{lem:interval-successor-function}, the out-degree of a node in $M(G)$ is at most $O(n^3)$.  In the rest we show that $M(G)$ is strongly connected by demonstrating a retaliation-free path from any solution $S$ to another solution $S^*$.  
Suppose that we fix a clique path for $G[S^*]$, and let $v^*$ be a vertex in $S^*\setminus S$ that is leftmost in this clique path, and let $\widehat S$ denote those vertices in $S\cap S^*$ that are to the left of $v^*$.
  Then the set of neighbors of $v^*$ in $\widehat S$, if nonempty, is a clique.
If in the clique path for $G[S]$ used in the calculation of $\operatorname{succ}(S, v^*)$, maximal cliques that are supersets of $N(v^*)\cap\widehat S$ appear in one end of $\widehat S$, then one successor of $S$ contains $\widehat S\cup \{v^*\}$.  
Otherwise, there may be victims.
According to Theorem~\ref{thm:prime-interval-graph}, we are always in the easy case if $G[\widehat S]$ is prime.  If there are victims, all of them must be contained in some nontrivial module of $G[\widehat S]$ by Lemma~\ref{lem:modules}.  What we try to do is to find the two borders of this module on the clique path for $G[S^*]$, where vertices inside are the vertices for $D(\langle v^*\rangle)$, and this gives us the required information to build a retaliation-free path.  See Figure~\ref{fig:interval-demonstration} for an illustration of this process.

\begin{figure}[h!]
  \centering
  \begin{subfigure}[b]{.8\linewidth}
    \centering\scriptsize
    \begin{tikzpicture}[scale=1.2]
      \begin{scope}[every path/.style={{|[right]}-{|[left]}}]
        \foreach \x in {1, ..., 11} \node[gray] at (\x-1, -.15) {$\x$};
        \foreach[count=\i] \l/\r/\y/\c in {-.02/.02/4/, 0/1/3/,
          0.98/1.02/4/, 1/8/1/, 1/9/0/,
          1.98/2.02/4/, 1.98/2.02/3/,
          2.98/3.02/4/, 3/4/3/,3.98/4.02/4/,
          4.98/5.02/5/, 5/6/4/,
5/7/3/, 5/8/2/, 6/7/5/, 6.98/7.02/4/, 8/9/3/, 
          9/10/2/, 9.98/10.02/3/}{
          \draw[\c, thick] (\l-.02, \y/5) node[left] {$\i$} to (\r+.02, \y/5);
        }
      \end{scope}
    \end{tikzpicture}
    \caption{Solution $S^*$.}
  \end{subfigure}  
  
  \begin{subfigure}[b]{.8\linewidth}
    \centering\scriptsize
    \begin{tikzpicture}[scale=1.2]
      \begin{scope}[every path/.style={{|[right]}-{|[left]}}]
        \foreach \x in {1, ..., 11} \node[gray] at (\x-1, .05) {$\x$};
        \foreach[count=\i] \l/\r/\y/\c in {-.02/.02/4/, 0/1/3/,
          0.98/1.02/4/, 1/10/1/white, 1/8/1/white,
          1.98/2.02/4/, 1.98/2.02/3/,
          3.98/4.02/4/, 3/4/3/, 2.98/3.02/4/, 
          6.98/7.02/3/, 6/7/4/, 5/6/3/white, 5/7/2/, 5/6/3/, 4.98/5.02/4/, 7.08/7.10/3/white,
        8/9/2/, 9/10/3/}{
          \draw[\c,thick] (\l-.02, \y/5) node[left] {$\i$} to (\r+.02, \y/5);
        }
        \foreach \l/\r/\y/\i in {9.98/10.02/2/20, 1/8/1/24}{
          \draw[violet, thick] (\l-.02, \y/5) node[left] {$\i$} to (\r+.02, \y/5);
        }
      \end{scope}
    \end{tikzpicture}
    \caption{$S_0$; $\sigma_0 = \langle \rangle$; $S^*\setminus S_0 = \{4, 5, 13\}$; $\alpha(0) = 4$.}
  \end{subfigure}

  \begin{subfigure}[b]{.8\linewidth}
    \centering\scriptsize
    \begin{tikzpicture}[scale=1.2]
      \begin{scope}[every path/.style={{|[right]}-{|[left]}}]
        \foreach \x in {1, ..., 11} \node[gray] at (\x-1, -.15) {$\x$};
        \foreach[count=\i] \l/\r/\y/\c in {-.02/.02/4/black, 0/1/3/black,
          0.98/1.02/4/black, 1/10/1/black, 1/9/0/black,
          6.98/7.02/4/black, 6.98/7.02/3/black,
          9/10/4/black, 8/9/3/black, 7.98/8.02/4/black, 
          4.98/5.02/5/black, 4/5/4/black, 3/6/3/black, 2/5/2/black, 3/4/5/black, 2.98/3.02/4/black, 1.98/2.02/3/black}{
          \draw[\c,thick] (\l-.02, \y/5) node[left] {$\i$} to (\r+.02, \y/5);
        }
        \foreach[count=\i from 20] \l/\r/\y/\c in {5.98/6.02/2/violet, 9.98/10.02/3/violet}{
          \draw[violet, thick] (\l-.02, \y/5) node[left] {$\i$} to (\r+.02, \y/5);
        }
      \end{scope}
    \node at (4, 1.25) {};
    \end{tikzpicture}
    \caption{$S_1$; $\sigma_1 = \langle \rangle$; $S^*\setminus S_1 = \{18, 19\}$; $\alpha(1) = 18$.}
  \end{subfigure}

\begin{subfigure}[b]{.75\linewidth}
    \centering\scriptsize
    \begin{tikzpicture}[scale=1.2]
      \begin{scope}[every path/.style={{|[right]}-{|[left]}}]
        \foreach \x in {1, ..., 10} \node[gray] at (\x-1, -.15) {$\x$};
        \foreach[count=\i] \l/\r/\y/\c in {-.02/.02/3/, 0/1/2/,
          0.98/1.02/3/, 1/7/1/, 1/9/0/,
          3.98/4.02/3/, 3.98/4.02/2/,
          2.98/3.02/4/white, 1.98/2.02/3/, 2/3/2/, 
          4.98/5.02/.5/white, 5.98/6.02/4/, 
          5/6/2/, 6/7/3/, 5.98/6.02/5/, 6.98/7.02/4/white, 7/9/2/, 
          8.98/9.02/3/violet, 9.98/10.02/2/white}{
          \draw[thick, \c] (\l-.02, \y/5) node[left] {$\i$} to (\r+.02, \y/5);
        }
        \foreach \l/\r/\y/\i in {
          4.98/5.02/3/20, 2.98/3.02/3/22, 7.98/8.02/3/23}{
          \draw[violet, thick] (\l-.02, \y/5) node[left] {$\i$} to (\r+.02, \y/5);
        }
      \end{scope}
    \end{tikzpicture}
    \caption{$S_2$; $\sigma_2 = \langle 18 \rangle$; $S^*\setminus S_2 = \{8, 11, 16, 19\}$; $\alpha(2) = 16$.}
  \end{subfigure}  

  \begin{subfigure}[b]{.75\linewidth}
    \centering\scriptsize
    \begin{tikzpicture}[scale=1.2]
      \begin{scope}[every path/.style={{|[right]}-{|[left]}}]
        \foreach \x in {1, ..., 10} \node[gray] at (\x-1, -.15) {$\x$};
        \foreach[count=\i] \l/\r/\y/\c in {-.02/.02/4/, 0/1/3/,
          0.98/1.02/4/, 1/8/1/, 1/9/0/,
          1.98/2.02/4/, 1.98/2.02/3/,
          2.98/3.02/4/white, 2.98/3.02/4/, 3/4/3/, 
          6.98/7.02/5/white, 7/8/4/, 5/7/3/, 6/8/2/, 6/7/5/, 5.98/6.02/4/, 8/9/3/, 
          8.98/9.02/2/, 10.98/11.02/4/white}{
          \draw[thick, \c] (\l-.02, \y/5) node[left] {$\i$} to (\r+.02, \y/5);
        }
        \foreach \l/\r/\y/\i in {4.98/5.02/4/20, 3.98/4.02/4/22}{
          \draw[violet, thick] (\l-.02, \y/5) node[left] {$\i$} to (\r+.02, \y/5);
        }
      \end{scope}
      \node at (4, 1.25) {};
    \end{tikzpicture}
    \caption{$S_3$; $\sigma_3 = \langle 18 \rangle$; $S^*\setminus S_3 = \{8, 11, 9\}$; $\alpha(3) = 11$.}
  \end{subfigure}  

  \begin{subfigure}[b]{.85\linewidth}
    \centering\scriptsize
    \begin{tikzpicture}[scale=1.2]
      \begin{scope}[every path/.style={{|[right]}-{|[left]}}]
        \foreach \x in {1, ..., 12} \node[gray] at (\x-1, -.15) {$\x$};
        \foreach[count=\i] \l/\r/\y/\c in {-.02/.02/4/, 0/1/3/,
          0.98/1.02/4/, 1/8/1/, 1/10/0/,
          1.98/2.02/4/, 1.98/2.02/3/,
          2.98/3.02/4/white, 2.98/3.02/4/, 3/4/3/, 
          5/6/5/, 6/7/4/,
          6/7/3/, 6/8/2/, 6.98/7.02/5/, 7.98/8.02/4/white, 8/10/3/, 
          10/11/2/, 10.98/11.02/4/}{
          \draw[thick, \c] (\l-.02, \y/5) node[left] {$\i$} to (\r+.02, \y/5);
        }
        \foreach \l/\r/\y/\i in {
          4/5/4/22, 8.98/9.02/2/23}{
          \draw[violet, thick] (\l-.02, \y/5) node[left] {$\i$} to (\r+.02, \y/5);
        }
      \end{scope}
    \node at (4, 1.25) {};
    \end{tikzpicture}
    \caption{$S_4$; $\sigma_4 = \langle 18, 11 \rangle$; $S^*\setminus S_4 = \{8, 16\}$; $\alpha(4) = 16$.}
  \end{subfigure}  

  \begin{subfigure}[b]{.85\linewidth}
    \centering\scriptsize
    \begin{tikzpicture}[scale=1.2]
      \begin{scope}[every path/.style={{|[right]}-{|[left]}}]
        \foreach \x in {1, ..., 12} \node[gray] at (\x-1, -.15) {$\x$};
        \foreach[count=\i] \l/\r/\y/\c in {-.02/.02/4/, 0/1/3/,
          0.98/1.02/4/, 1/9/1/, 1/10/0/,
          1.98/2.02/4/, 1.98/2.02/3/,
          2.98/3.02/4/white, 2.98/3.02/4/, 3/4/3/, 
          5/6/5/, 6/7/4/,
          6/8/3/, 6/9/2/, 7/8/5/, 7.98/8.02/4/, 9/10/3/, 
          10/11/2/, 10.98/11.02/4/}{
          \draw[thick, \c] (\l-.02, \y/5) node[left] {$\i$} to (\r+.02, \y/5);
        }
        \foreach \l/\r/\y/\i in {
          4/5/4/22}{
          \draw[violet, thick] (\l-.02, \y/5) node[left] {$\i$} to (\r+.02, \y/5);
        }
      \end{scope}
    \node at (4, 1.25) {};
    \end{tikzpicture}
    \caption{$S_5$; $\sigma_5 = \langle 18 \rangle$; $S^*\setminus S_5 = \{8\}$; $\alpha(5) = 8$.}
  \end{subfigure}  
  
  \caption{Demonstration of a retaliation-free path for the maximal connected induced interval subgraphs problem.  The solutions, $S^*$, $S_0$, $\ldots$, $S_5$, of $G$ are presented as clique paths (interval models) in (a)--(g).  The graph contains 23 vertices, of which the first 19 are numbered in (a).  The neighborhoods of other vertices are as follows: $N(20) = \{4,5, 13, 19\}$; $N(21) = \{4, 8\}$; $N(22) = \{4, 5, 8, 10, 11\}$; $N(23) = \{5, 16, 17\}$; and $N(24) = \{1,2, \ldots, 18\}$.  Finally, $S_6 = S^*$.}  
  \label{fig:interval-demonstration}
\end{figure}
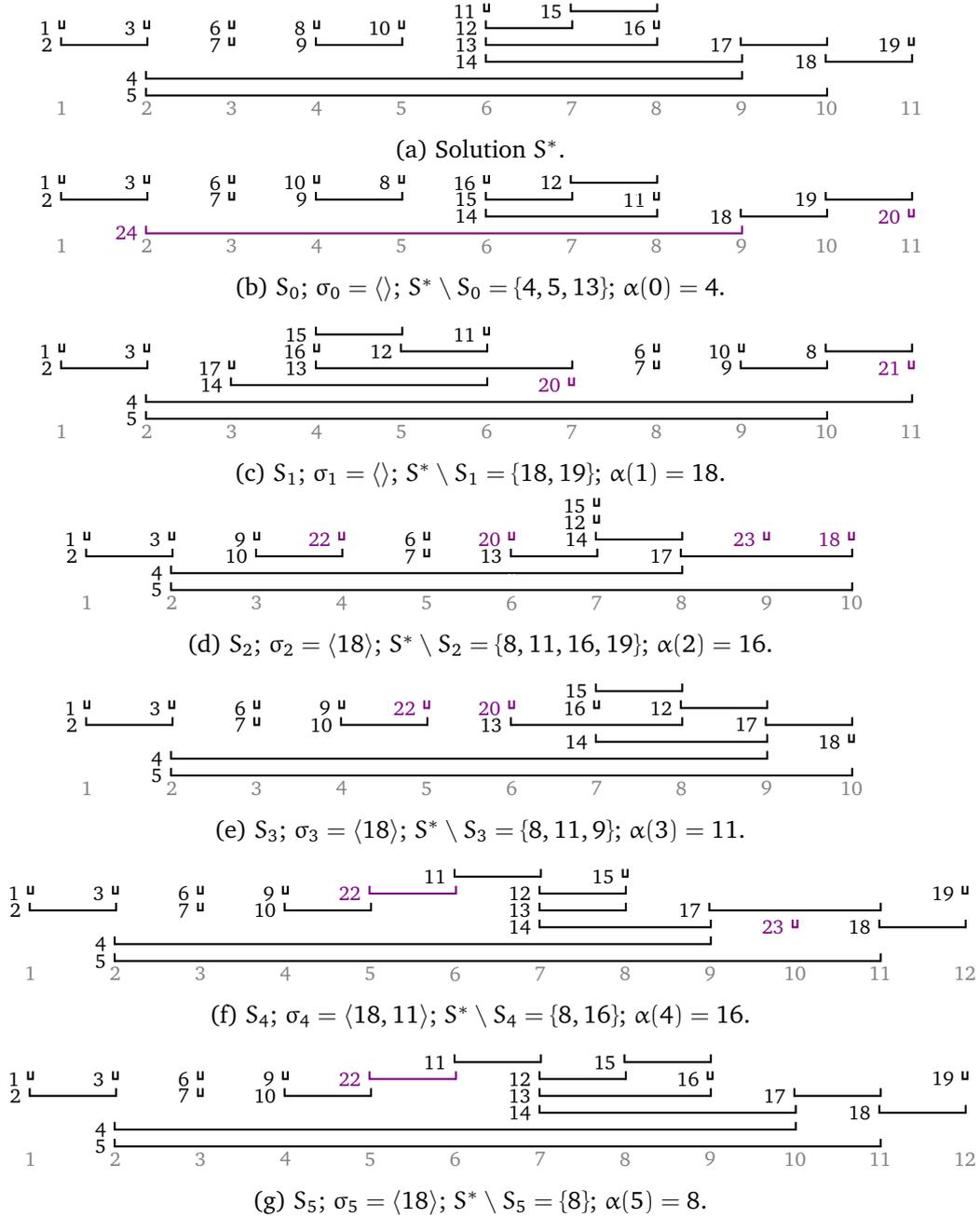

\begin{lemma}\label{lem:interval-retaliation-path}
  The solution map $M(G)$ is strongly connected.
\end{lemma}
\begin{proof}
  Let $S^*$ be a fixed destination solution of $G$.  We show that for any solution $S$ of $G$, there is a retaliation-free path from $S$ to $S^*$, and then the statement follows from Theorem~\ref{thm:formal-retaliation-free}.
  
  We fix a clique path, ${K}_1$, $\ldots$, ${K}_{\ell^*}$ for $G[S^*]$, where ${\ell^*}$ is the number of maximal cliques of $G[S^*]$.  We may add two empty sets $K_0$ and $K_{\ell^*+1}$ to the two ends respectively.
For each $v\in S^*$, let
\[
  \operatorname{lp}(v) = \min\{j\mid v\in K_j\} \text{ and }\operatorname{rp}(v) = \max\{j\mid v\in K_j\}.
\]
Note that $\{ [\operatorname{lp}(v), \operatorname{rp}(v)] \mid v\in S^*\}$ is an interval model for $G[S^*]$.
It is easy to use definition to verify that ${K}_{\ell^*+1}$, ${K}_{\ell^*}$, $\ldots$, ${K}_1$, ${K}_0$ is also a clique path for $G[S^*]$.  To distinguish them, we use $\overrightarrow{\mathcal{K}}$ to denote the path ${K}_0$, ${K}_1$, $\ldots$, ${K}_{\ell^*}$, ${K}_{\ell^*+1}$, and $\overleftarrow{\mathcal{K}}$ its full reversal.  Note that  $\overrightarrow{\mathcal{K}}$ and $\overleftarrow{\mathcal{K}}$ are different unless $S^*$ is a clique.
Moreover, the $i$th clique on $\overleftarrow{\mathcal{K}}$ is the $(\ell^*+1-i)$th clique on $\overrightarrow{\mathcal{K}}$.

For two nonadjacent vertices $u$ and $v$ in $S^*$, let $B(u, v)$ denote the set of vertices between $u$ and $v$ in $\mathcal{K}$, i.e.,
  \[
    B(u, v)=
    \begin{cases}
      \bigcup_{j=\operatorname{rp}(u)+1}^{\operatorname{lp}(v)-1} K_j
      & \text{if } \operatorname{rp}(u) < \operatorname{lp}(v),
      \\
      \bigcup_{j=\operatorname{rp}(v)+1}^{\operatorname{lp}(u)-1} K_j
      & \text{if } \operatorname{rp}(v) < \operatorname{lp}(u).
    \end{cases}
  \]

We are ready to define the function $D:\Sigma(S^*) \to \Sigma(S^*)$.
For the convenience of presentation, we introduce two dummy vertices $v_\dashv$ and $v_\vdash$, and define $\operatorname{lp^*}(v_\vdash) = \operatorname{rp^*}(v_\vdash) = 0$ and $\operatorname{lp^*}(v_\dashv) = \operatorname{rp^*}(v_\dashv) = \ell^*+1$.
For each input $\sigma$ of $D$, we add both dummy vertices at its beginning; in particular, $v_\dashv$ the first and $v_\vdash$ the second.  Hence the empty sequence will be interpreted as $\langle v_\dashv v_\vdash\rangle$.
The elements in the sequence $D(\langle  v_\dashv v_\vdash\rangle)$ is $S^*$, and they are ordered from left to right as they appear in the clique path $\overrightarrow{\mathcal{K}}$: A vertex $u\in S^*$ is before another $v\in S^*$ in $D(\langle  v_\dashv v_\vdash\rangle)$ if (1) $\operatorname{lp}(u)< \operatorname{lp}(v)$, or (2) $\operatorname{lp}(u)= \operatorname{lp}(v)$ and $\operatorname{rp}(u)< \operatorname{rp}(v)$, and true twins of $G[S^*]$ are ordered arbitrarily.
We define $D(\sigma+v)$ recursively as follows.  Let $z$ be the last vertex of $\sigma$.  If $v\not\in S^*\setminus N(z)$, (we consider $N(\vdash)$ to be empty,) or $v$ is not in the sequence $D(\sigma)$, then $D(\sigma+v)$ is defined to be the empty sequence.  Otherwise, we define $D(\sigma+v)$ to be vertices in $B(z, v)$ with the following ordering.  If the ordering of $D(\sigma)$ is from left to right on $\overrightarrow{\mathcal{K}}$, then the ordering for $D(\sigma+v)$ is from left to right on $\overleftarrow{\mathcal{K}}$, with true twins of $G[S^*]$ ordered arbitrarily, and vice versa.
Note that $u$ is to the left of $v$ in $\overleftarrow{\mathcal{K}}$ if (1) $\operatorname{rp}(u)> \operatorname{rp}(v)$, or (2) $\operatorname{rp}(u)= \operatorname{rp}(v)$ and $\operatorname{lp}(u)> \operatorname{lp}(v)$.  %
One should note that the ordering for $D(\sigma+v)$ is not the reverse of the same set of vertices in $D(\sigma)$.  Indeed, for two vertices $u$ and $v$ both appearing in $D(\sigma+v)$, it is possible that $u$ is before $v$ in both  $D(\sigma)$ and $D(\sigma+v)$, either when they are true twins, or when $\operatorname{lp}(u)< \operatorname{lp}(v)< \operatorname{rp}(v)< \operatorname{rp}(u)$.
We remark that $S^*$ can be viewed as $B(v_\vdash, v_\dashv)$.
Abusing notation, we may sometimes treat $D(\sigma)$ as a set, by which we mean the set of vertices that are in the sequence of $D(\sigma)$.

If $K_1\not\subseteq S$, then the first vertex $v$ of $D(\langle v_\dashv v_\vdash\rangle)$ that is not in $S$ has to be in $K_1\setminus S$.
The successor obtained by \eqref{eq:transtion1} with $a$ being any number such that $K_1\cap S\subseteq K_a$ contains $v$ and all vertices in $K_1\cap S$.  Moreover, $\sigma_1 = \sigma_0 = \langle v_\dashv v_\vdash\rangle$.  Therefore, we may assume that $K_1\subseteq S$.  
Let $S_0 = S$.

  For $i\ge 0$ with $S_i\ne S^*$, we make the following definitions.
  Let $z$ and $y$ be the last two vertices in $\sigma_i$; i.e.,  $\sigma_i = \langle \cdots y z\rangle$.  We choose the clique path $\overrightarrow{\mathcal{K}}$ or $\overleftarrow{\mathcal{K}}$ in which $z$ is to the left of $y$, and let it be denoted by ${\mathcal{K}^*}$.  (One may note that we are using $\overrightarrow{\mathcal{K}}$ if and only if $\sigma_i$ has an even length, disregarding whether the two dummy vertices are counted.)  For a vertex $v\in S^*$, the values of $\operatorname{lp^*}(v)$ and $\operatorname{rp^*}(v)$ are respectively, the smallest number $j$ and, respectively, the largest number $j$ such that the $j$th clique of ${\mathcal{K}^*}$ contains $v$.  Note that $\operatorname{lp^*}(v)$ is either $\operatorname{lp}(v)$ or $\ell^* + 1 - \operatorname{lp}(v)$.
  Denote by ${\alpha(i)}$ the first vertex of $D(\sigma_i)$ that is not in $S_i$, which exists because $S_i \ne S^*$ and the calculation of $\sigma_{i+1}$. Let
  \[
    r(i) = \operatorname{rp^*}(z) + 1; t_i = \operatorname{lp^*}(y); \text{ and } s(i) = \operatorname{lp^*}({\alpha(i)}).
  \]
 Let $K^i_1, \ldots, K^i_\ell$ be the clique path of $G[S_i]$ used in the successor function.

 We prove a stronger statement than required by the definition of retaliation-free paths.  We set $A_0=\emptyset$.  We prove that for every $i\ge 0$, there exist a successor $S_{i +1}$ in $\operatorname{succ}(S_i, {\alpha(i)})$ and a set $A_{i+1}\subseteq K^*_{r(i+1)}$
 such that the following invariants are maintained

  \begin{enumerate}[(i)]
  \item $K^*_{r(i)}\subseteq S_{i}$;
  \item if $\sigma_{i}\ne \langle v_\dashv v_\vdash\rangle$, then $A_{i+1}\not\subseteq N({\alpha(i)})$;
  \item if $\sigma_{i}= \langle v_\dashv v_\vdash\rangle$, then $(S_{i}\setminus S_{i+1})\cap S^*\subseteq D(\sigma_i)\setminus K^*_{r(i)}( = S^*\setminus K_1)$; 
  \item if $\sigma_{i}\ne \langle v_\dashv v_\vdash\rangle$, then $(S_{i}\setminus S_{i+1})\cap S^*\subseteq \bigcup_{j=r(i)+1}^{t'(i) - 1} K^*_j\setminus (K^*_{r(i)}\cup K^*_{t'(i)})$; and 
  \item  if $\sigma_{i}\ne \langle v_\dashv v_\vdash\rangle$, then there exists $t'(i)$ with $s(i) < t'(i)< t(i)$ such that $A_{i-1}\subseteq K^*_{t'(i)-1}\cap K^*_{t'(i)} \subseteq A_{i}$.
  \end{enumerate}
  Recall that $\sigma_{i+1}$ to be the longest prefix $\sigma'$ of $\sigma_{i} + {\alpha(i)}$ such that $D(\sigma')\not\subseteq S_{i+1}$.

  We may assume that these invariants hold for $0, \ldots, i-1$, and we show that they hold for $i$ as well.  In the base case, $i=0$, invariant (i) is the assumption $K_1\subseteq S_0$, and other invariants hold vacuously.
  Note that $r(i) <s(i) < t(i)$ because ${\alpha(i)}$ is in $D(\sigma_i)$ and the invariant $K^*_{r(i)}\subseteq S_i$.  Moreover, from $A_i \subseteq K^*_{r(i)}$ and $A_i \subseteq K^*_{t'(i)}$ it can be inferred that $A_i \subseteq K^*_{j}$ for all $j, r(i)\le j\le t'(i)$.

  Let us deal with the simple cases, when $D(\sigma_i + {\alpha(i)})\subseteq S_{i+1}$, i.e., there are no victims.  If the condition in invariant (iii) or (iv), depending on whether $\sigma_{i}= \langle v_\dashv v_\vdash\rangle$ is true, and every vertex in $(S_{i}\setminus S_{i+1})\cap S^*$ is after ${\alpha(i)}$ in the sequence $D(\sigma_i)$, then $\sigma_{i+1}$ is a prefix of $\sigma_i$.  We find the largest $k < i$ such that $\sigma_k = \sigma_i$; it exists because the length of $\sigma_{j+1}$ can be at most one plus that of $\sigma_{j}$ for all $j$.  We set $A_{i+1} = A_k$ and $t'(i+1) = t'(k)$.
  Invariant (i) follows from invariants (iii) and (iv) for $j=k, \ldots, i-1$; note that $r(i+1) = r(k)$, and either $r(i+1)\le r(j)$, or $r(i+1)\ge t(j)$.
   Invariants (ii) and (v) follow from that the same invariant holds in step $k$.  Invariant (iii) or (iv) holds by assumption.
 
  In the first simple case, $(S_{i}\setminus S_{i+1})\cap S^*$ is not empty; let $v'$ be any vertex in this set.
  We claim that the successor obtained by \eqref{eq:transtion1} with any $a$ satisfying $K^*_{s(i)}\cap S_i\subseteq K^i_a$ has no victims.   We have nothing to show if $(S_{i}\setminus S_{i+1})\cap S^*$ is empty.  Now suppose that $x$ is any vertex in $(S_{i}\setminus S_{i+1})\cap S^*$.  By \eqref{eq:transtion1}, $x$ is adjacent to exactly one of $v'$ and ${\alpha(i)}$.  Hence, it is not in $A_i$.  From $\operatorname{lp^*}(v') = \operatorname{lp^*}({\alpha(i)}) = s(i)$ it follows that $\operatorname{lp^*}(x) > s(i)$; moreover, since $x$ is adjacent to $v'$ or ${\alpha(i)}$, we have $\operatorname{rp^*}(x) < t'(i)$ when $\sigma_{i}\ne \langle v_\dashv v_\vdash\rangle$.
Hence, $x$ is after ${\alpha(i)}$ in $D(\sigma_i)$.

  Henceforth, we may assume that $(S_{i}\setminus S_{i+1})\cap S^*$, then $N({\alpha(i)})\cap K^*_{s(i) - 1} = K^*_{s(i)} \cap S_i$.  As consequences, ${\alpha(i)}$ is the first vertex of $K^*_{s(i)}\setminus K^*_{s(i)-1}$ in the sequence $D(\sigma_i)$; and a vertex $x$ is before ${\alpha(i)}$ in $D(\sigma_i)$ if and only if it is in $B(z, {\alpha(i)})$.
  Let $p, q$ be the numbers such that maximal cliques of $G[S_i]$ that are supersets of $N({\alpha(i)})\cap K^*_{s(i) - 1}$ are precisely $K^i_p, \ldots, K^i_q$.

  The second simple case is when there exists $q' > q$ such that $A_i\subseteq (K^i_{q'-1}\cap K^i_{q'})\cap B(z, {\alpha(i)})\subseteq N({\alpha(i)})$, and
  $K^i_j\setminus K^i_q$ is disjoint from $B(z, {\alpha(i)})$ for all $j$ with $q <j < q'$.  Then we take the successor $S_{i+1}$ obtained by \eqref{eq:transtion2} with $a = q$ and $b = q'$.  Since $q \le q'-1$, we have $(K^i_{q}\cap K^i_{q'})\setminus N({\alpha(i)}) \subseteq (K^i_{q'-1}\cap K^i_{q'}) \setminus N({\alpha(i)})$, and it is disjoint from $B(z, {\alpha(i)})$.  Since $K^i_q$ is a superset of $N({\alpha(i)})\cap K^*_{s(i) - 1}$, the set $N({\alpha(i)})\setminus K^i_q$ is disjoint from $B(z, {\alpha(i)})$.  Finally, by assumption, $K^i_j\setminus K^i_q$ is disjoint from $B(z, {\alpha(i)})$ for all $j$ with $q <j < q'$; note that a vertex in $(K_{q'}\setminus K_q)\cap N(K_q\setminus K_{q'})$ is in $K^i_j$ for some $j, q <j < q'$.  Therefore, the new solution $S_{i+1}$ satisfies that $(S_{i}\setminus S_{i+1})\cap S^* \subseteq D(\sigma_i)\setminus B(z, {\alpha(i)})$.  It is symmetric and similar if there exists $p' < p$ such that $A_i\subseteq (K^i_{p'}\cap K^i_{p'+1})\cap B(z, {\alpha(i)})\subseteq N({\alpha(i)})$ and $K^i_j\setminus K^i_p$ is disjoint from $B(z, {\alpha(i)})$ for all $j$ with $p' <j < p$.  Then we use the solution obtained by \eqref{eq:transtion3} with $a = p'$ and $b = p$.  In the rest we may assume that neither of them is true.  
 In particular, neither $(K^i_{p-1}\cap K^i_{p})\cap B(z, {\alpha(i)})$ nor $(K^i_{q}\cap K^i_{q+1})\cap B(z, {\alpha(i)})$ is a subset of $N({\alpha(i)})$.

 Let $p'$ be the smallest number such that $K^i_{p'}\cap K^i_{p}\cap B(z, {\alpha(i)})\setminus N({\alpha(i)})$ is not empty,
 and
let $q'$ be the largest number such that $K^i_{q'}\cap K^i_{q}\cap B(z, {\alpha(i)})\setminus N({\alpha(i)})$ is not empty.  Then $p' < p$ and $q' > q$.
We also find the smallest number $p''$ such that $K^i_{p''}\cap K^i_{p'}\cap B(z, {\alpha(i)})\setminus N({\alpha(i)})$ is not empty, and the largest number $q''$ such that $K^i_{q''}\cap K^i_{q'}\cap B(z, {\alpha(i)})\setminus N({\alpha(i)})$ is not empty.  Since $p''=p'$ and $q''=q'$ satisfy the conditions respectively, $p'' \le p'$ and $q'' \ge q'$; on the other hand, it is possible that one or both of $p'' \le p'$ and $q'' \ge q'$ hold with equality.  We argue that at least one of the following is true:
\begin{enumerate}[(I)]
\item no vertex in $B(z, {\alpha(i)})$ is in $K^i_{j}\setminus K^i_{p'}$ for any $j, p''\le j < p'$; and
\item no vertex in $B(z, {\alpha(i)})$ is in $K^i_{j}\setminus K^i_{q'}$ for any $j, q' < j \le q''$.
\end{enumerate}
  Note that (1) holds vacuously when $p'' = p'$ and  (2) holds vacuously when $q'' = q'$.
Suppose that neither is true, then $p'' < p'$ and $q'' > q'$.
We can find $x\in K^i_{p'}\cap K^i_{p}\cap B(z, {\alpha(i)})\setminus N({\alpha(i)})$,  $x'\in K^i_{p''}\cap K^i_{p'}\cap B(z, {\alpha(i)})\setminus N({\alpha(i)})$, and $x''\in B(z, {\alpha(i)})\cap (K^i_{j}\setminus K^i_{p'})$ for any $j, p''\le j < p'$.  Let $c$ be any vertex in $N({\alpha(i)})\cap K^*_{s(i) - 1}$, then $c$ is adjacent to $x$.  Since ${\alpha(i)}$ is not adjacent to $x$, it cannot be adjacent to $x'$ or $x''$; otherwise there is a hole in $G[S^*]$.   Moreover, by the selection of $p'$ and $p''$, neither of $x'$ and $x''$ is adjacent to $c$.  Likewise, we can find vertices $w, w'$, and $w''$, where $w$ is adjacent to $c$ but not ${\alpha(i)}$ and neither of $y'$ and $y''$ is adjacent to $c$.  Note that there is no edge between $x', x''$ and $w', w''$.
 In any clique path for $G[S^*]$, the maximal cliques containing ${\alpha(i)}$ have to be in between those containing $x'$ and  those containing $w'$.  This is however not true for $\mathcal{K}^*$.  This contradiction means that at least one of (I) and (II) is true.

 In the final case, we have victims.
 We may assume without loss of generality that (II) is true if only one of them holds true; i.e., no vertex in $B(z, {\alpha(i)})$ is in $K^i_{j}\setminus K^i_{q'}$ for any $j, q' < j \le q''$.
 If both (I) and (II) hold true, we assume without loss of generality that maximal cliques that are supersets of $K^*_{r(i)}$ are to the left of $K^i_p$.
 We take the solution $S_{i+1}$ in $\operatorname{succ}(S_i, {\alpha(i)})$ by \eqref{eq:transtion2} with $a = q$ and $b = q' + 1$.
  Since $(z, {\alpha(0)})$ is a subset of both $S^*$ and $S_i$, from $\mathcal{K}^*$ and $\mathcal{K}^i$ we can derive two clique paths for $G[B(z, {\alpha(0)})]$.  The maximal clique $K^*_{s(i) - 1}$ is an end in one of them but not the other.  Therefore, by Lemma~\ref{lem:modules}, we can find a nontrivial module $U$ of $G[B(z, {\alpha(0)})]$ that contains all the vertices in $B(z, {\alpha(0)})\setminus S_{i+1}$; note that $K^*_{r(i)}$ is not a subset of this module.
     Let $r'$ be the largest number such that $U \subseteq \bigcup_{j=r'+1}^{s(i)-1}$.
Then $r(i) \le r'(i) < s(i)$.
 We let $A_{i+1} = K^*_{r'+1}\cap K^*_{s(i)-1}$; note that $A_i \subset A_{i+1}$ because $K^i_{p'}\cap K^i_{p}\cap B(z, {\alpha(i)})\setminus N({\alpha(i)})\subseteq A_{i+1}$.

 It remains to verify that the invariants hold true.
 Since $\sigma_{i+1} = \sigma_i + {\alpha(i)}$, the clique path for step ${i+1}$ is the full reversal of $\mathcal{K}^*$.
 First, the $({r(i+1)})$th clique of the clique path for step ${i+1}$ is $K^*_{s(i)-1}$, hence a subset of $B(z, {\alpha(i)})$.  Second, $A_i \subset A_{i+1}\not\subseteq N({\alpha(i)})$ because $K^i_{p'}\cap K^i_{p}\cap B(z, {\alpha(i)})\setminus N({\alpha(i)})\subseteq A_{i+1}\setminus N({\alpha(i)})$.
 Invariants (iii) and (iv) follow from the selection of $a$ and $b$, and the definition of the successor function.  For invariant (v), let $t'(i+1) = \ell^*+1- r'$, then $s'(i+1) < s'(i+1) <t(i+1) = \ell^*+1- r(i)$.  This concludes the proof.
\end{proof}

\begin{lemma}
  The maximal induced interval subgraphs problem
  and the maximal connected induced interval subgraphs problem
  can be solved with polynomial delay.
\end{lemma}
\begin{proof}
  The first result follows from Theorem~\ref{thm:solution-map}, and Lemmas~\ref{lem:interval-successor-function} and \ref{lem:interval-retaliation-path}.  The second then follows from Proposition~\ref{lem:connected} because the class of interval graphs is closed under adding universal vertices.
\end{proof}

\section{Enumeration in incremental polynomial time}\label{sec:incremental-poly}

We generalize the core concept of Cohen et al.~\cite{cohen-08-all-maximal-induced-subgraphs}, namely, the input-restricted version of the (connected) maximal induced $\mathcal{P}$ subgraphs problem.
For any nonnegative integer $t$, we formally define the $t$-restricted version of the problem as follows.

\begin{problem}{Maximal (connected) induced $\mathcal{P}$ subgraphs, $t$-restricted version}
  \Input & a graph $G$ and a set $Z$ of $t$ vertices such that $Z$ is a subset of \emph{every} forbidden set of $G$.\\
  \Prob & all maximal (connected) $\mathcal{P}$ sets of $G$.
\end{problem}

For the situations where the $t$-restricted version is motivated, there must be forbidden sets in $G$, which is hence not in $\mathcal{P}$, and we are interested in solutions that are supersets of $Z$.
Since we now study the $t$-restricted version as a problem by itself, we do not make such assumptions.  In particular, $G$ might be in $\mathcal{P}$; in this case, it is vacuously that $Z$ is a subset of {every} forbidden set of $G$.  We characterize those solutions that are not supersets of $Z$ with the following proposition.

\begin{proposition}\label{lem:not-z}
  Let $\mathcal{P}$ be a hereditary graph class, and $G$ a graph.  If there is a set $Z$ of $t, t \ge 1$, vertices in $G$ such that {every} forbidden set of $G$ contains $Z$, then at most $|Z|$ maximal $\mathcal{P}$ sets and at most $|Z|n$ maximal connected $\mathcal{P}$ sets of $G$ do not contain $Z$ as a subset.  Moreover, these sets can be found in polynomial time.
\end{proposition}
\begin{proof}
  If $G$ is in $\mathcal{P}$, then $V(G)$ is the only maximal $\mathcal{P}$ set of $G$, while the components of $G$ are the only maximal connected $\mathcal{P}$ sets of $G$.  In the rest $G$ is not in $\mathcal{P}$.  For every $z\in Z$, the set $V(G)\setminus \{z\}$ is a $\mathcal{P}$ set of $G$ by assumption; it is maximal because its only proper superset $V(G)$ is not a $\mathcal{P}$ set.
  These $|Z|$ sets are all the maximal $\mathcal{P}$ sets of $G$ that are not superset of $Z$.
  Every maximal connected $\mathcal{P}$ set of $G$ that does not contains all vertices in $Z$ is a subset of $V(G)\setminus \{z\}$ for some $z\in Z$, hence a component in $G - z$.  There are at most $|Z|(n-1)$ such sets, and it suffices to check each of them to see whether it is maximal.  This concludes the proof.
\end{proof}

The $0$-restricted and the $1$-restricted versions in our definition are respectively the original problem and the input-restricted version defined in \cite{cohen-08-all-maximal-induced-subgraphs}.
The following, together with Theorem~\ref{thm:connected-incp=totalp} and \ref{thm:incp=totalp}, implies Theorem~\ref{thm:poly-total} immediately.
\begin{lemma}
  \label{lem:t-restricted}
   Let $\mathcal{P}$ be a hereditary graph class, and $t$ a nonnegative integer.
   The $t$-restricted version of the maximal (connected) induced $\mathcal{P}$ subgraphs problem can be solved in polynomial total time if and only if the ($t+1$)-restricted version of the same problem can be solved in polynomial total time.
\end{lemma}
\begin{proof}
  The only if direction is trivial: If $(G, Z)$ is a $(t + 1)$-restricted instance, then  $(G, Z\setminus \{z\})$ for any vertex $z\in Z$ is a $t$-restricted instance.  
  For the if direction, suppose that algorithm $A$ solves the $(t+1)$-restricted version of the problem in time $p(n, N)$ for some polynomial function $p$.
  Let $(G, Z)$ be a $t$-restricted instance.

  Let us start from some trivial cases.  If $G\in \cal P$, then the only maximal $\mathcal{P}$ set of $G$ is $V(G)$, and the only maximal connected $\mathcal{P}$ sets are the components of $G$.   We have characterized in Proposition~\ref{lem:not-z} all solutions that are not supersets of $Z$.  If $Z$ is a forbidden set, then there is no other solutions.  If $Z$ is a maximal $\mathcal{P}$ set, then there cannot be any other maximal (connected) $\mathcal{P}$ sets of $G$; the set $Z$ itself is a maximal connected $\mathcal{P}$ sets of $G$ if and only if $G[Z]$ is connected.
All the special cases can be checked in polynomial time, and hence the problem can be solved in polynomial time if any of them is true.  In the rest we assume that $G\not\in \cal P$ and that $Z$ is a $\mathcal{P}$ set of $G$ but  not maximal.

For each solution $S$ with $Z\subset S$, and each vertex $v\in V(G)\setminus S$, we claim that $(G[S\cup\{v\}], Z\cup \{v\})$ is a ($t+1$)-restricted instance of the problem.  (For the connected variation, it suffices to consider each vertex $v\in N(S)$.)  Let $G' = G[S\cup\{v\}]$ and $Z' = Z\cup \{v\}$.  Since $G' - v = G[S]\in \cal P$, every forbidden set of $G'$ contains $v$.  On the other hand, by assumption, every forbidden set of $G'$, which is also a forbidden set of $G$, contains $Z$.  Therefore, every forbidden set of $G'$ contains $Z'$, and we can use algorithm $A$ to solve the instance $(G', Z')$.  For each maximal (connected) $\mathcal{P}$ set $S'$ of $G'$ that is different from $S$, we use proposition~\ref{lem:extension} to extend $S'$ to a solution $S''$ of $G$.  The set $\operatorname{succ}(S, v)$ comprises of all these solutions $S''$.
According to Proposition~\ref{lem:subgraph-cardinality}(i) and (iii), the number of maximal (connected) $\mathcal{P}$ set of $G'$ is at most $n N$.  Thus, the successor function can be calculated in time $n\cdot p(n, n N)$.

It remains to show that the solution map $M(G)$ defined as above is strongly connected, for which we show the existence of a path from any solution $S$ to any other solution $S^*$.  We proceed differently for the two variations.  The claim then follows from Theorem~\ref{thm:solution-map}.

Consider first the maximal induced $\mathcal{P}$ subgraphs problem.  Let $\widehat S = S\cap S^*$, and let $v^*$ be any vertex in $S^*\setminus S$, which exists because $S\ne S^*$.  Since $\widehat S \cup\{v^*\}\subseteq S^*$, it is a $\mathcal{P}$ set.  There exists a maximal $\mathcal{P}$ set $S'$ of $G[S\cup\{v^*\}]$ such that $\widehat S \cup\{v^*\}\subseteq S'$.  Let $S''$ be the set in $\operatorname{succ}(S, v^*)$ that is obtained from extending $S'$.  Then $|S''\cap S^*| \ge |\widehat S| + 1 > |\widehat S| = |S\cap S^*|$.  Thus, we can reach $S^*$ from $S$ after at most $|S^*|$ steps.

Now consider the maximal connected induced $\mathcal{P}$ subgraphs problem.  Let $\widehat S$ be the largest component in the subgraph induced by $S\cap S^*$, and let $v^*$ be any neighbor of $\widehat S$ in $S^*$.  Note that $v^*$ exists because $\widehat S \subset S^*$ and $G[S^*]$ is connected.  Since $G[\widehat S]$ is connected, and $v^*$ is a neighbor of $\widehat S$, the subgraph induced by $\widehat S \cup\{v^*\}$ is connected as well.   Then from $\widehat S \cup\{v^*\}\subseteq S^*$ we can conclude that $\widehat S \cup\{v^*\}$ is a connected $\mathcal{P}$ set.  There exists a maximal connected $\mathcal{P}$ set $S'$ of $G[S\cup\{v^*\}]$ such that $\widehat S \cup\{v^*\}\subseteq S'$.  Let $S''$ be the set in $\operatorname{succ}(S, v^*)$ that is obtained from extending $S'$.  Then the largest component in the subgraph induced by $S''\cap S^*$ has at least $|\widehat S| + 1$ vertices.  Thus, we can reach $S^*$ from $S$ after at most $|S^*|$ steps.
\end{proof}

Before we use Theorem~\ref{thm:poly-total} to develop new algorithms, let us mention that
Theorem~\ref{thm:poly-total} subsumes the main result of Cohen et al.~\cite[Characterization 1]{cohen-08-all-maximal-induced-subgraphs}, which corresponds to the simplest case, namely, when $t = 1$. 

We have explained in Section~\ref{sec:map} that the connected variation is more challenging because the lack of hereditary property.
We have also seen easy reductions from the maximal induced $\mathcal{P}$ subgraphs problem to its connected variation, but not the other way round.  The following is the technical version of Theorem~\ref{thm:equivalence}.  The main observation here is that as a consequence of the connectivity condition, vertices in $Z$ are in the same component in every maximal induced $\mathcal{P}$ subgraph of $G$ that contains all vertices in $Z$.  Since, as mentioned, our main work is on these solutions, the $t$-restricted version of the two variations are practically equivalent under the stated condition.
Also note that a graph class is closed under disjoint union if and only if each forbidden induced subgraph of this class is connected.

\begin{lemma}\label{lem:disconnected}
  Let $\mathcal{F}$ be a set of graphs such that every graph in $\mathcal{F}$ of order $c$ or above is biconnected.  If for some $t\ge c$, the $t$-restricted version of the maximal induced $\mathcal{F}$-free subgraphs problem can be solved in polynomial time (polynomial total time), then the $t$-restricted version of the maximal connected induced $\mathcal{F}$-free subgraphs problem can be solved in polynomial time (polynomial total time). 
\end{lemma}
\begin{proof}
  Let $\mathcal{P}$ be the class of $\mathcal{F}$-free graphs.
  First, by proposition~\ref{lem:enumeration-recognition}, the fact that the maximal (connected) induced $\mathcal{P}$ subgraphs problem can be solved in polynomial total time implies a polynomial-time algorithm for deciding in polynomial time whether a (connected) graph is in $\mathcal{P}$.  Even the algorithm works only on connected graphs, we can use it to decide whether a general graph $G$ is in $\mathcal{P}$ as follows.  We check whether $G$ contains any induced subgraph in $\{ F\in \mathcal{F}\mid |F| \le c\}$ in $O(n^{c})$ time; we return ``no'' if one is found.  Otherwise, we check whether each component of $G$ is in $\mathcal{P}$, and we return ``yes'' if and only if all of them are.  Since a forbidden induced subgraph $F$ of order larger than $c$ is connected, if $F$ is an induced subgraph of $G$, then it is completely contained in one component of $G$.

  Let $G$ be the input graph and $Z$ a set of $t$ vertices such that $Z$ is a subset of {every} forbidden set of $G$.  Since $t \ge c$, every forbidden induced subgraph of $G$ is biconnected by assumption.  If $G$ is actually in $\mathcal{P}$, then the only maximal $\mathcal{P}$ set of $G$ is $V(G)$, and the only maximal connected $\mathcal{P}$ sets of $G$ are its components.  Otherwise, there exists at least one forbidden set in $G$; it is a superset of $Z$, and hence has order at least $t$ and is connected, by assumption.  Therefore, vertices in $Z$ are in the same component $C$ of $G$.
  If $G$ is not connected, then it suffices to consider $C$.   Every other component of $G$ different from $C$ is a maximal connected $\mathcal{P}$ set of $G$, and the other maximal connected $\mathcal{P}$ sets of $G$ are the maximal connected $\mathcal{P}$ sets of $C$.  On the other hand, by Proposition~\ref{lem:redundant-vertices}, each maximal $\mathcal{P}$ set of $G$ consists of a maximal $\mathcal{P}$ set of $C$, and all the other components of $G$.
  In the rest we may assume without loss of generality that $G$ is connected and that $G$ is not in $\mathcal{P}$.

  We argue first that in any maximal induced $\mathcal{P}$ subgraph of $G$ that contains all vertices in $Z$, vertices in $Z$ are in the same component.  Suppose for contradiction that $S$ is a maximal $\mathcal{P}$ set of $G$ but vertices in $Z$ are in different components of $G[S]$.  For any vertex $x\in V(G)\setminus S$, the maximality of $S$ implies that $G[S\cup \{x\}]$ is not in $\mathcal{P}$.  Let $X$ be a forbidden set of $G[S\cup \{x\}]$; by assumption, $Z\cup \{x\}\subseteq X$.  However, either $G[X]$ is disconnected, or the vertex $x$ is a cutvertex in $G[X]$, contradicting that every graph in $\mathcal{F}$ of order $\ge c$ is biconnected.

  We have characterized in Proposition~\ref{lem:not-z} all maximal (connected) $\mathcal{P}$ sets of $G$ that are not supersets of $Z$.
  Therefore, for this proof, it suffices to consider supersets of $Z$.  We show that if $G$ is a connected graph not in $\mathcal{P}$, then a set $S\subseteq V(G)$ with $Z\subseteq S$ and $G[S]$ connected is a maximal connected $\mathcal{P}$ set of $G$ if and only if $V(G)\setminus N(S)$ is a maximal $\mathcal{P}$ set of $G$.
  Note that $V(G)\setminus N(S)$ consists of two parts, $S$, and $V(G)\setminus N[S]$, and there is no edge between them.  The set $V(G)\setminus N[S]$ is a $\mathcal{P}$ set because it is disjoint from $Z$.
  For the only if direction, suppose that $S$ is a maximal connected $\mathcal{P}$ set of $G$.  Then for any $x\in N(S)$,  there must be a forbidden set in $S\cup\{x\}$.  Since any forbidden induced subgraph of $G$ is biconnected, $V(G)\setminus N(S)$ is a $\mathcal{P}$ set.  Therefore, $V(G)\setminus N(S)$ is a maximal $\mathcal{P}$ set.
  For the if direction, suppose that $V(G)\setminus N(S)$ is a maximal $\mathcal{P}$ set of $G$.  Then $S$
  is a connected $\mathcal{P}$ set of $G$.  If it is not maximal, then there is some $x\in N(S)$ such that $S\cup \{x\}$ is a connected $\mathcal{P}$ set of $G$.  But then $V(G)\setminus N(S)$ remains a $\mathcal{P}$ set of $G$ with $x$ added: In the subgraph induced by $(V(G)\setminus N(S))\cup \{x\}$, either $G[S\cup \{x\}]$ is a component or $x$ is a cutvertex.  This concludes the proof.
\end{proof}

Lemma~\ref{lem:disconnected} and Theorem~\ref{thm:poly-total} imply Theorem~\ref{thm:equivalence}.
In passing we remark that the connectivity requirement in Lemma~\ref{lem:disconnected} cannot be relaxed to one.  See Figure~\ref{fig:biconnected-necessary} for an example.

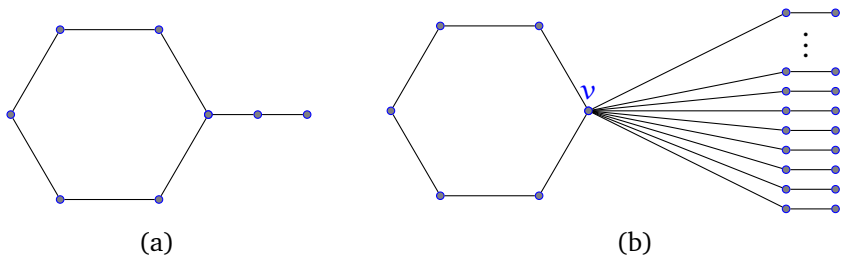
\begin{figure}[h!]
  \centering
  \begin{subfigure}[b]{.25\linewidth}
    \centering
    \begin{tikzpicture}[scale=1.3]
      \draw [white, rounded corners] (2,-1) rectangle (4, .5);
      \begin{scope}[every node/.style={filled vertex}]
        \node (v0) at (1, 0) {};
        \foreach[count=\i from 0] \x in {1, ..., 6}{
          \node (v\x) at (60*\x : 1) {};
          \draw (v\i) -- (v\x);
        }
        \node(x) at (1.5, 0) {};
        \node(y) at (2, 0) {};
        \draw (v0) -- (x) -- (y);              
      \end{scope}
    \end{tikzpicture}
    \caption{}
  \end{subfigure}
  \qquad
  \begin{subfigure}[b]{.43\linewidth}
    \centering
    \begin{tikzpicture}[scale=1.3]
      \draw [white, rounded corners] (2,-1) rectangle (4, .5);
      \begin{scope}[every node/.style={filled vertex}]
        \node["$v$"] (v0) at (1, 0) {};
        \foreach[count=\i from 0] \x in {1, ..., 6}{
          \node (v\x) at (60*\x : 1) {};
          \draw (v\i) -- (v\x);
        }
        \foreach[count=\i] \p in {-5, ..., 2, 5} {
          \node(x\i) at (3, \p/5) {};
          \node(y\i) at (3.5, \p/5) {};
          \draw (v0) -- (x\i) -- (y\i);              
        }
      \end{scope}
      \node at (3.2, .75) {$\vdots$};
    \end{tikzpicture}
    \caption{}
  \end{subfigure}
  \caption{Let $\mathcal{F}$ be the set of graphs obtained by identifying a vertex of a cycle and a $2$-path, e.g., (a).  The graph in (b), in which every forbidden set contains $v$, has only $\ell$ (the length of the cycle) maximal connected induced $\mathcal{F}$-free subgraphs, while an exponential number of maximal induced $\mathcal{F}$-free subgraphs.}
  \label{fig:biconnected-necessary}
\end{figure}

In the rest of this section we show nontrivial applications of Theorem~\ref{thm:poly-total}, in solving the maximal (connected) induced $\mathcal{P}$ subgraphs problem for graph classes that have an infinite number of forbidden induced subgraphs.  With Lemma~\ref{lem:disconnected} and Theorem~\ref{thm:equivalence}, we can ignore the connected requirement.

\subsection{Chordal graphs and subclasses}
A graph is \emph{chordal} if it contains no holes, i.e., a simple cycle on four or more vertices.
Conte and Uno~\cite{conte-19-polynomial-delay} have presented algorithms solving the maximal (connected) induced chordal subgraphs problem with polynomial delay.  Our results on chordal graphs are inferior to theirs.  We include them here because they make a very nice use of Theorem~\ref{thm:poly-total}.  Moreover, since other results in this section are developed based on the algorithm for chordal graphs, including them here also makes this section self-contained.

We start from a nice characterization of instances in the $3$-restricted version of the problem.
A path with ends $x$ and $y$ is called an \stpath{x}{y}, and all other vertices on the path are its \emph{inner vertices}.  A path is nontrivial if its length is at least two; for an induced path, this means in particular that its two ends are not adjacent to each other.

\begin{lemma}\label{lem:restricted-chordal-structure}
  Let $G$ be a graph, and $Z$ a set of three vertices of $G$.  If every vertex of $G$ is contained in some hole, and every hole of $G$ contains $Z$, then
  \begin{enumerate}[(i)]
  \item $G - Z$ has $3 - |E(G[Z])|$ components, and 
  \item each component in $G - Z$ is adjacent to precisely two nonadjacent vertices in $Z$.
  \end{enumerate}
\end{lemma}
\begin{proof}
  Let $Z = \{z_1, z_2, z_3\}$.  By assumptions there exists at least one hole in $G$, which contains $Z$.  Hence, $Z$ cannot be a clique.  Moreover, no vertex in $V(G)\setminus Z$ can be adjacent to all three vertices in $Z$.  We establish this lemma by a sequence of claims.

  Our first claim is that each component of $G - Z$ is adjacent to at least two vertices in $Z$.  Suppose for contradiction that a component $C$ of $G - Z$ has only one neighbor in $Z$; without loss of generality let it be $z_1$.  Then $\{z_1\}$ is a \stsep{v}{z_2} for every $v\in C$, and there cannot be a hole containing $v$ and $z_2$, contradicting the assumptions.  

  The second claim is that if two vertices in $Z$ are adjacent, then there is no component of $G - Z$ adjacent to both of them.  Suppose for contradiction that $z_1 z_2\in E(G)$ while $z_1$ and $z_2$ are both adjacent to a component $C$ of $G - Z$.  No vertex in $C$ can be adjacent to both $z_1$ and $z_2$; a vertex in $N(z_1)\cap N(z_2)$ cannot be on a hole with them.  Since $C$ is connected and adjacent to $z_1$ and $z_2$, we can find a \stpath{u}{v} $P$ in $G[C]$ such that $u$ and $v$ are the only neighbor of $z_1$ and $z_2$, respectively, on this path.
Then together with edge $z_1 z_2$, the path $P$ forms a hole avoiding $z_3$, contradicting the assumptions.

  The third claim is that there cannot be two components of $G - Z$ adjacent to the same pair of vertices in $Z$.  Suppose for contradiction that components $C_1$ and $C_2$ of $G - Z$ are both adjacent to $z_1$ and $z_2$.  By the second claim, we know that $z_1 z_2\not\in E(G)$.  We can find a nontrivial \stpath{z_1}{z_2} $P_1$ with inner vertices from $C_1$ and another nontrivial \stpath{z_1}{z_2} $P_2$ with inner vertices from $C_2$.  These two paths together make a hole avoiding $z_3$, contradicting the assumptions.  

  The last and main claim is that there cannot be a component of $G - Z$ adjacent to all three vertices in $Z$.  Suppose for contradiction that there is such a component of $G - Z$.  Then by the previous three claims, it is the only component of $G - Z$ and $Z$ is an independent set.
  We take any hole $H$ of $G$, which, by assumption, visits all three vertices in $Z$.
  Let $x$ and $y$ be the two neighbors of $z_3$ on $H$ such that the hole can be traversed in the order of $x \cdots z_1\cdots z_2\cdots y z_3$.  Since $G - Z$ is connected, we can find an induced \stpath{x}{y} $P$ in it.  This path is nontrivial because $x$ and $y$ are not adjacent.  All vertices on $P$ are adjacent to $z_3$; otherwise there is a hole avoiding $z_1$ and $z_2$.  Let $x'$ and $y'$ be the neighbors of, respectively,  $x$ and $y$ on $P$.  Neither of them can be in $N(z_1)\cap N(z_2)$, because they are already adjacent to $z_3$.

  Let $P'$ be the \stpath{x}{y} obtained from $H$ by deleting $z_3$.  We number the vertices such that it is $x = v_1 v_2 \ldots v_{|H| - 1} = y$.
  The neighbors of $x'$ on this path must be consecutive; otherwise, there is a hole involving only $x'$ and some vertices on $P'$.
  Since $x'$ is adjacent to $v_1$ and nonadjacent to at least one of $z_1$ and $z_2$, both on this path, there exists $a < |H| - 1$ such that the neighborhood of $x'$ on $P'$ is $\{v_1, \ldots, v_a\}$.  For the same reason, there exists $b > 1$ such that the neighborhood of $y'$ on $P'$ is $\{v_b, \ldots, v_{|H| - 1}\}$.  As a result, $x' \ne y'$.
 If $a \ge b$, then $z_3 x \cdots v_b y'$ is a hole of $G$ avoiding $z_2$, which is impossible.

  In the rest, $a < b$.  If $x'$ and  $y'$ are adjacent, then $x' v_a \cdots v_b y'$ is a hole avoiding $z_3$.  Now that $x' y'\not\in E(G)$, let $x''$ be the third vertex on $P$, i.e., the next neighbor of $x'$.  It is adjacent to $v_a$: Otherwise we can find an induced \stpath{x''}{v_a} from the path $x''\cdots y' v_b \cdots v_a$, of which the inner vertices are nonadjacent to $x'$, and with $x'$, it makes a hole avoiding $z_3$.   Then $x''$ is adjacent to the two ends of the path $z_3v_1 \ldots v_a$, but nonadjacent to at least one inner vertex of it, namely, $x$, and there is a hole involving $x''$ and vertices on this path.  This hole avoids $z_2$, which is impossible.

  Assertion (ii) of the lemma follows from the first and the last claims.  For assertion (i), the second and third claims imply that the number of components of $G - Z$ is at most $3 - |E(G[Z])|$.  On the other hand, suppose that $z_1 z_2\not\in E(G)$ and there is no component of $G - Z$ adjacent to $z_1$ and $z_2$, then $z_3$ is a cutvertex, which is impossible.  This concludes the proof.
\end{proof}

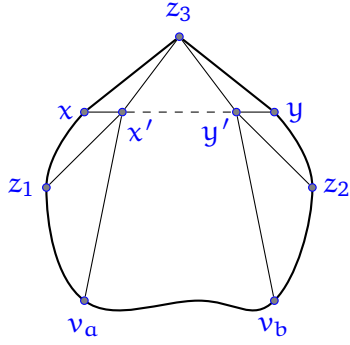
\begin{figure}[h!]
  \centering
  \begin{tikzpicture}[every node/.style={filled vertex}, scale=1.]
    \draw[thick] plot[smooth, tension=.7] coordinates {(0,0) (-0.5,-1) (0,-2.5) (1.5,-2.5) (2.5,-2.5) (3,-1) (2.5,0)};
    \node["$x$" left] (x1) at  (0,0) {};
    \node["$y$" right] (x2) at (2.5,0) {};

    \node["$v_a$" below] (v3) at (0, -2.5) {};
    \node["$v_b$" below] (v4) at (2.5, -2.5) {};

    \node["$z_1$" left] (z1) at (-0.5,-1) {};
    \node["$z_2$" right] (z2) at (3,-1) {};
    \node["$z_3$"] (z3) at (1.25,1) {};
    \node["$x'$" below right] (y1) at (0.5,0) {};
    \node["$y'$" below left] (y2) at (2,0) {};

    \draw[thick] (x1) -- (z3) -- (x2);
    \draw (y1) -- (z3) -- (y2);
    \draw (x1) edge (y1);
    \draw (y2) edge (x2);
    \draw[dashed] (y1) edge (y2);
    \draw (y1) edge (v3);
    \draw (y2) edge (v4);
    \draw (y1) edge (z1);
    \draw (y2) edge (z2);
  \end{tikzpicture}
  \caption{Illustration for the proof of Lemma~\ref{lem:restricted-chordal-structure}.}
  \label{fig:chordal}
\end{figure}

One may show that for each component $C$ of $G - Z$, the subgraph induced by $C\cup N(C)$ is an interval graph, of which the only two simplicial vertices are $N(C)$.  We will not prove this fact because it is not used in our algorithm.

The characterization presented in Lemma~\ref{lem:restricted-chordal-structure} only holds when every vertex is visited by some hole.  We can find all such vertices in polynomial time.
\begin{proposition}\label{lem:hole-free-vertex}
  We can decide in $O(n^4)$ time whether a vertex is contained in some hole.
\end{proposition}
\begin{proof}
  Let $v$ be the vertex.  For each pair of vertices $u, w\in N(v)$ such that $uw\not\in E(G)$, we check whether they are connected in $G - N[v]$.  Every induced \stpath{u}{w} in this subgraph makes a hole together with $v$.  It is also easy to verify that if $v$ is in some hole, then such a pair must exist.
\end{proof}

We are now ready to solve the $3$-restricted version of the maximal induced chordal subgraphs problem.
\begin{lemma}\label{lem:alg-restricted-chordal}
The $3$-restricted version of the maximal induced chordal subgraphs problem can be solved in polynomial time.
\end{lemma}
\begin{proof}
  Let $G$ be the input graph and $Z$ the set of three vertices such that every hole of $G$ contains $Z$.  We can use Proposition~\ref{lem:hole-free-vertex} to find all vertices that are not in any hole.   By Proposition~\ref{lem:redundant-vertices}, we can remove them and focus on the remaining graph, of which every vertex is contained in some hole.  By Lemma~\ref{lem:restricted-chordal-structure}, there are precisely $3 - |E(G[Z])|$ components in $G - Z$, each of which is adjacent to two nonadjacent vertices in $Z$.

  Let $Z =\{z_1, z_2, z_3\}$ and suppose $z_1$ and $z_2$ are nonadjacent.  Denote by $C$ the component of $G - Z$ that is adjacent to $z_1$ and $z_2$.  For each minimal \stsep{z_1}{z_2} $Y$ in $G[C\cup \{z_1, z_2\}]$, every \stpath{z_1}{z_2} in $G - Y$ visits the vertex $z_3$.  Therefore, $G - Y$ is chordal.  It is maximal because adding any vertex $x\in Y$ back will introduce an induced \stpath{z_1}{z_2} with inner vertices from $C$.  There is another induced \stpath{z_1}{z_2} with inner vertices from other components of $G - Z$ and $z_3$.  These two paths make a hole.
  Therefore, a set $S'\subseteq V(G)$ is a maximal chordal set of $G$ if and only if $V(G)\setminus S'$ is such a minimal separator.

The subgraph induced by $C\cup \{z_1, z_2\}$ is chordal, because it does not contain $z_3$.  A chordal graph has at most $n$ minimal separators, and they can be found in linear time \cite{blair-91-chordal-graphs-clique-trees}.  There are at most three such subgraphs, and hence there are at most $n$ solutions, and they can be found in polynomial time.
\end{proof}

The algorithm for the maximal induced chordal subgraphs problem follows from Theorem~\ref{thm:poly-total} and Lemma~\ref{lem:alg-restricted-chordal}.
Since every hole is biconnected, the algorithm for the connected variation then follows from Lemma~\ref{lem:disconnected}.
\begin{lemma}
  The maximal induced chordal subgraphs problem
  and
  the maximal connected induced chordal subgraphs problem
  can be solved in incremental polynomial time.
\end{lemma}

The class of unit interval graphs is a subclass of chordal graphs.  Apart from holes, its forbidden induced subgraphs include claw, net, and tent.  Likewise, the forbidden induced subgraphs of the class of block graphs, of the class of 3-leaf powers \cite{dom-06-leaf-power}, and the class of basic 4-leaf powers \cite{brandstadt-08-4-leaf-powers} are holes and some small graphs.  See the appendix for details.  It is straightforward to use the observations above to solve the maximal induced $\mathcal{P}$ subgraphs problem for these graph classes.

\begin{lemma}\label{lem:alg-unit-interval}
  The maximal (connected) induced $\mathcal{P}$ subgraphs problem can be solved in incremental polynomial time for the following graph classes: unit interval graphs, block graphs, 3-leaf powers, and basic 4-leaf powers.
\end{lemma}
\begin{proof}
  We consider its $6$-restricted version of the maximal induced unit interval subgraphs problem.  Let $G$ be the input graph and $Z$ the set of six vertices such that every forbidden set of $G$ contains $Z$.  There cannot be any claw in $G$.  On the other hand, if $G$ contains any net or tent, then it is $Z$, and it is the only forbidden set.  In this case, the solutions are $\{V(G)\setminus \{v\}\mid v\in Z\}$.  In the rest, $G$ does not contain any claw, net, or tent.  In other words, all forbidden induced subgraphs of $G$ are holes.  We can take any three vertices $Z'\subset Z$, and consider $(G, Z')$ as a $3$-restricted version of the maximal induced chordal subgraphs problem.  It can be solved by Lemma~\ref{lem:alg-restricted-chordal}.  By Theorem~\ref{thm:poly-total} and Theorem~\ref{thm:connected-incp=totalp}, the maximal induced unit interval subgraphs problem can be solved in incremental polynomial time.  Moreover, the only forbidden induced subgraphs that are not biconnected are the claw and the net, the algorithm for the connected variation then follows from Lemma~\ref{lem:disconnected}.

  The proofs for the other three graph classes are word-by-word translations from that of unit interval graphs, hence omitted.
\end{proof}

We remark that the maximal (connected) induced unit interval subgraphs problem can be solved directly (actually more efficiently) as follows.  This time we consider the connected variation.  Let $(G, Z)$ be a $6$-restricted instance of the maximal connected induced unit interval subgraphs problem.  From \cite[Proposition 2.4]{cao-17-unit-interval-editing} one can infer that either it is a trivial instance, where $Z$ is the only forbidden set, or every component of $G$ is a proper Helly circular-arc graph.  In the second case, there are $O(n)$ solutions and they can be found in $O(n)$ time.

However, the approach used in Lemma~\ref{lem:alg-unit-interval} cannot be directly generalized to graph classes that have an infinite number of forbidden induced subgraphs apart from holes.  The most famous graph class with this property are arguably the class of interval graphs and the class of strongly chordal graphs.  We have solved the former (in a better time) in the previous section.  The latter remains open to the best of our knowledge.

\subsection{Wheel-free graphs}
A \emph{wheel} consists of an induced cycle and another vertex, called the \emph{center} of the wheel, that is adjacent to all vertices on the cycle.  Note that only the center can have a degree larger than three in a wheel, and a wheel of four vertices is a clique, in which every vertex can be viewed as the center.  Since all wheels are biconnected, by Lemma~\ref{lem:disconnected}, we focus on the maximal induced wheel-free subgraphs problem.

Let $G$ be the input graph.  If there is a universal vertex $v\in V(G)$ such that $G - v$ is wheel-free, then $(G, v)$ is a special instance of the $1$-restricted version of the problem.  This special version is indeed equivalent to finding all the maximal induced forests of $G - v$, because every wheel of $G$ has $v$ as its center.  Enumerating induced wheel-free subgraphs of $G$, apart from components of $G - v$, is then equivalent as enumerating induced forests of $G - v$, which can be solved by the algorithm of Schwikowski and Speckenmeyer~\cite{schwikowski-02-enumerate-fvs}.  This however does not work when the special vertex of the $1$-restricted version is not universal.

We solve this problem by considering the $5$-restricted version.  Let ($G, Z$) be the instance.  Either $Z$ contains a vertex that is the center of all the wheels in $G$, then we are in the easy case previously mentioned; or we can partition $V(G)$ into center vertices and cycle vertices.

\begin{lemma}\label{lem:wheel-free}
  The $5$-restricted version of the maximal induced wheel-free subgraphs problem can be solved in polynomial total time.
\end{lemma}
\begin{proof}
  We can decide whether a vertex in $G$ is in a wheel as follows.  For each vertex $v\in V(G)$, we check whether the subgraph induced by $N(v)$ is a forest.  If not, we mark $v$, and all the vertices that are contained in some cycles of $G[N(v)]$.  After the marking is done for all the vertices in $G$, those unmarked vertices are not contained in any wheel.  By Proposition~\ref{lem:redundant-vertices}, we may assume that every vertex in $G$ is contained in some wheel.  Suppose that $Z$ is a set of five vertices such that every wheel in $G$ contains all the five vertices in $Z$.  A particular consequence of this assumption is that there cannot be any $4$-clique, i.e., a wheel on four vertices.  As a consequence, the center of a wheel in $G$ is unambiguous.

  If one vertex $u$ in $Z$ is adjacent to all the other four vertices in $Z$, then $u$ has to be the center of every wheel of $G$.  By assumption, every vertex in $V(G)\setminus \{u\}$ is  contained in some wheel, hence adjacent to $u$.  In this case, $V(G)\setminus \{u\}$ is a solution, and every other solution contains $u$.  A set $S\subseteq V(G)$ that contains $u$ is a solution if and only if $S\setminus\{u\}$ is a forest.  Therefore, we can call the algorithm of Schwikowski and Speckenmeyer~\cite{schwikowski-02-enumerate-fvs} to find all maximal induced forests of $G - u$, and for each of them, add $u$ to produce a solution of $G$.

  Now that no vertex in $Z$ is universal in $G[Z]$, the maximal degree in $G[Z]$ is at most two.
  Since every vertex $x$ in $V(G)\setminus Z$ is in some wheel, $x$ is adjacent to either all vertices or at most two vertices in $Z$.  Therefore, the vertex set $V(G)$ can be partitioned into $A$ and $B$, where every vertex in $A$ is the center of some wheel, but not on the cycle of any wheel, and every vertex in $B$ is on the cycle of some wheel but not a center.  Note that $Z\subseteq B$.

  Since there is no center vertex in $B$, the induced subgraph $G[B]$ is wheel-free.  Since every vertex in $A$ is the center of some wheel, of which the vertices on the cycle are from $B$, the wheel-free induced subgraph $G[B]$ is maximal.  Therefore, $B$ is a solution.
  We create a search tree $\cal T$ containing only $B$ initially, and add all the solutions of $G$ into $\mathcal{T}$.  For each solution $S$ in $\mathcal{T}$ with $A\not\subseteq S$, and each center vertex $c\in A\setminus S$, we add the following solutions to $\mathcal{T}$.
  We identify the subset $B'$ of vertices contained in cycles in $G[S\cap N(c)]$; in other words, those vertices in $S$ that can form a wheel with $c$.  We call the algorithm of  Schwikowski and Speckenmeyer \cite{schwikowski-02-enumerate-fvs} to find all maximal induced forests of $G[B']$.  For each $F$ of them, we extend $F\cup (S\cap B\setminus B')$ to a solution $S$; note that vertices in $S\cap B\setminus B'$ do not form any wheel with $c$.  This solution is inserted into $\mathcal{T}$ if it was not already there.  Since each solution obtained as such has a different intersection with $B'$, they are different.  Their number is thus upper bounded by $N$ by Proposition~\ref{lem:subgraph-cardinality}(i).  The total running time of this algorithm is polynomial total.

We now argue the correctness of the algorithm by induction on the number of vertices from $A$ in a solution.  The base case is $0$, and the only solution disjoint from $A$ is $B$.  Now suppose that for every solution with at most $i$ vertices from $A$ has been found, we show that every solution with precisely $i + 1$ vertices from $A$ are found as well.  Let $S$ be a solution with $|S\cap A| = i + 1$.  The set of solutions $S'$ satisfying
\[
  S\cap B\subseteq S' \text{ and }  S'\cap A\subseteq S
\]
is not empty because $S\cap B\subseteq B$ and $B\cap A = \emptyset$.  Let $S'$ be a solution from this set with $|S'\cap A|$ maximized, and let $c$ be any vertex in $A\setminus S'$.  Then one solution we find for $S'$ and $c$ contains $(S\cap B) \cup (S\cap A) \cup \{c\}$.  By the selection of $S'$, which maximize $|S'\cap A|$, this solution has to be $S$.  This concludes the proof.
\end{proof}

The main result of this subsection then follows from Theorem~\ref{thm:poly-total} and Lemmas~\ref{lem:wheel-free} and \ref{lem:disconnected}.
\begin{lemma}
  The maximal induced wheel-free subgraphs problem
  and the maximal connected induced wheel-free subgraphs problem
  can be solved in incremental polynomial time.
\end{lemma}

\section{The \textsc{cks} property}\label{sec:cks}

We say that a graph class $\mathcal{P}$ has the \emph{\textsc{cks} property} if 
\begin{quote}
   In any graph $G$ that contains a maximal $\mathcal{P}$ set of size $n - 1$, there are at most a polynomial number of maximal $\mathcal{P}$ sets.
\end{quote}
The class of edgeless graphs clearly has this property: If $V(G)\setminus \{v\}$ is a maximal independent set, then all edges of $G$ are incident to $v$, and hence the only other maximal independent set of $G$ is $V(G)\setminus N(v)$.  To characterize graph classes with the \textsc{cks} property, we need the following folklore result.
Note that a trivial and hereditary graph class contains either all graphs or only a finite number of graphs.
\begin{proposition}\label{lem:finite-graph-class}
  A hereditary graph class $\mathcal{P}$ has a finite number of graphs if and only if there are positive integers $p$ and $q$ such that $p$-clique and $q$-independent set are its forbidden induced subgraphs.
\end{proposition}
\begin{proof}
  The if condition is a simple result of Ramsey theorem.  For the only if direction, let $p$ and $q$ be the smallest numbers such that $p$-clique and $q$-independent set are not in $\mathcal{P}$; they exist because there are only a finite number of graphs in $\mathcal{P}$.  Then $p$-clique and $q$-independent set are forbidden induced subgraphs of $G$ because all their induced subgraphs, if not empty, are smaller cliques or independent sets, hence in $\mathcal{P}$ by the selection of $p$ and $q$.
\end{proof}

As aforementioned, the enumeration of maximal cliques and the enumeration of maximal independent sets are equivalent (up to a linear factor in the running time).  This observation holds in general.
\begin{proposition}\label{lem:complement-and-subclass}
  Let $\mathcal{P}$ be a nontrivial hereditary graph class. If  $\mathcal{P}$ has the \textsc{cks} property, then so are the complement class of $\mathcal{P}$.
\end{proposition}

The forbidden induced subgraph of the class of edgeless graphs is a single edge, or $1$-path.  The graph class forbidding $2$-path is cluster graphs, which, as we will see, also has the \textsc{cks} property.  However, for $\ell \ge 3$, the class of $\ell$-path-free graphs no longer has it.  See Figure~\ref{fig:p4-free} for an example on $3$- and $4$-path.
For $s \ge 0$, an \emph{($s$-)star} is a complete bipartite graph of which one part has one vertex and the other $s$ vertices; here we treat a single vertex as a (degenerated) star.
We can view $1$- and $2$-path as  $1$- and $2$-star, respectively.  Again for $s\ge 3$, the class of $s$-star-free graphs does not have the \textsc{cks} property; see Figure~\ref{fig:claw-free}.
We say that a graph is a \emph{star forest} if every component of the graph is a star.  It turns out that star forests play a crucial role in characterizing graph classes with the \textsc{cks} property.

\begin{figure}[h!]
  \centering
  \begin{subfigure}[b]{.43\linewidth}
    \centering
    \begin{tikzpicture}
      \draw [white, rounded corners] (2,-1) rectangle (4, .5);
      \begin{scope}[every node/.style={filled vertex}]
        \node["$v$"] (v) at (0, 1.5) {};
        \foreach[count=\i] \p in {-3, ..., 0, 1, 3} {
          \node(x\i) at (\p, 0) {};
          \node(y\i) at (\p, -0.5) {};
          \draw (v) -- (x\i) -- (y\i);              
        }
      \end{scope}
      \node at (2, -.25) {$\cdots$};
    \end{tikzpicture}
    \caption{}
    \label{fig:p4-free}
  \end{subfigure}
  \qquad
  \begin{subfigure}[b]{.43\linewidth}
    \centering
    \begin{tikzpicture}
      \draw [gray, rounded corners, fill=gray!20] (-3.5,-1) rectangle (3.74, .5);
      \begin{scope}[every node/.style={filled vertex}]
        \node["$v$"] (v) at (2, 1.5) {};
        \node["$u$"] (u) at (0, 1.5) {};
        \draw (v) -- (u);
        \foreach[count=\i] \p in {-3, ..., 0, 1, 3} {
          \node(x\i) at (\p, 0) {};
          \node(y\i) at (\p+.25, -0.5) {};
          \draw [blue, thick,dashed] (x\i) edge (y\i);
          \draw (x\i) -- (u) -- (y\i);              
        }
      \end{scope}
      \node at (2, -.25) {$\cdots$};
    \end{tikzpicture}
    \caption{}
    \label{fig:claw-free}
  \end{subfigure}  
  \caption{Examples showing that the $3$-path-free graphs, $4$-path-free graphs, and claw-free graphs do not have the \textsc{cks} property.  (a) A tree $G$ on $2k + 1$ vertices.  The subgraph $G- v$ does not contain any 3-path, but $G$ has $2^k$ maximal induced $3$-path-free subgraphs, and  $k\cdot 2^{k - 1} + 1$ maximal induced $4$-path-free subgraphs.  (b) A graph $G$ on $2 k + 2$ vertices.  The $2k$ vertices in the gray box form $k$ pairs, and two vertices in it are adjacent if and only if they are \emph{not} in the same pair.  The vertex $u$ is universal and is the only neighbor of $v$.  Both $G-v$ and $G - u$ are claw-free, but $G$ has other $2^k$ maximal induced claw-free subgraphs.
  } 
  \label{}
\end{figure}
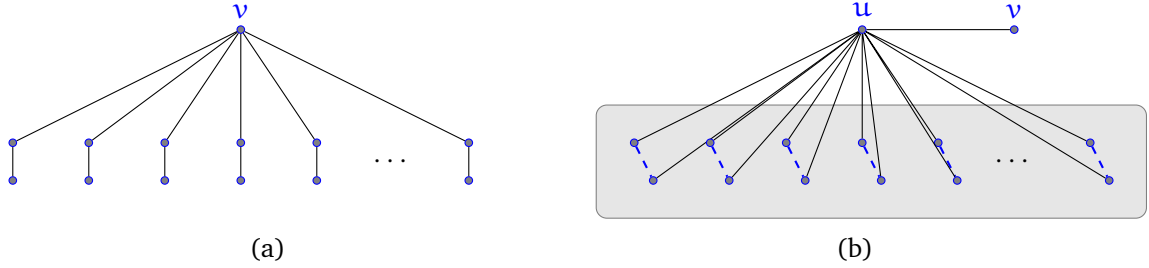
  
Before presenting the proof of Lemma~\ref{lem:no-claw-forest}, which is inspired by Lewis and Yannakakis~\cite{lewis-80-node-deletion-np}, let us consider a special and simple case, where every graph in $\cal F$ is biconnected and has at least three vertices.  Let $F$ be a graph in $\cal F$ with the least number of vertices, and let $s = |V(F)|$; note that $s\ge 3$.  By the selection of $F$, if every biconnected component (maximal biconnected subgraph) of a graph contains less than $s$ vertices, then the graph is $\cal F$-free.  Let $t$ be a positive integer.  We take $t$ disjoint copies of $F$, and identify one vertex from each copy; denoted by $v$ this identified vertex, and $G$ the resulting graph, which has $n = (s- 1) t + 1$ vertices.  Then $G - v$ and $G - \{v_1, \ldots, v_t\}$, where $v_i, 1 \le i\le t$, is a vertex in the $i$th copy of $F$ different from $v$, are maximal induced subgraphs of $G$ that are $\cal F$-free.  They are $\cal F$-free because each biconnected component of them has less than $s$ vertices, and they are maximal because adding any vertex back would introduce an $F$.
Therefore, $G$ has $(s- 1)^t + 1 =( s-1)^{(n-1)/(s - 1)} + 1$ solutions, which is exponential on $n$ because $s\ge 3$.

  For the general case, the main difficulty of \cite{lewis-80-node-deletion-np} lies in handling disjoint copies of forbidden induced subgraphs, while for us it is in forbidden stars.  Our definitions below are mostly borrowed from \cite{lewis-80-node-deletion-np}, but they have small deviations from theirs, with the excuse of simplicity.  Let $H$ be an arbitrary graph.  For each component $C$ of $H$, we define a sequence as follows.  If $C$ is not biconnected, then for each cutvertex $x$ of $C$, we can sort the components of $C - x$ by their cardinalities, which  gives a sequence
\[
  \langle n_{1}, n_{2}, \ldots, n_{p}  \rangle, \text{ where } n_{1} \ge n_{2} \ge \cdots\ge n_{p}.
\]
We fix a cutvertex $x$ such that the sequence defined by $x$ is the lexicographically smallest among all cutvertexs of $C$; we call it the \emph{pivot} of $C$, and denote the sequence by $\alpha(C)$.  If $C$ is biconnected, then we let any vertex of $C$ be the pivot, and set $\alpha(C) = \langle\, |C| - 1\, \rangle$.  Finally, sort the components of $H$ such that $\alpha(C_{1}) \ge_L \alpha(C_{2}) \ge_L \cdots\ge_L \alpha(C_{k})$, where  $\ge_L$ denotes the lexicographical ordering, and let
\begin{equation}
  \label{eq:beta}
  \beta(H) := \langle \alpha(C_{1}), \alpha(C_{2}), \ldots, \alpha(C_{k})  \rangle.  
\end{equation}
The first number in $\alpha(C_{1})$ is called the \emph{major index} of $H$.  Star forests are precisely the graphs with major index $0$ (only when they are independent sets) or $1$.
\begin{proposition}\label{lem:beta-subgraph}
  Let $H$ be a graph and $H'$ a proper induced subgraph of $H$.  Then $\beta(H') <_R \beta(H)$.
\end{proposition}
\begin{proof}
  Each component $C'$ of $H'$ is a subset of some component $C$ of $H$, and $\alpha(C) \ge_L \alpha(C')$.  They are equal only when $C = C'$, but they cannot be all equal because $H'$ is a proper subgraph of $H$. 
\end{proof}

\begin{proof}[Proof of Lemma~\ref{lem:no-claw-forest}]
  We prove the contrapositive of the statement.  We may assume without loss of generality that $\cal F$ does not contain any star forest; otherwise it suffices to consider the complement graph class by Proposition~\ref{lem:complement-and-subclass}.  This implies, in particular, that $\cal F$ does not contain any edgeless graph; or equivalently, all edgeless graphs are $\cal F$-free.  

  The sequence of $\beta$ defined above determines a total ordering $R$ among all graphs (though different graphs may have the same sequence).  Take $F$ to be a smallest graph in $\cal F$ with respect to this ordering.   By Proposition~\ref{lem:beta-subgraph}, any graph $H$ with $\beta(H)<_R\beta(F)$ is $\cal F$-free.
  Let $F_1$ be the first component of $F$ in the definition (1), $v$ its pivot, and $C_1$ a component of $F_1 - v$ with the maximum cardinality.  Let $t$ be a positive integer.  We take the disjoint union of $F - V(C_1)$ and $t$ disjoint copies of $F[C_1\cup \{v\}]$, and identify the $t + 1$ copies of the pivot of $F_1$ into a single vertex; we denote the vertex by $v$ and the component containing $v$ by $C$.
  This finishes the construction of the graph $G$.
  
  We first argue that $G - v$ is $\cal F$-free.  Since each component of $C - v$ is an induced subgraph of $F_1$, its $\alpha$-sequence is smaller than $\alpha(F_1)$.  Then $\beta(G - v) <_R \beta(F)$.

  We take a vertex from each of the $t$ copies of $C_1$, and let $X$ denote this set of vertices.  Note that the $\alpha$-sequence of each component of $C' = C - X$  is smaller than $\alpha(F_1)$: The number of components of $C' - v$ having size $|C_1|$ is one less than that of $F_1 - v$.  Therefore, $\beta(G - X) <_R \beta(F)$, which means that $G - X$ is $\cal F$-free.  It is maximal because adding any vertex $x\in X$ back to $G - X$ introduces a copy of $F$.
  Since $F$ is not a star forest, $|C_1|$, which is the major index of $F$, is larger than $1$.   We have seen $|C_1|^t$ maximal induced $\cal F$-free subgraphs of $G$, which is exponential on $n$.  Therefore, the class of $\cal F$-free graphs does not have the \textsc{cks} property.
\end{proof}
Unfortunately, this necessary condition given in Lemma~\ref{lem:no-claw-forest} is not sufficient to ensure a graph class to have the \textsc{cks} property, except for those graph classes defined by a single forbidden induced subgraph.

\begin{corollary}\label{cor:h-free}
  Let $F$ be a fixed graph.  The class of $F$-free graphs has the \textsc{cks} property if and only if $F$ or its complement is a path on at most three vertices.
\end{corollary}
\begin{proof}
  If $F$ is not a star forest, then the class of $F$-free graphs does not have the \textsc{cks} property because of Lemma~\ref{lem:no-claw-forest}.  Suppose then that $F$ is a star forest.   If $F$ has four or more vertices, then $\overline F$ contains a $3$- or $4$-cycle, hence not a star forest.  This is also true when $F$ consists three vertices and no edge.
  In either case, the class of $F$-free graphs does not have the \textsc{cks} property by Lemma~\ref{lem:no-claw-forest}.  Therefore, it remains to verify that the class of $F$-free graphs has the \textsc{cks} property when $F$ has at most three vertices and is not the complement of a triangle; it is easy to verify that either $F$ or $\overline F$  is a path on at most three vertices.
    
  The graph class is trivial when $F$ has only one vertex.
  If $F$ is $1$-path, then $F$-free graphs are precisely edgeless graphs, which we have discussed.
  The class of $2$-path-free graphs is cluster graphs.
  Let $G$ be a graph and $v\in V(G)$ such that $G - v$ is a maximal induced cluster subgraph of $G$.   In any other maximal cluster subgraph of $G$, either $v$ is with vertices from one clique of $G - v$, or it forms a clique by itself.  If $v$ is with clique $K$, then the subgraph has to be $G - ((K\setminus N(v)) \cup (N(v)\setminus K))$; otherwise, it is $G - N(v)$.  The total number of maximal induced cluster subgraphs is at most $n + 1$, the number of components of $G - v$ plus two.  Therefore, the class of cluster graphs has the \textsc{cks} property and this concludes the proof.
\end{proof}

According to Corollary~\ref{cor:h-free} and Proposition~\ref{lem:complement-and-subclass}, the following graph classes have the \textsc{cks} property because they all forbid the complement of $2$-path, and hence are subclasses of the complement of cluster graphs.  Forbidden induced subgraphs of these graph classes are listed in the appendix.

\begin{corollary}
  The following graph classes have the \textsc{cks} property: complete bipartite graphs, complete $p$-partite graphs for any positive integer $p$, and complete split graphs.   
\end{corollary}

We are not able to fully characterize graph classes that have the \textsc{cks} property.  Our investigation suggests that only few of those classes that
satisfy the necessary condition given in Lemma~\ref{lem:no-claw-forest} 
do have the \textsc{cks} property.  For graph classes defined by two forbidden induced subgraphs, one start forest $F$  and the other the complement of a star forest $F'$, here are some of those graph classes that have: (1) $F$ being $2 K_2$ and $F'$ consisting of three isolated vertices; (2) both $F$ and $F'$ being $2 K_2$; (3) $F$ being any star and $F'$ consisting of isolated vertices; (4) $F$ and $F'$ being any star; and (5) $F$ consisting of a star and an isolated vertex and $F'$ being a star.  To make the situation worse, it is possible that a class $\mathcal{P}$ of graphs has the \textsc{cks} property but a proper subclass of $\mathcal{P}$ does not.  For example, the class defined by (6) $F$ consisting of a star and an isolated vertex and $F'$ consisting of three isolated vertices is a subclass of (5), but it does not have the \textsc{cks} property.
    
The last part of this section is devoted to showing that the following graph classes have the \textsc{cks} property.  Note that independent sets are precisely graphs of maximum degree 0.  In general, for any fixed integer $d$, the forbidden induced subgraphs of maximum degree-$d$ graphs comprise all $(d + 2)$-vertex graphs that have a universal vertex, e.g., claw, paw, diamond, and $4$-clique when $d = 2$.  They include $({d+1})$-star, a star forest, and $({d+2})$-clique, the complement of a star forest.

\begin{proposition}\label{lem:trivial}
  The following graph classes have the \textsc{cks} property:  split graphs, pseudo-split graphs, threshold graphs,
  and graphs of maximum degree $d$ for any fixed integer $d$.
\end{proposition}
\begin{proof}
  We first show that the class of split graphs has the \textsc{cks} property.  The vertex set of a split graph can be partitioned into a clique $C$ and an independent set $I$, which may not be unique.  Given a graph $G$ such that $G-v$ is a maximal induced split subgraph of $G$ for some $v\in V(G)$, we show that there are polynomial number of maximal induced split subgraphs of $G$. Let $C\uplus I$ be a partition of $G-v$ such that $C$ is a maximal clique in $G-v$. 
  By the maximality of $G - v$, vertex $v$ is adjacent to some vertex in $I$ and nonadjacent to some vertex in $C$; in other words, $C\cup \{v\}$ is not a clique and $I\cup \{v\}$ is not an independent set.  For each maximal induced split subgraph $G[U]$ of $G$, we consider the clique in a split partition $C'\uplus I'$ of $G[U]$, where $C'$ is a maximal clique of $G[U]$.
  We separate the discussion into four cases, based on whether $C'$ contains $v$ and whether $C'$ intersects $I$.
\begin{itemize}
\item Case 1, $C' = C$.  Then $I'$ is either $I$ or $(I\setminus N(v))\cup \{v\}$.
\item Case 2, $C' = (N(v)\cap C)\cup\{v\}$.  Then $I'$ is either $I$ or $I''\cup \{u\}$ for some $u\in C\setminus C'$ and $I''= I\setminus N(u)$.  There are at most $n$ such maximal split subgraphs of $G$.
\item  Case 3, $C' = (N(v)\cap N(u)\cap C)\cup \{v, u\}$ for some vertex $u\in I\cap N(v)$.  Then $I'$ is either $I\setminus \{u\}$ or $I\setminus (N(w) \cup \{u\}) \cup \{w\}$ for some $w\in C\setminus C'$.  There are less than $n^2$ such maximal split subgraphs of $G$.
\item Case 4, $C' = (N(u)\cap C)\cup \{u\}$ for some vertex $u\in I$.
  Note that $v$ must be in $I'$ because $G - v$ is a maximal induced split subgraph.  As a consequence, $I'$ is disjoint from $I\cap N(v)$.  Then $I'$ is either $\{v,w\}\cup I\setminus N(\{w, v\})$ for some $w\in C\setminus (N(u)\cup N(v))$, or $\{v\}\cup (I\setminus \{u\}\setminus N(v))$.  There are less than $n^2$ such maximal split subgraphs of $G$.
\end{itemize}
In total, there are at most $O(n^2)$ maximal induced split subgraphs of $G$. 

Next, we show that the class of pseudo-split graphs has the \textsc{cks} property.  According to \cite{maffray-94-pseudo-split}, a graph is pseudo-split if and only if its vertex set can be partitioned into three (possibly empty) sets $C$, $S$, and $I$ such that
  (1) $C$ is a complete graph, I is independent and $S$ (if nonempty) induces a pentagon ($5$-cycle);
  (2)  every vertex in $C$ is adjacent to every vertex in $S$; and
  (3) no vertex in $I$ is adjacent to any vertex in $S$.
  If $S\ne\emptyset$, then the pseudo-split partition is unique, and $S$ is the only pentagon of the graph.

  Given a graph $G$ such that $G-v$ is a maximal induced pseudo-split subgraph of $G$ for some $v\in V(G)$, we show that there are polynomial number of maximal induced pseudo-split subgraphs of $G$.
  If there is no pentagon in $G$, then every induced pseudo-split subgraph of $G$ is an induced split subgraph of $G$.  Therefore, it reduces to the split graphs discussed above and we are done.  Henceforth we may assume $G$ contains a pentagon.
  Let $C\uplus S\uplus I$ be a partition of $G - v$, and we look for another maximal induced pseudo-split subgraph of $G$; let $C'\uplus S'\uplus I'$ be a partition of it.
By the maximality of $G - v$, vertex $v$ is adjacent to some vertex in $I\cup S$ and nonadjacent to some vertex in $C\cup S$. 
  
  \begin{itemize}
  \item Case 1, $S$ is empty.  There are only polynomial number of solutions with no pentagon (since it is same as split graphs).  We are hence focused on solutions with $S' \ne \emptyset$.  Clearly $S'$ must contain $v$, two vertices from $C$, and two vertices from $I$, because $C$ is a clique and $I$ is an independent set.  Therefore, $C' = C \cap N(S')$ and $I' = I\setminus N(S')$.
    There are at most ${|C| \choose 2} {|I| \choose 2} < \frac{n^4}{64}$ such maximal pseudo-split subgraphs of $G$.
  \item Case 2, $S$ is not empty.
  For each $x\in S$, the number of solutions not containing $x$ is bounded exactly the same as above: We are looking for pseudo-split subgraphs of $G - x$.  There are $O(n^4)$ such solutions in total.  In the rest we consider solutions containing all the five vertices in $S$.  Then $S$ has to be the pentagon in every such solution, i.e., $S' = S$, and $v$ is either in $C'$ or $I'$.  This cannot be arranged when $v$ is adjacent to $S$ but not all of them.  If $v$ is adjacent to every vertex in $S$, then the only solution is $C' = (C\cap N(v)) \cup \{v\}$ and $I' = I$.  Otherwise, $v$ is nonadjacnet to any vertex in $S$, then the only solution is $C'=C$ and $I' = (I\setminus N(v) ) \cup \{v\}$.
\end{itemize}
In total, there are at most $O(n^4)$ maximal induced pseudo-split subgraphs of $G$.

The proof for threshold graphs is very similar on that for split graphs, hence omitted.

Finally, we consider the class of graphs of maximum degree $d$ for any fixed integer $d$.  Suppose that the maximum degree of $G - v$ is at most $d$, while the maximum degree of $G$ is larger than $d$.  Let $G[U]$ be a maximal induced subgraph of $G$ with maximum degree at most $d$.  If $v\in U$, then $|U\cap N(v)| \le d$.  There are $n^d$ possible choices of them.  For each $u\in N(v)\cap U$, if the degree of $u$ is $d + 1$, then one of $N(u) \setminus \{v\}$ is not in $U$.  There are at most $d^2$ such vertices, and hence less than $d^2 \choose d$ choices.  Therefore, there are at most $O(d^{2d} n^d)$, which is polynomial on $n$ when $d$ is fixed, maximal induced subgraphs of $G$ of maximum degree $d$ in total. 
\end{proof}

The algorithms then follow from the result of Cohen et al.~\cite{cohen-08-all-maximal-induced-subgraphs}.
\begin{lemma}[\cite{cohen-08-all-maximal-induced-subgraphs}]
  \label{lem:cohen-poly-delay}
  Let $\mathcal{P}$ be a hereditary graph class.  The maximal (connected) induced $\mathcal{P}$ subgraphs problem can be solved with polynomial delay if the input-restricted version of the problem can be solved in polynomial time.
\end{lemma}

\paragraph{Acknowledgment.} The author would like to thank Mamadou Moustapha Kant{\'{e} for bringing Theorem 5.2 of \cite{eiter-95-hypergraph-transversals} (part of Corollary~\ref{thm:fintie}) to our attention, and for very helpful comments on an early version of the paper.

\appendix
\section{Appendix: Summary of graph classes and results}

The forbidden induced subgraphs of all graph classes studied in this paper are summarized in Table~\ref{fig:classes-containment}.  The graphs can be found in Figure~\ref{fig:small-graphs}, where by convention we use $P_\ell$ to denote ($\ell - 1$)-path.
For a comprehensive treatment and for references to the extensive literature on these graph classes, one may refer to the monograph of \cite{golumbic-2004-perfect-graphs}, the survey of \cite{brandstadt-99-graph-classes}, and its companion website \url{http://www.graphclasses.org/}.

\begin{figure}[h]
  \centering\footnotesize
  \begin{subfigure}[b]{0.11\linewidth}
    \centering
    \begin{tikzpicture}[every node/.style={filled vertex},scale=.5]
      \node (a) at (-1,0) {};
      \node (c) at (1,0) {};
      \node (b) at (-1,2) {};
      \draw (b) -- (a) -- (c);
    \end{tikzpicture}
    \caption{$P_3$}\label{fig:p3}    
  \end{subfigure}
  \,
  \begin{subfigure}[b]{0.11\linewidth}
    \centering
    \begin{tikzpicture}[every node/.style={filled vertex},scale=.5]
      \node (a) at (-1,0) {};
      \node (c) at (1,0) {};
      \node (b) at (-1,2) {};
      \draw (b) -- (a);
    \end{tikzpicture}
    \caption{$\overline{P_3}$}\label{fig:p3-complement}    
  \end{subfigure}
  \,
  \begin{subfigure}[b]{0.11\linewidth}
    \centering
    \begin{tikzpicture}[every node/.style={filled vertex},scale=.5]
      \node (a) at (-1,0) {};
      \node (c) at (1,0) {};
      \node (b) at (-1,2) {};
      \node (d) at (1,2) {};
      \draw (a) -- (b) (c) -- (d);
      \node[white] at (1.25,0) {};
      \node[white] at (-1.25,0) {};
    \end{tikzpicture}
    \caption{$2 K_2$}\label{fig:2k2}
  \end{subfigure}
  \,
  \begin{subfigure}[b]{0.11\linewidth}
    \centering
    \begin{tikzpicture}[every node/.style={filled vertex},scale=.5]
      \node (a) at (-1,0) {};
      \node (c) at (1,0) {};
      \node (b) at (-1,2) {};
      \node (d) at (1,2) {};
      \draw (d) -- (b) -- (a) -- (c);
      \node[white] at (1.25,0) {};
      \node[white] at (-1.25,0) {};
    \end{tikzpicture}
    \caption{$P_4$}\label{fig:p4}    
  \end{subfigure}
  \,
  \begin{subfigure}[b]{0.11\linewidth}
    \centering
    \begin{tikzpicture}[every node/.style={filled vertex}, scale=.5]
      \node (a) at (-1,0) {};
      \node (c) at (1,0) {};
      \node (b) at (-1,2) {};
      \node (d) at (1,2) {};
      \draw (a) -- (b) -- (d) -- (c) -- (a);
    \end{tikzpicture}
    \caption{$C_4$}\label{fig:c4}
  \end{subfigure}
  \,
  \begin{subfigure}[b]{0.11\linewidth}
    \centering
    \begin{tikzpicture}[every node/.style={filled vertex}, scale=.6]
      \node (a) at (18:1) {};
      \node (b) at (90:1) {};
      \node (c) at (162:1) {};
      \node (d) at (234:1) {};
      \node (e) at (-54:1) {};
      \draw (a) -- (b) -- (c) -- (d) -- (e) -- (a);
    \end{tikzpicture}
    \caption{$C_5$}\label{fig:c5}    
  \end{subfigure}
  \,
  \begin{subfigure}[b]{0.11\linewidth}
    \centering
    \begin{tikzpicture}[every node/.style={filled vertex}, scale=.8]
      \node (v1) at (0,1) {};
      \node (v2) at (0,0) {};
      \node (v3) at (1,1) {};
      \node (v4) at (1,0) {};
      \node (v5) at (0.5,1.5) {};
      \draw (v3) -- (v1) -- (v2) -- (v4) -- (v3) -- (v5) -- (v1);
    \end{tikzpicture}
    \caption{house}\label{fig:house}
  \end{subfigure}

  \begin{subfigure}[b]{0.15\linewidth}
    \centering
    \begin{tikzpicture}[scale=.25]
      \node [filled vertex] (a1) at (-2, 0) {};
      \node [filled vertex] (v) at (0, -2.5) {};
      \node [filled vertex] (b1) at (2, 0) {};
      \node [filled vertex] (c) at (0, 2.5) {};
      \node at (3, 0) {}; //position adjustment
      \node at (-3, 0) {}; //position adjustment
      \draw (b1) -- (c) -- (a1) -- (v) -- (b1) -- (a1);
    \end{tikzpicture}
    \caption{diamond}\label{fig:diamond}
  \end{subfigure}
  \,
  \begin{subfigure}[b]{0.11\linewidth}
    \centering
    \begin{tikzpicture}[scale=.25]
      \node [filled vertex] (a1) at (-3., 0) {};
      \node [filled vertex] (v) at (0, 0) {};
      \node [filled vertex] (b1) at (3., 0) {};
      \node [filled vertex] (c) at (0,3.5) {};
      \node at (-3.5, 0) {};
      \node at (3.5, 0) {};
      \draw[] (a1) -- (v) -- (b1);
      \draw[] (v) -- (c);
    \end{tikzpicture}
    \caption{claw}\label{fig:claw}
  \end{subfigure}
  \,
  \begin{subfigure}[b]{0.11\linewidth}
    \centering
    \begin{tikzpicture}[every node/.style={filled vertex},scale=.6]
      \node (a) at (-1,0.5) {};
      \node (b) at (-.5,0) {};
      \node (c) at (.5,0) {};
      \node (d) at (1,0.5) {};
      \node (e) at (0,2) {};
      \draw (e) -- (a) -- (b) -- (c) -- (d) -- (e);
      \draw (b) -- (e) -- (c);
    \end{tikzpicture}
    \caption{gem}\label{fig:gem}    
  \end{subfigure}
  \,
  \begin{subfigure}[b]{0.15\linewidth}
    \centering
    \begin{tikzpicture}[scale=.2]
      \node [filled vertex] (s) at (0,6) {};
      \node [filled vertex] (a) at (-5,0) {};
      \node [filled vertex] (a1) at (-2,0) {};
      \node [filled vertex] (b1) at (2,0) {};
      \node [filled vertex] (b) at (5,0) {};
      \node [filled vertex] (c) at (0,3.5) {};
      \draw[] (a) -- (a1) -- (b1) -- (b);
      \draw[] (c) -- (s);
      \draw[] (a1) -- (c) -- (b1);
    \end{tikzpicture}
    \caption{net}\label{fig:net}    
  \end{subfigure}
  \,
  \begin{subfigure}[b]{0.15\linewidth}
    \centering
    \begin{tikzpicture}[scale=.2]
      \node [filled vertex] (s) at (0,6) {};
      \node [filled vertex] (a) at (-4,0) {};
      \node [filled vertex] (a1) at (0, 0) {};
      \node [filled vertex] (b) at (4,0) {};
      \node [filled vertex] (c1) at (-2,3) {};
      \node [filled vertex] (c2) at (2,3) {};
      \draw[] (a) -- (a1) -- (b) -- (c2) -- (s) -- (c1) -- (a);
      \draw[] (c1) -- (c2) -- (a1) -- (c1);
    \end{tikzpicture}
    \caption{tent}\label{fig:tent}    
  \end{subfigure}
  \,    
  \begin{subfigure}[b]{0.13\linewidth}
    \centering
    \begin{tikzpicture}[scale=.2]
      \node [filled vertex] (v1) at (-3,0) {};
      \node [filled vertex] (v2) at (-3,-3) {};
      \node [filled vertex] (v3) at (3,-3) {};
      \node [filled vertex] (v4) at (3,0) {};
      \node [filled vertex] (v5) at (3,3) {};
      \node [filled vertex] (v6) at (-3,3) {};
      \draw (v1) -- (v2)  -- (v3) -- (v4) -- (v5) -- (v6) -- (v1) -- (v4);
    \end{tikzpicture}
    \caption{domino}\label{fig:domino}
  \end{subfigure}

  \begin{subfigure}[b]{0.2\linewidth}
    \centering
    \begin{tikzpicture}[scale=.25]
      \node [filled vertex] (a) at (-5, 0) {};
      \node [filled vertex] (b) at (-2.5, 0) {};
      \node [filled vertex] (c) at (0,0) {};
      \node [filled vertex] (d) at (2.5,0) {};
      \node [filled vertex] (e) at (5,0) {};
      \node [filled vertex] (v) at (0,2.5) {};
      \node [filled vertex] (u) at (0,5) {};
      \draw (a) -- (b) -- (c) -- (d) -- (e);
      \draw (u) -- (v) -- (c) ;
    \end{tikzpicture}
    \caption{long claw}\label{fig:f2}
  \end{subfigure}
  \begin{subfigure}[b]{0.2\linewidth}
    \centering
    \begin{tikzpicture}[scale=.2]
      \node [filled vertex] (s) at (0,2.8) {};
      \node [filled vertex] (a) at (-7, 0) {};
      \node [filled vertex] (a1) at (-4, 0) {};
      \node [filled vertex] (v) at (0, 0) {};
      \node [filled vertex] (b1) at (4, 0) {};
      \node [filled vertex] (b) at (7, 0) {};
      \node [filled vertex] (c) at (0,-3.5) {};
      \draw[] (a) -- (a1) -- (v) -- (b1) -- (b) -- (c) -- (a);
      \draw[] (b1) -- (c) -- (a1);
      \draw[] (v) -- (c) -- (s);
    \end{tikzpicture}
    \caption{whipping top}\label{fig:whipping-top}
  \end{subfigure}
  \begin{subfigure}[b]{0.23\linewidth}
    \centering
    \begin{tikzpicture}[scale=.34]
      \node [filled vertex] (s) at (0,3.8) {};
      \node [filled vertex] (a) at (-5, 0) {};
      \node [filled vertex] (a2) at (-3, 0) {};
      \node [filled vertex] (bi) at (0, 0) {};
      \node [filled vertex] (b2) at (3, 0) {};
      \node [filled vertex] (b) at (5, 0) {};
      \node [filled vertex] (c) at (0,2.2) {};
      \draw[] (a) -- (a2) (b2) -- (b);
      \draw[] (bi) -- (c) -- (s);
      \draw[] (a2) -- (c) -- (b2);
      \draw[dashed] (a2) -- (b2);
    \end{tikzpicture}
    \caption{\dag}\label{fig:dag}
  \end{subfigure}
  $\,$
  \begin{subfigure}[b]{0.23\linewidth}
    \centering
    \begin{tikzpicture}[scale=.34]
      \node [filled vertex] (s) at (0,3.8) {};
      \node [filled vertex] (a) at (-5, 0) {};
      \node [filled vertex] (a2) at (-3, 0) {};
      \node [filled vertex] (bi) at (0, 0) {};
      \node [filled vertex] (b2) at (3, 0) {};
      \node [filled vertex] (b) at (5, 0) {};
      \node [filled vertex] (c1) at (-1,2.2) {};
      \node [filled vertex] (c2) at (1,2.2) {};
      \draw[] (a) -- (a2) (b2) -- (b);
      \draw[] (bi) -- (c1) -- (s) -- (c2) -- (bi);
      \draw[] (a) -- (c1) -- (c2) -- (b);
      \draw[] (a2) -- (c1) -- (b2) -- (c2) -- (a2);
      \draw[dashed] (a2) -- (b2);
    \end{tikzpicture}
    \caption{\ddag}\label{fig:ddag}
  \end{subfigure}
  \caption{Forbidden induced graphs. }
  \label{fig:small-graphs}
\end{figure}
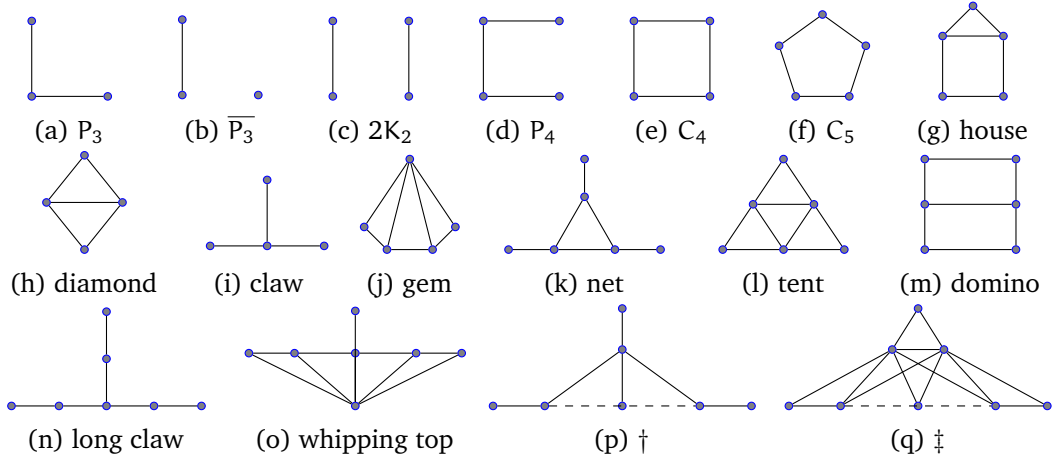

\begin{table}[h]
  \caption{Graph classes studied in this paper, their forbidden induced subgraphs and enumeration complexity.  In this table, $\ell \ge 4$.}
  \label{fig:classes-containment}
  \centering
  \begin{tabular}{l l | l}
    \toprule
    Graph class & forbidden induced subgraphs & enumeration
    \\
    \midrule
    {edgeless} & $P_2$
    \\
    {cluster} & $P_3$
    \\
    {complete $p$-partite} & $\overline{P_3}, K_{p + 1}$
    \\
    {maximum degree-$d$} & all $(d + 2)$-vertex graphs & \textsc{cks} property
    \\
    & with a universal vertex
    \\
    {threshold} & $2 K_2, C_4, P_4$
    \\
    {split} & $2 K_2, C_4, C_5$
    \\
    {complete split} & $\overline{P_3}, C_4$
    \\
    {pseudo-split} & $2 K_2, C_4$
    \\
    \midrule
    {acyclic} & $C_3, C_\ell$ & 
    \\
    {chordal} & $C_\ell$ & poly-delay
    \\
    {interval} & $C_\ell$, caws
    \\
    \midrule
    {unit interval} & $C_\ell$, claw, net, tent
    \\
    {block} & $C_\ell$, diamond & incremental poly 
    \\
    3-leaf power & $C_\ell$, bull, dart, gem \cite{dom-06-leaf-power}
    \\
    4-leaf power & $C_\ell$, $K_5-e$, tent, and other six \cite{brandstadt-08-4-leaf-powers}
    \\
    \bottomrule 
  \end{tabular}
\end{table}

\subsection{A direct proof of Theorem~\ref{thm:incp=totalp}}

\begin{figure}[h!]
  \tikz\path (0,0) node[draw=gray!50, text width=.9\textwidth, rectangle, rounded corners, inner xsep=20pt, inner ysep=10pt]{
    \begin{minipage}[t!]{\textwidth} \small
      Procedure $\textsc{next}(G, {\cal S}, A)$

      {\sc Input}: A graph $G$ with vertices $v_1$, $\ldots$, $v_n$, a collection $\cal S$ of solutions, and
      \\ \phantom{\sc Input:}
      a $p(n, N)$-time algorithm $A$ for the maximal induced $\mathcal{P}$ subgraphs problem.
      \\
      {\sc Output}: A solution not in $\cal S$, or ``completed'' if there is no further solution.

      \begin{tabbing}
        AAa\=AAa\=AAa\=AAa\=MMMMMMMMMMMMMAAAAAAAAAAAAAAAAAAAAAAAAA\=A \kill
        1.\> apply algorithm $A$ to $G$, aborted after $p(n, |{\cal S}| + 1)$ steps;
        \\
        2. \> {\bf if} it finishes {\bf then}
        \\
        2.1. \>\> {\bf if} all solutions found are in $\cal S$ {\bf then return} ``completed'';
        \\
        2.2. \>\> {\bf else return} a solution not in $\cal S$;
        \\
        3. \> $G_0 \leftarrow G$; 
        \\
        4. \> {\bf for each} $i\leftarrow 1, \ldots, n$ {\bf do}
        \\
        4.1. \>\> apply algorithm $A$ to $G_{i-1} - v_{i}$, aborted after $p(n, |{\cal S}| + 1)$ steps;
        \\
        4.2. \>\> {\bf if} it finishes {\bf then}
        \\
        4.2.1. \>\>\> {\bf if} a maximal $\mathcal{P}$ set $S'$ of $G_{i-1} - v_{i}$ is not a subset of any solution in $\cal S$ {\bf then}
        \\
        \>\>\>\> extend $S'$ to a solution $S$ of $G$ and {\bf return} $S$;
        \\
        4.2.2. \>\>\> $G_{i} \leftarrow G_{i-1}$;
        \\
        4.3. \>\> {\bf else} $G_{i} \leftarrow G_{i-1} - v_{i}$;
       \\
       5. \> apply algorithm $A$ to $G_n$; \qquad\qquad \comment{This time wait for it to finish.}
       \\
       6. \> find a maximal $\cal P$ set $S'$ of $G_n$ that is not a subset of any solution in $\cal S$;
       \\
       7. \> extend $S'$ to a solution $S$ of $G$ and {\bf return} $S$.
      \end{tabbing}  

    \end{minipage}
  };
\caption{The procedure for finding the next solution of the maximal induced $\mathcal{P}$ subgraphs problem.}
\label{fig:alg-incremental}
\end{figure}

We give here a direct proof of Theorem~\ref{thm:incp=totalp}.  The original proof of Bioch and Ibaraki~\cite{bioch-95-identification-dualization-positive-boolean-functions} was through.

\paragraph{Theorem~\ref{thm:incp=totalp} (restated).}
  For any hereditary graph class $\mathcal{P}$, the maximal induced $\mathcal{P}$ subgraphs problem can be solved in polynomial total time if and only if it can be solved in incremental polynomial time.
\begin{proof}
  The if direction is trivially true, and we now show the only if direction.
  Let $G$ be the input graph, and by a solution we mean  a maximal $\mathcal{P}$ set of $G$.
  Suppose that algorithm $A$ solves the maximal induced $\mathcal{P}$ subgraphs problem with at most $p(n, N)$ steps for some polynomial function $p$.  By Proposition~\ref{lem:enumeration-recognition}, there is a polynomial function $q$ such that we can decide in $q(n)$ time whether a graph on $n$ vertices is in $\mathcal{P}$.
We can thus use Proposition~\ref{lem:extension} to extend any $\mathcal{P}$ set of $G$ to a maximal $\mathcal{P}$ set in time $n q(n)$.
We repetitively call the procedure \textsc{next} described in Figure~\ref{fig:alg-incremental}, which makes calls, abortive or not, to $A$ to find the next solution of $G$, until the procedure returns ``completed.'' 

For the correctness of this procedure, we show that whatever the results of the calls to $A$, procedure \textsc{next} always returns a correct answer: a new solution of $G$ if $|{\cal S}| < N$ or ``completed'' otherwise.  If the call of $A$ in step~1 finishes, which means that it has found all solutions of $G$ , then it is clear that step 2 returns the correct answer.  This must happen when $|{\cal S}| = N$.
Therefore, when the procedure proceeds to step~4, there are more than $|\mathcal{S}|+1$ solutions of $G_0$, i.e., $G$.
Once a solution is returned by step~4.2.1, its correctness is ensured by Proposition~\ref{lem:extension}.
Otherwise, the procedure continues to step~5, and we show that in this case, there are more than $|\mathcal{S}|+1$ solutions of $G_i$ for all $i = 0, \ldots, n$.   Suppose that $p$ is the smallest number such $G_p$ does not have that number of solutions, then $G_p \ne G_{p-1}$, i.e., $G_p = G_{p-1} - v_p$, but this cannot happen because the procedure should have entered step~4.2 in the $p$th iteration.
  The call made in step~5 is guaranteed to find all the solutions of $G_n$.  By Proposition~\ref{lem:subgraph-cardinality}(ii), at least one solution of $G_n$ is not contained in any solution in $\cal S$.  This justifies step~6 and then step~7 always returns a correct solution, again, by Proposition~\ref{lem:extension}.

We now calculate the running time of the procedure, for which the focus is on step~5, because this is the only call of $A$ that is never aborted.  
Let $v_i$ be a vertex in $G_n$; note that the call of $A$ made in step~4.1 of the $i$th iteration on $G_{i - 1} - v_i$ (step~4.1) has finished.
If $G_{i - 1} - v_i$, an induced subgraph of $G$, has more than $|\cal S|$ maximal $\mathcal{P}$ sets, then by Proposition~\ref{lem:subgraph-cardinality}(ii), at least one of them is not a subset of any solution in $\mathcal{S}$, and the condition of step~4.2.1 is true, whereupon the procedure should have terminated before reaching step~5.
Therefore, the subgraph $G_{i - 1} - v_i$ has at most $|\cal S|$ maximal $\mathcal{P}$ sets.
By Proposition~\ref{lem:subgraph-cardinality}(i), $G_n - v_i$ has at most $|\cal S|$ maximal $\mathcal{P}$ sets because it is an induced subgraph of $G_{i - 1} - v_i$.
The fact that $G_n$ has more than $|\mathcal{S}| + 1$ maximal $\mathcal{P}$ sets, as we have seen above, implies that $G_n$ itself is not in $\mathcal{P}$, and hence each solution of $G_n$ is a proper subset of $V(G_n)$, which is hence also a solution of $G_n - v$ for some $v\in V(G_n)$.  Thus, the total number of solutions of $G_n$ is at most
$|V(G_n)| \cdot |{\cal S}|\le n |{\cal S}|$, which means that step~5 takes time $p(n, n |{\cal S}|)$.
Steps~1--3 take $p(n, |{\cal S}| + 1)$, $O(n^2)$, and $O(1)$ time respectively.  Step~4 takes $n \cdot ( p(n, |{\cal S}| + 1) + n^3\cdot q(n))$ time.  Steps~6 and 7 take $O(n^2\cdot |{\cal S}| + n\cdot q(n))$ time.  Putting them together, we can conclude that the running time of procedure \textsc{next} is polynomial on $n$ and $|\cal S|$, and  this completes the proof.
\end{proof}
\end{document}